\documentclass[
notitlepage 
,floatfix,
aps,
prx,
reprint,
superscriptaddress,
]{revtex4-2}
\usepackage[utf8]{inputenc}
\usepackage[most]{tcolorbox}
\usepackage[english]{babel}
\usepackage{graphicx}    %% Including graphics into  latex (any images)
\usepackage{amsmath,amssymb}
\usepackage{braket}
\usepackage{algorithm2e}
\usepackage{amsthm}
\usepackage{xcolor}

\newcommand{\smc}[0]{MConnect}
\newcommand{\pcm}[0]{PrepCM}

\definecolor{algorithmBackground}{RGB}{248,251,254}
\definecolor{algorithmTitle}{RGB}{224,238,249}
\definecolor{algorithmBorder}{RGB}{155,185,207}

\definecolor{algorithmBackground}{RGB}{248,249,250}
\definecolor{algorithmTitle}{RGB}{232,235,238}
\definecolor{algorithmBorder}{RGB}{185,190,195}

\newtcolorbox{algorithmBox}[1]{
  enhanced,
  colback=algorithmBackground,
  colbacktitle=algorithmTitle,
  colframe=algorithmBorder,
  coltitle=black,
  fonttitle=\normalfont\normalsize,
  title={#1},
  boxrule=0.4pt,
  arc=1pt,
  left=5pt,
  right=5pt,
  top=3pt,
  bottom=3pt,
  toptitle=1.5pt,
  bottomtitle=1.5pt,
  before skip=8pt,
  after skip=8pt
}

\usepackage[dvipsnames]{xcolor}

\DeclareMathOperator{\trace}{tr}

\theoremstyle{plain}
\newtheorem{definition}{Definition}

\newtheorem{result}{Result}
\newtheorem{lemma}{Lemma}

\newenvironment{restatement}[1]{\textbf{Restatement of #1}}

\begin{document}

\title{Quantum Circuit Fragments and Link Products in Continuous Variables}

\author{Amalina Lai}
\email{nura0089@e.ntu.edu.sg}
\affiliation{Nanyang Quantum Hub, School of Physical and Mathematical Sciences, Nanyang Technological University, Singapore}
\affiliation{Centre for Quantum  Technologies, Nanyang Technological University, Singapore}

\author{Graeme D. Berk}
\email{graemedean.berk@ntu.edu.sg}
\affiliation{Nanyang Quantum Hub, School of Physical and Mathematical Sciences, Nanyang Technological University, Singapore}
\affiliation{Centre for Quantum  Technologies, Nanyang Technological University, Singapore}

\author{Minjeong Song}
\affiliation{Centre for Quantum Technologies, National University of Singapore, Singapore}

\author{Mile Gu}
\email{gumile@ntu.edu.sg}
\affiliation{Nanyang Quantum Hub, School of Physical and Mathematical Sciences, Nanyang Technological University, Singapore}
\affiliation{Centre for Quantum  Technologies, Nanyang Technological University, Singapore}
\affiliation{Centre for Quantum Technologies, National University of Singapore, Singapore}

\begin{abstract}
Quantum circuits are often drawn as complete processes, with fixed inputs and outputs. In many quantum-information tasks, however, the natural object is only a fragment of such a circuit: an unknown source of non-Markovian noise to be probed, a subroutine to be inserted into a larger algorithm, or an agent implementing an adaptive strategy. In finite dimensions, the link product provides a systematic means to analyze how such circuit fragments interact and compose. Here, we develop the corresponding framework for continuous-variable systems. We introduce continuous-variable circuit fragments and associated link products that stitch such fragments together into larger processes. We show that the formalism simplifies substantially in the Gaussian regime, where link products can be evaluated efficiently using covariance-matrix representations. We use the formalism to construct non-Markovian processes and adaptive agent-environment interactions from modular components. This extends the quantum-comb toolkit to continuous variables, providing systematic methods for adaptive sensing, non-Markovian noise mitigation, and higher-order quantum circuit design.
\end{abstract}

\maketitle

\setlength{\parskip}{0.6em}
\setlength{\parindent}{0pt}

Many quantum-information tasks do not neatly fit the picture of a single quantum process acting on an input state to produce an output to be measured. A probe of a system need not be a one-shot measurement: powerful strategies can be interactive, with an agent actively intervening, sending probes into a process, and adapting future actions depending on earlier outcomes~\cite{yang2019memory, yunlong2023uncertainty, liu2024efficient} (see Fig.~\ref{fig:general-cv-comb}). The dynamics being probed may also be non-Markovian, such that successive channel uses are correlated and the effect of a later intervention can depend on earlier interventions~\cite{diosi2014general, camasca2021memory, benedetti2014non, vasile2011quantifying}. Indeed, the input to a quantum algorithm is often not encoded in an input state, but rather in a small quantum circuit to be processed by inserting it as a subroutine within a larger circuit~\cite{chiribella2008transforming,chiribella2008quantum,thompson2018quantum,quintino2019probabilistic}. In all these settings, the relevant object is not a closed quantum channel, but an open circuit fragment; a central question is then how such fragments interact and compose into larger quantum processes.
\begin{figure}[t!]
\centering
\includegraphics[width=\linewidth]{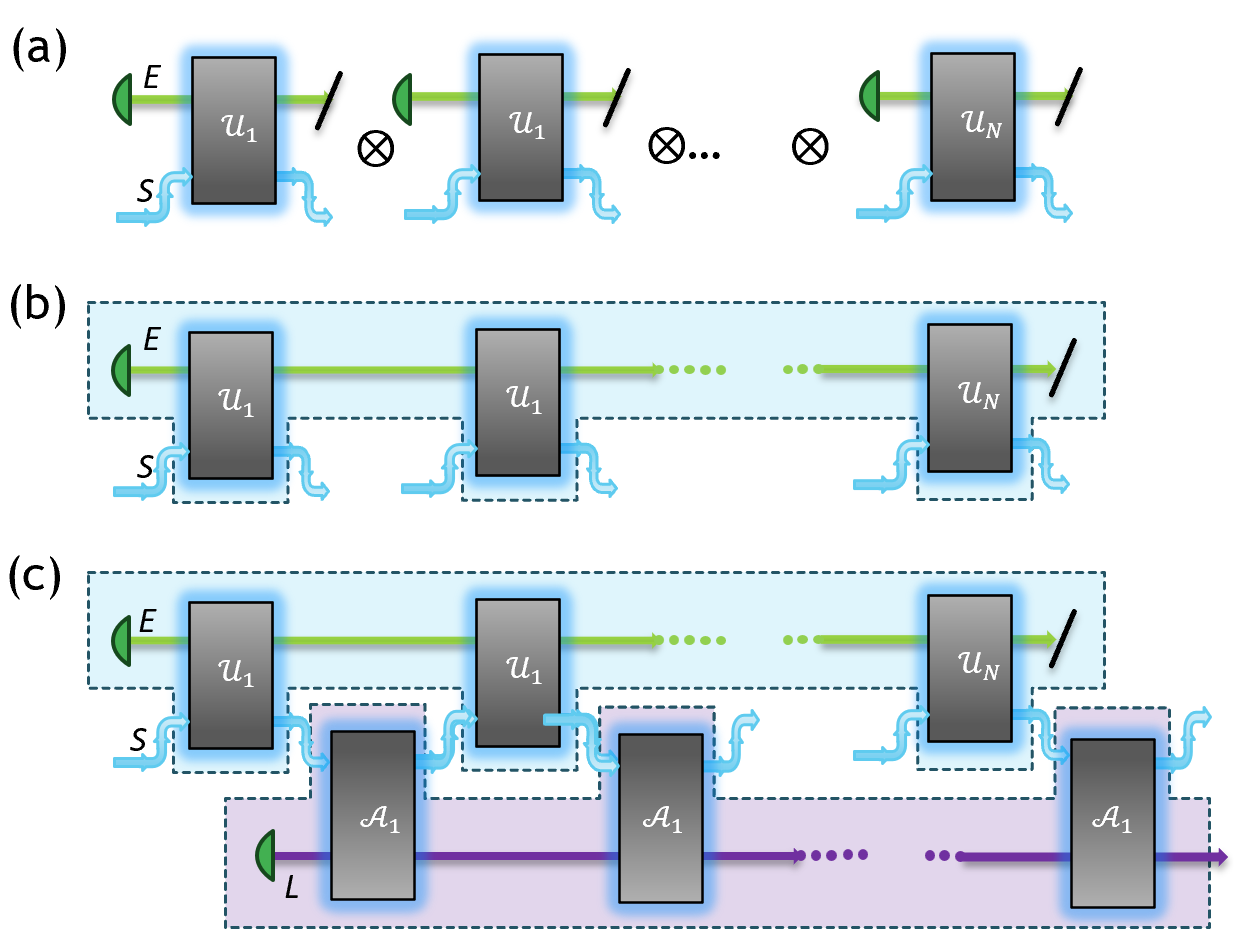}
\caption{\textbf{Quantum circuit fragments and link products}. In traditional settings, we often treat a quantum protocol as the application of some quantum process to a system $S$ - mathematically equivalent to a unitary interaction with an environment $E$ that is subsequently traced out (a). In many settings, however, the environment is non-Markovian, such that outputs at future times can depend on past inputs. Such objects are naturally treated as multi-time quantum circuit fragments with open ports (b). The quantum comb formalism provides a systematic means to describe such fragments and to connect them through the link product, as illustrated by the blue and violet fragments in (c). This allows us to describe interactive measurements in which an agent probes a process over multiple time steps. The same formalism is also central to higher-order quantum operations, where the inputs to a protocol may themselves be circuits or devices.}
\label{fig:general-cv-comb}
\end{figure}

In finite dimensions, quantum circuit fragments and their associated link products provide systematic tools for answering such questions~\cite{chiribella2008quantum,gutoski2007toward,chiribella2009theoretical,taranto2025higher}. In this picture, the input to an algorithm, the evolution of a non-Markovian quantum process, and the experimenter's probing strategy are all modular circuit fragments - intuitively understood as building blocks with open slots that can connect to a larger quantum circuit~\cite{ried2015quantum,xing2023fundamental}. An associated link product then formalizes how these circuit fragments interconnect~\cite{chiribella2009theoretical}. This toolkit enjoys remarkable success, underlying the use of quantum combs to optimize circuit architectures and describe adaptive metrology~\cite{chiribella2008quantum,chiribella2009theoretical,yang2019memory,kurdzialek2025quantum,liu2024efficient}, quantum circuit fragments for discerning causal structure~\cite{yunlong2023uncertainty, ried2015quantum}, higher-order quantum operations that transform one circuit to another~\cite{quintino2019probabilistic,taranto2025higher}, quantum strategies to play quantum games~\cite{gutoski2007toward}, the process tensor formalism for non-Markovian open systems~\cite{milz2017introduction,pollock2018non,milz2021quantum} and its applications in suppressing non-Markovian noise~\cite{milz2018reconstructing,white2022non,white2023filtering,tanggara2024strategic,kobayashi2024tensor}. 

Despite this vast diversity in application, such studies have so far focused on finite-dimensional quantum systems. Yet in many of these settings, continuous-variable information is also natural~\cite{braunstein2005continuous}. In sensing and communication, photonic systems are standard physical platforms~\cite{grosshans2003quantum,tan2008quantum,nair2020fundamental,PhysRevA.87.012107,PhysRevLett.127.010502} and multi-time measurements are gaining interest~\cite{PhysRevLett.126.230401,zhu2026twotoothbosonicquantumcomb}. Meanwhile, CV non-Markovianity has been experimentally observed using opto-mechanical systems~\cite{groeblacher2015observation,chen2011non}. 
Moreover, CV systems are particularly compelling in the Gaussian regime, where large-scale entanglement is both succinct to describe in theory, and relatively accessible in experiment~\cite{weedbrook_gaussian_2012}. These welcome properties propelled CV quantum information as a key platform for testing quantum protocols - underpinning landmark experiments in quantum teleportation, quantum key distribution, and generating entangled CV cluster states of record-breaking size~\cite{vaidman1994teleportation,furusawa1998unconditional,grosshans2003quantum,menicucci2006universal}. A formalism for CV quantum circuit fragments and link products could thus be of immediate relevance in a diverse range of applicative settings.

Here, we develop such a formalism. We introduce continuous-variable quantum circuit fragments and define a  corresponding CV link product. We then show that these objects admit a particularly efficient simplification in the Gaussian regime. There, Gaussian circuit fragments are represented by covariance matrices, and their link products can be evaluated directly by a covariance-matrix algorithm, avoiding explicit manipulation of infinite-dimensional density operators. We illustrate the framework through state-channel and channel-channel composition, and then use it to characterize Gaussian non-Markovian processes and complex agent-environmental interactions. This provides a route to importing the operational tools of quantum combs into continuous-variable and Gaussian quantum information.

\section{Preliminaries}

\textbf{Quantum Combs in Finite Dimensions}. We begin with a brief review of the quantum comb formalism for discrete variables~\cite{chiribella2008quantum,chiribella2009theoretical}. Consider a quantum channel $\mathcal {C} $ on a quantum system of dimension $d$. The Choi state $\Upsilon_{\mathcal C}$ is the quantum state resulting from applying $\mathcal C$ on one arm of a $d$-dimensional maximally entangled state $\ket{\Phi} = \frac{1}{\sqrt{d}}\sum_k \ket{k}\ket{k}$. $\Upsilon_{\mathcal C}$ is thus a bipartite state. Commonly, we label the partition that passes through $\mathcal {C}$ as the output $o$, while the other is labeled as input $i$. The rationale is that a projective measurement of $\Upsilon_{\mathcal C}$ that yields some state $\ket{\psi}$ collapses $o$ onto $\mathcal{C}(\ket{\psi}\bra{\psi})$, up to the usual transposition and normalization conventions.

The quantum comb formalism extends this idea to \emph{general quantum circuit fragments} - black-box objects with $N$ inputs and $M$ output wires that can be inserted within a quantum circuit. The corresponding Choi state $\Upsilon_{\mathcal C}$ is formed by feeding one arm of a maximally entangled state $\ket{\Phi}$ into each of the $N$ input wires of $\mathcal C$ (see Fig.~\ref{fig:circuit-fragment} for an example). This results in a multipartite quantum state  $\Upsilon_{\mathcal C}$ with $N + M$ partitions. Here, we assume all circuit fragments have a well-defined causal order. Like the Choi formalism for channels, $\Upsilon_{\mathcal C}$ captures all black-box properties of $\mathcal{C}$, such that if two circuit fragments have the same Choi state, there exists no experiment we can do on its inputs and outputs to distinguish one from the other. The quantum comb formalism thus converts the dynamical properties of a quantum process into the observable properties of its Choi state.

\begin{figure}[tp]
\centering
\includegraphics[width=1\linewidth]{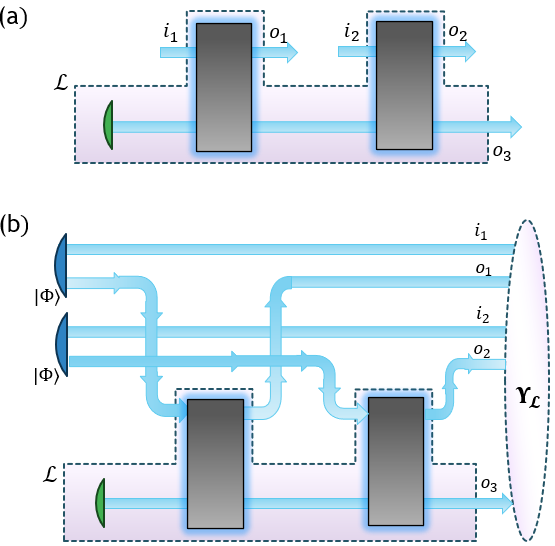}
\makeatother % justified caption
\caption{\textbf{Choi States of Circuit Fragments}. Consider a circuit fragment $\mathcal{L}$ with $N = 2$ input qudit wires and $M = 3$ output qudit wires (a). Its black-box properties can then be fully captured by the corresponding $N + M = 5$ qudit Choi state $\Upsilon_L$ prepared by the circuit in (b). Here, each blue disk segment denotes a maximally entangled qudit pair $\ket{\Phi}$, while the green disk segment denotes the initial state of an ancilla qudit that can interact with qudits input at $i_1$ and $i_2$.}
\label{fig:circuit-fragment}
\end{figure}

A key benefit of the Choi formalism is the existence of the link product, which provides a mathematical tool for understanding how different circuit fragments compose. Consider now two quantum circuit-fragments $\mathcal K$, $\mathcal L$. We can compose $\mathcal K$ and $\mathcal L$ together to form a larger circuit fragment $\mathcal{N}$ (for example, see Fig.~\ref{fig:general-cv-comb}c) by pairing some input wires from one fragment with the output wires of the other, or vice versa. The link product then provides a systematic means of determining the Choi state of $\mathcal{N}$ from the Choi states of $\mathcal K$ and $\mathcal L$. That is, we can define an operator $\star_J$ such that 
\begin{align}
\Upsilon_{\mathcal N} = \Upsilon_{\mathcal K} \star_J \Upsilon_{\mathcal L},
\label{eqn: dv-link-prod}
\end{align}
where $J$ defines the specific way in which we join the wires of $\mathcal K$ and $\mathcal L$. Here, we omit the mathematical details of this operation in finite-dimensional settings for conciseness as we will focus on the continuous setting~\footnote{For those interested, please see \cite{chiribella2008quantum}}. The key message here is that our goal is to build a comb formalism for infinite dimensions, which would naturally involve identifying a continuous variable link product.

\textbf{Continuous Variable Quantum Systems}. Here, we briefly review relevant background in continuous variable quantum information~\cite{braunstein2005continuous,adesso2014continuous}. Instead of qubits, the fundamental system in continuous variables is a quantum mode (qumode) - representing the state of a quantized harmonic oscillator. We set $\hbar = 1$, such that the position and momentum quadrature operators of each mode $k$ (denoted respectively by $\hat X_k$ and $\hat P_k$) satisfy the commutation relation $[\hat X_l, \hat P_k] = i \delta_{lk}$. We then denote the state-space of an $n$-mode system, (i.e., the space of density operators) by $\mathcal D(\mathcal H) =\mathcal D( \bigotimes_{k=1}^n \mathcal H^k)$ where $\mathcal H^k$ is the infinite-dimensional Hilbert space of the $k^{th}$ mode. 

Characteristic functions are a useful means of representing general $n$-qumode states. Let
$$\Omega = \bigoplus^n_{j=1} \begin{pmatrix}0 & 1 \\ -1 & 0\end{pmatrix}$$
be the symplectic form of phase space. Let $\hat R = (\hat X_1,\hat P_1,\dots, \hat X_n,\hat P_n)$, $\vec\alpha = (x_1,p_1,\dots x_n,p_n)^T \in \mathbb R^{2n}$, such that $\hat D(\vec \alpha) = \exp(i \hat R \Omega\vec\alpha)$ represents the Weyl (displacement) operator~\footnote{The Weyl operator displaces states in phase space, such that $\hat D(\vec\alpha)^\dagger \hat  X_j \hat D(\vec\alpha)
= \hat X_j+x_j$ and $\hat D(\vec\alpha)^\dagger \hat P_j \hat D(\vec\alpha) = \hat P_j+p_j$.}. The characteristic function 
\begin{align}
\chi_\rho(\vec \alpha) := \trace[\rho \hat D(\vec \alpha)]
\label{eqn: characteristic-function-general}
\end{align}
then provides an alternative and complete characterization of an $n$-mode quantum state $\rho$.

We will frequently use the partial transpose. Consider a state $\rho$ on $n$ modes, and let $J\subseteq \{1,\ldots,n\}$ be a subset of these modes. The partial transpose $\rho^{T_J}$ is the transpose on $\bigotimes_{j\in J}\mathcal H^j$ and the identity elsewhere. In our characteristic-function convention, partial transposition on mode $j$ acts as
\begin{align}
\chi_{\rho^{T_j}}(\dots,x_j,p_j,\dots)
=
\chi_{\rho}(\dots,-x_j,p_j,\dots),
\end{align}
with all other phase-space variables unchanged. Applying this transformation for all $j\in J$ gives $\chi_{\rho^{T_J}}$. 

When $\chi_\rho(\vec \alpha)$ is Gaussian, $\rho$ is a Gaussian state. Such states carry particular operational significance because their quadrature statistics are Gaussian and are particularly common in laboratory environments. Such states can then be entirely characterized by their quadrature expectation values $d_j = \braket{\hat R_j}$ and covariances $\Gamma_{ij}=\braket{\hat R_i \hat R_j + \hat R_j \hat R_i} - 2 \braket{\hat R_i}\braket{\hat R_j}$, whereby the resulting characteristic function has the form 
 \begin{align}
     \chi_{\rho_G}(\vec\alpha) = e^{-\frac{1}{4}\vec\alpha^T\Omega^T \Gamma \Omega\vec\alpha}e^{i\vec \alpha^T\Omega^T \vec{d}},
     \label{eqn: characteristic-func-gaussian-form}
 \end{align}
 where $\vec{d} = (d_1,d_2,\ldots d_{2n})^T$ is the vector of quadrature means and $\Gamma$ (with matrix elements $\Gamma_{ij}$) is the associated covariance matrix~\cite{adesso2014continuous}.

The two-mode squeezed vacuum (TMSV) state then provides a continuous-variable analog of Bell states in discrete-variable settings. Consider two qumodes $A$ and $B$. The idealized state of maximal entanglement $\sum_k \ket{k}_A\ket{k}_B$ is clearly unphysical as it cannot be normalized.  Instead, we can define a family of states $\ket{\Phi(r)} = \sum_{k=0}^\infty \tanh(r/2)^k /\cosh {(r/2)} \, \ket{k}_A\ket{k}_B$ whose entanglement grows with squeezing level $r$. Such states are Gaussian with $\vec d = \vec 0$, and covariances
\begin{align}
\Gamma_{\Phi} (r) = \begin{pmatrix}
\cosh r I_2 & \sinh r \sigma_z \\
\sinh r \sigma_z & \cosh r I_2
\end{pmatrix}
, \label{eqn: covmatrix-tmsv}
\end{align}
where each of the four block matrices in $\Gamma_{\Phi}$ is a $2\times 2$ square matrix, and $\sigma_z = \text{diag}(1,-1)$.  In the limit $r \rightarrow \infty$, the quadratures become perfectly correlated, and such states exhibit quadrature statistics $\braket{(\hat X_A - \hat X_B)^2} \rightarrow 0$ and $\braket{(\hat P_A + \hat P_B)^2} \rightarrow 0$. $r$ is referred to as the squeezing parameter, and any $r > 0$ implies entanglement.

\section{CV Quantum Combs}
We begin by extending the Choi states of circuit fragments and their corresponding link products to continuous variables. In analogy with finite-dimensional systems, we define a CV circuit fragment as any physically allowed quantum operation with designated input qumodes ($I$) and output qumodes ($O$) that map valid input states to valid output states. 

Formally, a continuous-variable analog of the Choi construction is obtained by feeding one half of the maximally entangled state \(\sum_{k=0}^{\infty}\ket{k}\ket{k}\) into each input qumode and retaining the other half as a reference. However, this state is nonphysical. We therefore instead use the TMSV state \(\Phi(r)\) with squeezing parameter \(r\) (as defined in preliminaries):

\begin{definition}[Choi state of a CV circuit fragment]
\label{definition: choi-state-cv-fragment}
Let \(\mathcal C\) be a CV circuit fragment with input qumodes \(I=\{i_1,\dots,i_m\}\) and output qumodes \(O=\{o_1,\dots,o_\ell\}\). Consider the following procedure:
\begin{enumerate}
\item[(a)] For each input \(i\in I\), prepare a TMSV state $\ket{\Phi(r_i)}$ of squeezing level $r_i$.
\item[(b)] Feed one arm of each TMSV state into the corresponding input wire of \(\mathcal C\).\end{enumerate}
Let $\mathbf{r} = (r_1, r_2, \ldots)$ denote the ordered list of squeezing levels used. We define the resulting state $\Upsilon_{\mathcal{C}}(\mathbf{r})$ as the \(\mathbf r\)-squeezed Choi state of \(\mathcal C\).
\end{definition}

Our definition reduces to the existing finite-squeezing Choi-state
constructions for CV quantum channels
\cite{fiurasek_gaussian_2002,giedke2002characterization,
holevo2010choi}. For each squeezing vector $\mathbf r$ with positive entries, the state, $\Upsilon_{\mathcal C}(\mathbf r)$ fully characterizes
$\mathcal C$. We denote the
full family of these states, parametrized by $\mathbf r$,
with $\Upsilon_{\mathcal C}$ and refer to them collectively as the
the Choi states of $\mathcal C$.

\textbf{The Continuous Variable Link Product}. Let \(\mathcal A\) and \(\mathcal B\) be two CV circuit fragments, with respective sets of input qumodes \(I_{\mathcal A}\) and \(I_{\mathcal B}\) and sets of output qumodes \(O_{\mathcal A}\) and \(O_{\mathcal B}\). We define a \emph{mode-stitching} \(j=(\underline{o}_j,\underline{i}_j)\) to represent either (i) feeding an output qumode \(\underline{o}_j\in O_{\mathcal A}\) of circuit fragment \(\mathcal A\) into an input qumode \(\underline{i}_j\in I_{\mathcal B}\) of circuit fragment \(\mathcal B\), or (ii) feeding an output qumode \(\underline{o}_j\in O_{\mathcal B}\) of circuit fragment \(\mathcal B\) into an input qumode \(\underline{i}_j\in I_{\mathcal A}\) of circuit fragment \(\mathcal A\) (see Fig.~\ref{fig:jcomp}). Here \(j\) labels the stitching pair: \(\underline{o}_j\) and
\(\underline{i}_j\) need not carry the same numerical labels as the
output and input modes in their respective circuit fragments. The underbar is used to emphasize this distinction.

\begin{figure}[tp]
    \centering
    \includegraphics[width=\linewidth]{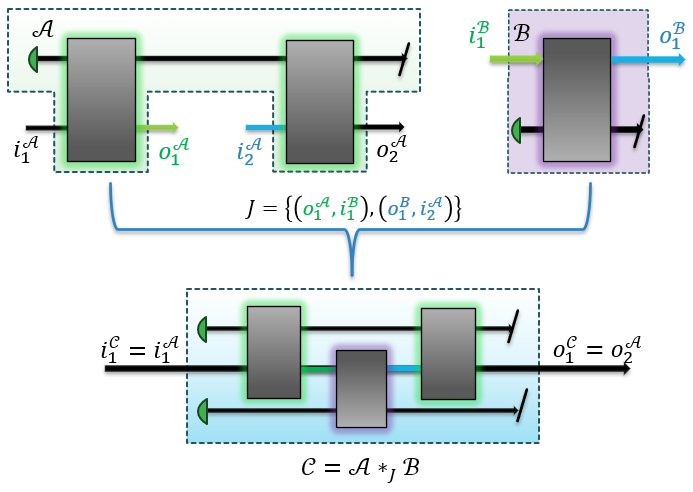}
    \caption{\textbf{$J$-stitching circuit fragments}. Two circuit fragments can be composed together to form larger circuit fragments via a $J$-stitching, where $J$ specifies exactly how the inputs and outputs on one fragment connect to the other. In the example above, circuit fragments $\mathcal{A}$ and $\mathcal{B}$ are composed to form $\mathcal{C} = \mathcal{A} \star_J \mathcal{B}$, where each pair of modes in $J$ share the same color. Formally, we can write $J = \{(o^\mathcal{A}_1, i^\mathcal{B}_1), (o^\mathcal{B}_1, i^\mathcal{A}_2)\}$, where each element denotes a pair of input/output modes to be connected. Note that this specific example is bidirectional: $J$ includes connections from  $\mathcal{A}$ to $\mathcal{B}$ and from $\mathcal{B}$ to $\mathcal{A}$.}
    \label{fig:jcomp}
\end{figure}

Let \(J=\{j_1,\ldots,j_m\}\) denote a finite set of such mode-stitchings,
where \(j_\ell=(\underline{o}_\ell,\underline{i}_\ell)\). When an
explicit ordering is needed, we use the displayed ordering
\(j_1,\ldots,j_m\). Applying all stitchings in \(J\) to connect \(\mathcal A\) and \(\mathcal B\) yields a composite circuit fragment \(\mathcal C\), whose input qumodes \(I_{\mathcal C}\) and output qumodes \(O_{\mathcal C}\) are precisely the unstitched input and output qumodes of \(\mathcal A\) and \(\mathcal B\). That is,
\begin{equation}
I_{\mathcal C} = \left(I_{\mathcal A}\cup I_{\mathcal B}\right)\setminus I_J,
\qquad
O_{\mathcal C} = \left(O_{\mathcal A}\cup O_{\mathcal B}\right)\setminus O_J,
\end{equation}
where \(I_J=\{\, \underline{i}_\ell : (\underline{o}_\ell,\underline{i}_\ell)\in J \,\}\) and \(O_J=\{\, \underline{o}_\ell : (\underline{o}_\ell,\underline{i}_\ell)\in J \,\}\) are the input and output qumodes that have been stitched together.

Note that \(J\) may be bidirectional - simultaneously including stitching from outputs of \(\mathcal A\) to inputs of \(\mathcal B\) and from outputs of \(\mathcal B\) to inputs of \(\mathcal A\) (indeed Fig.~\ref{fig:jcomp} illustrates such a scenario).  As such, we write \(J_{\mathcal A}\) and \(J_{\mathcal B}\) for the stitched modes
belonging to \(\mathcal A\) and \(\mathcal B\), respectively \footnote{Explicitly, $
J_{\mathcal A}=(I_{\mathcal A}\cup O_{\mathcal A})\cap (I_J\cup O_J)$ and $J_{\mathcal B}=(I_{\mathcal B}\cup O_{\mathcal B})\cap (I_J\cup O_J)$}. This motivates the following operational definition of the link product between \(\mathcal A\) and \(\mathcal B\).

\begin{definition}[CV link product]
\label{def:link-prod-CV} 
Let \(\mathcal C\) be the \(J\)-stitching of two CV circuit fragments \(\mathcal A\) and \(\mathcal B\). Let \(\Upsilon_{\mathcal A}\) , \(\Upsilon_{\mathcal B}\), and \(\Upsilon_{\mathcal C}\) be the respective Choi representations of the circuit fragments \(\mathcal A\), \(\mathcal B\), and \(\mathcal C\). We define the CV link product \(\star_J\) as the operation satisfying
\begin{align}
\Upsilon_{\mathcal C}
=
\Upsilon_{\mathcal A}\star_J \Upsilon_{\mathcal B}.
\end{align}
Where context is clear, we also denote $\mathcal{C}  =  \mathcal{A} \star_J \mathcal{B}$ to denote the $J$-stitching directly between two circuit fragments.
\end{definition}

To obtain a concrete expression for the link product, we first consider the case where all stitchings are directed from \(\mathcal A\) to \(\mathcal B\). That is, \(J_{\mathcal A}=O_J\) and \(J_{\mathcal B}=I_J\). In Appendix \ref{proof: link-prod-circuitfragments}, we show the following:

\begin{result}[Evaluating one-way $J$-stitchings]
\label{jstitch}
Let quantum circuit fragments $\mathcal A$ and $\mathcal B$
have Choi-states $\Upsilon_{\mathcal A}$ and
$\Upsilon_{\mathcal B}$, with squeezing vectors
$\mathbf r=(r_i)_{i\in I_{\mathcal A}}$ and
$\mathbf s=(s_i)_{i\in I_{\mathcal B}}$. Let $J$ be a
unidirectional stitching connecting the outputs
$J_{\mathcal A}=O_J$ of $\mathcal A$ to the inputs
$J_{\mathcal B}=I_J$ of $\mathcal B$. Let
$\mathcal C=\mathcal A\star_J\mathcal B$, with associated
squeezing vector $\mathbf w$. Writing
$\mathbf s_J=(s_i)_{i\in I_J}$ and
$\mathbf s_{\bar J}=(s_i)_{i\in I_{\mathcal B}\setminus I_J}$
for the stitched and unstitched components of $\mathbf s$,
respectively, so that
$\mathbf s=\mathbf s_J\oplus\mathbf s_{\bar J}$, we have
\begin{equation}
\label{eqn:link-prod-cv}
\Upsilon_{\mathcal C}(\mathbf w)
=
\lim_{\mathbf s_J\to\infty}
K(\mathbf s_J)\operatorname{tr}_J
\!\left[
\Upsilon_{\mathcal A}(\mathbf r)^{T_{J_{\mathcal A}}}
\Upsilon_{\mathcal B}
(\mathbf s_J\oplus\mathbf s_{\bar J})
\right],
\end{equation}
where $K(\mathbf s_J)=\prod_{i\in I_J}\cosh^2(s_i/2)$,
$\mathbf s_J\to\infty$ means $s_i\to\infty$ for every
$i\in I_J$, and
$\mathbf w=\mathbf r\oplus\mathbf s_{\bar J}$.
\end{result}

Here, the partial transpose \(T_{J_{\mathcal A}}\) is taken on the modes of \(\Upsilon_{\mathcal A}\) participating in the stitching,
and \(\operatorname{tr}_J\) denotes the partial trace over the
pairwise-identified stitched Hilbert spaces
\(\mathcal H^{\underline{o}_{\ell}}\cong\mathcal H^{\underline{i}_{\ell}}\) for all \(j_{\ell}=(\underline{o}_{\ell},\underline{i}_{\ell})\in J\) and $\oplus$ represents appending the entries of two vectors together. We can evaluate this expression iteratively. That is, we introduce the operation
\begin{equation}
\mathsf C_\ell[\Upsilon]
:=
\lim_{s_{\ell}\to\infty}
\cosh^2(s_{\ell}/2)\bra{\Omega_\ell}
\Upsilon(\mathbf{s})
\ket{\Omega_\ell}\nonumber
\end{equation}
where $\ket{\Omega_{\ell}} = \sum_n \ket{n}\ket{n}$ is the un-normalized maximally entangled state on qumodes $\underline{o}_{\ell}$ and $\underline{i}_{\ell}$, and $s_{\ell}$ is the squeezing parameter associated with $\underline{i}_{\ell}$. We can then write
\[
\Upsilon_{\mathcal A}
\star_J
\Upsilon_{\mathcal B}
=
\left(
\mathsf C_m
\circ
\cdots
\circ
\mathsf C_1
\right)
\!\left[
\Upsilon_{\mathcal A}
\otimes
\Upsilon_{\mathcal B}
\right].
\]
Each \(\mathsf C_{\ell}\) implements the contraction associated with the
mode stitching \(j_\ell=(\underline{o}_\ell,\underline{i}_\ell)\). For a
single stitching, this reduces to
\(
\mathsf C_\ell[
\Upsilon_{\mathcal A}\otimes
\Upsilon_{\mathcal B}]
=
\Upsilon_{\mathcal A}
\star_{\{j_\ell\}}
\Upsilon_{\mathcal B}
\). See Appendix~\ref{proof: link-prod-circuitfragments}.

Although Result~\ref{jstitch} is stated for unidirectional stitching from
\(\mathcal A\) to \(\mathcal B\), it is sufficient for evaluating any
valid bidirectional circuit composition with a well-defined temporal order. Indeed, each elementary
connection is still an output-to-input stitching. Provided no qumode is
stitched more than once, and the resulting circuit has a global causal
order, these elementary contractions act on distinct pairs of modes and
therefore commute. Hence, the bidirectional link product can be evaluated
by applying the one-way rule successively, in any convenient order.

\textbf{Describing Non-Markovianity}. The quantum comb formalism then provides an elegant means for describing non-Markovian quantum processes. Consider the circuit fragment $\mathcal T$ consisting of a system $S$ and environment $E$ that evolve unitarily with $\mathcal U_k$ at each time-step $t_k$. Since the environmental state is retained between time-steps, these settings can support information backflow and non-Markovian behavior (see Fig.~\ref{fig:choi-construct}).

Let $\mathcal{T}_{\mathcal{U}_k}$ be the circuit fragment consisting of a single quantum operation $\mathcal{U}_k$ on two qumodes - one we denote as the system, the other the environment. The circuit fragment thus has $2$ input qumodes and $2$ output qumodes (one input-output pair for the system and one for the environment). Let $\mathcal T_k$ denote the circuit fragment implementing
the operation $\mathcal U_k$, and let
$\Upsilon_{\mathcal T_k}$ denote its Choi states. The circuit fragment of two consecutive timesteps $\mathcal{T}_{k,k+1}$ will then have the Choi states $\Upsilon_{\mathcal{U}_{k}} \star_{J_k} \Upsilon_{\mathcal{U}_{k+1}}$, where $J_k = \{(E_o^{k},E_i^{k+1})\}$ is a single-mode link that connects the environmental output of $\Upsilon_{\mathcal{U}_{k}}$ with the input environment of $\Upsilon_{\mathcal{U}_{k+1}}$. Iterating in this manner from timestep $1$ to timestep $n$, setting in the environmental input of $\mathcal{T}_1$ to $\rho_E$ via \(J_0=\{(E_o^0,E_i^1)\}\) and discarding the final environmental mode, then results in a Choi states
\begin{figure}
    \centering
    \includegraphics[width=0.5\textwidth]{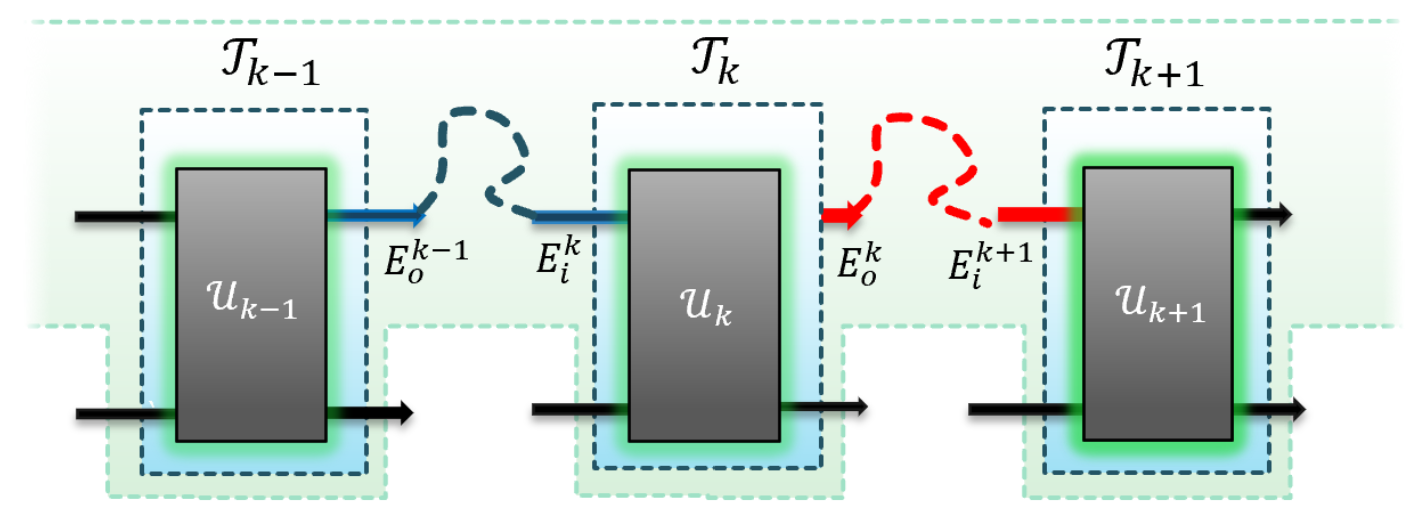}
    \caption{\textbf{Stitching a Non-Markovian Process}. We can construct the Choi representation of a Non-Markovian process by stitching a sequence of circuit fragments $\mathcal{T}_k$ in series, such that each two consecutive circuit fragments $\mathcal{T}_k$ and $\mathcal{T}_{k+1}$ are stitched together by the stitching $(E^k_o,E^{k+1}_i)$.}
    \label{fig:choi-construct}
\end{figure}

\begin{align}\label{eqn:nonmarkenv}
\Upsilon_{\mathcal T}
=
\operatorname{tr}_{E}
\!\left[
\rho_E
\star_{J_0}
\Upsilon_{\mathcal T_1}
\star_{J_1}
\cdots
\star_{J_{n-1}}
\Upsilon_{\mathcal T _n}
\right],
\end{align}
that aligns with the output of the associated quantum circuit in Fig.~\ref{fig:choi-construct}. The squeezing parameters associated with the system input modes survive as \(\mathbf r\), while those associated with the stitched environment input modes are taken to infinity in the corresponding link products. The resulting Choi states have $n$ input system qumodes that form the set $I_{\mathcal T}$, and $n$ output system qumodes that form $O_{\mathcal T}$; it is equivalent to a full tomographical characterization of the process over those times. 

\textbf{Interventions and Adaptive Measurements}. The most general means by which an agent interacts with an unknown process is through a sequence of adaptive interventions - formalized in the literature as multi-time adaptive measurements~\cite{yunlong2023uncertainty,yang2019memory}. We can represent such actions as circuit fragments, as depicted in Fig.~\ref{fig:general-cv-comb}c. Here, \(\Upsilon_{\mathcal T}\) denotes a possibly non-Markovian process, represented by the blue circuit fragment. The agent's sequence of interactions is then represented by the purple circuit fragment (\(\Upsilon_{\mathcal A}\)). This fragment can include an ancillary agent memory \(L\), allowing the agent to retain information from previous time steps that affects their choice of future interactions. This sequence of potentially adaptive interactions can be wholly described by a link product 
\begin{align}
\Upsilon_\text{combined} = \Upsilon_{\mathcal T} \star_J \Upsilon_\mathcal{A}.
\end{align}
Here, the set $J$ contains the pairs of output-input mode stitches between $\Upsilon_{\mathcal T}$ and $\Upsilon_\mathcal{A}$.

The resulting link product then outputs the bipartite state $\rho_{SL}$ on the system $S$, and the ancillary memory $L$. From this, the resulting agent measurement outcomes and probabilities can be easily computed by suitable projective operators on $S$ and $L$.

\section{Gaussian Link Product}
A key feature of CV quantum information is its convenience
in the Gaussian regime. Here we extend this convenience to
circuit fragments and link products. We say that a CV
circuit fragment $\mathcal A$ is Gaussian if all of its
Choi states are Gaussian. Operationally, any circuit
fragment built from Gaussian channels, Gaussian unitary
interactions, Gaussian ancilla states, and discarding operations
is Gaussian in this sense.

The Choi states of such Gaussian circuit fragments can thus
also be expressed in terms of their first two moments.
Consider a Gaussian circuit fragment \(\mathcal A\) with
$n$ input and output modes in total. For any positive
squeezing vector $\mathbf{r}$, the corresponding Choi state
is a Gaussian state on $n$ modes. It is therefore completely
described by an associated $2n\times2n$ covariance matrix
and a $2n$-dimensional displacement vector. This quadratic
scaling in $n$ compares very favorably to density-operator
descriptions, which, under any fixed local Fock-state
cutoff, scale exponentially with $n$. This motivates us to
define $\Gamma_{\mathcal A}$ as the family of covariance
matrices corresponding to the Choi states of $\mathcal A$
as $\mathbf r$ is varied. We refer to $\Gamma_{\mathcal A}$ as the Choi covariance of $\mathcal{A}$.

We can, of course, always evaluate the link product for two
Gaussian circuit fragments \(\mathcal A\) and \(\mathcal B\)
via their Choi states. However, working directly with
Choi covariances provides a
means to avoid dealing with density matrices of
exponentially growing size. In this context, our next result establishes efficient means to
to compute the Choi covariance of
$\mathcal A \star_{J} \mathcal B$ that scales efficiently
with the total number of modes:

\begin{result} We can compute the Gaussian link product of two Gaussian circuit fragments with complexity scaling as $O(m \bar{m}^2)$, where $m \geq 1$ is the number of qumode-pairs to be linked and $\bar{m} \geq m$ is the total number of qumodes used to describe the circuit fragments to be $J$-stitched. \end{result}

Further details and associated proofs, including how the resulting displacement
vectors can also be computed, are given in
Appendices~\ref{proof: gaussian-link-prod-int}--%
\ref{subsec: complexity-analysis}. Here, we proceed by describing a procedure for performing this computation, and leave the analysis of its complexity to Appendix~\ref{subsec: complexity-analysis}. Specifically, the procedure can be understood as an algorithm consisting of two subroutines, \texttt{\pcm} and \texttt{\smc}:

\begin{algorithm}[H]
\DontPrintSemicolon
\SetAlgoVlined
\SetAlgoNoEnd
\SetFuncSty{texttt}

\SetKwFunction{multiLinkProd}{multimodeLinkProduct}
\SetKwFunction{combMatrix}{\pcm}
\SetKwFunction{singleLinkProd}{\smc}
\SetKwProg{myalgo}{}{}{}

\myalgo{\multiLinkProd{$\Gamma_{\mathcal A},
\Gamma_{\mathcal B},J$}}{
  $\Gamma \gets \Gamma_{\mathcal A}\oplus\Gamma_{\mathcal B}$\;
  $\gamma_{\mathrm{tot}}
    \gets \combMatrix{$\Gamma,J$}$\;

  \For{$y_\ell\in J$}{
    $\gamma_{\mathrm{tot}}
      \gets
      \singleLinkProd{$\gamma_{\mathrm{tot}},y_\ell$}$\;
  }

  \KwRet $\gamma_{\mathrm{tot}}$\;
}
\end{algorithm}

\texttt{\pcm}, prepares the covariance matrix for all stitchings in $J$ at once, converting each pair $(\underline{o}_\ell,\underline{i}_\ell)$ into a corresponding effective stitched variable labeled by $y_\ell$. \texttt{\smc}, removes these effective stitched variables one at a time by Gaussian Schur complements. We proceed to describe each subroutine in detail.

\textbf{Covariance Preparation}. To define \texttt{\pcm}, we begin with the covariance matrix of the product Choi state, $\Gamma = \Gamma_{\mathcal A}\oplus\Gamma_{\mathcal B}$. \texttt{\pcm} first reorders the modes of $\Gamma$ so that all modes not eliminated by the link product appear first, followed by the $m$ modes of $J_A$ (i.e., both input and output modes of $\mathcal A$ participating in the join), then the $m$ modes of $J_B$, both in the reverse order of the pairs in $J$. Note that this reordering convention on $J$ is only a bookkeeping device: the final link product is independent of the order in which the stitched variables are eliminated \footnote{Since no qumode is stitched more than once, the elementary contractions act on disjoint mode pairs and hence commute.}. We use the reverse order so that $(y_1,y_2,\ldots,y_m)$ can be eliminated successively from the lower right.

Specifically, let $X$ denote all modes not eliminated, and let $C$ denote the collection of mode pairs to be stitched. Then, the updated $\Gamma$ has the form
\begin{align}
    \Gamma
    =
    \left(
    \begin{array}{c|c}
        \gamma_X & \gamma_{X,C} \\
        \hline
        \gamma_{X,C}^T & \gamma_C
    \end{array}
    \right),
    \label{eqn: cov-before-multilink-prep}
\end{align}
where $\gamma_C$ is the $4m \times 4m$ covariance block of the $2m$ qumodes that align with the mode pairs. The subroutine \texttt{\pcm}($\Gamma,J$) then returns

\begin{align}
    \texttt{\pcm}(\Gamma,J)
    =
    \left(
    \begin{array}{c|c}
        \gamma_X & \gamma_{X,C}Q_J \\
        \hline
        Q_J^T\gamma_{X,C}^T & Q_J^T\gamma_C Q_J
    \end{array}
    \right),
    \label{eqn: combined-cov-matrix-definition}
\end{align}

where $Q_J = (-\Lambda_m, I_{2m})^T$, $I_{2m}$ is the identity matrix of size $2m\times 2m$, and $\Lambda_m = \bigoplus_{\ell=1}^m \mathrm{diag}(1,-1)$. Here, $Q_J$ has dimensions $4m\times2m$, so the resulting lower-right block $Q_J^T\gamma_C Q_J$ has dimensions $2m\times2m$ and describes the $m$ effective stitched variables that are subsequently eliminated by \texttt{\smc}. If $J=\varnothing$, we use the convention $\texttt{\pcm}(\Gamma,\varnothing)=\Gamma$. After \texttt{\pcm}, the $m$ qumode pairs are converted to $m$ stitched variables, one for each $y_\ell \in J$, that are listed in reverse order.

Explicitly, consider any (including bidirectional) J-stitching between $\mathcal A$ and $\mathcal B$. In this setting, we can first write the Choi-covariances of $\mathcal A$ and $\mathcal B$ in the form
\begin{align}\nonumber
    \Gamma_{\mathcal A} = \left( \begin{array}{c|c}
        \gamma_{A \setminus J_A} & \gamma_\text{A,c} \\
        \hline
        \gamma_\text{A,c}^T & \gamma_{J_A}
    \end{array} \right) , \,  \Gamma_{\mathcal B} = \left(\begin{array}{c|c}
        \gamma_{B \setminus J_B} & \gamma_\text{B,c} \\
        \hline
        \gamma_\text{B,c}^T & \gamma_{J_B}
    \end{array} \right).
    % \label{eqn:link-prod-cov-mat-division}
\end{align}
Here $\gamma_{J_A}$ denotes the covariance block of the modes of $\mathcal A$ being stitched, while $\gamma_{J_B}$ denotes the covariance block of the corresponding modes of $\mathcal B$, with both ordered according to the elements of $J$ in reverse order. Eq.~\eqref{eqn: combined-cov-matrix-definition} then takes the form

\begin{align}\nonumber
    {\texttt{PrepCM($\Gamma,J$)}}
    =
    \left( \begin{array}{cc|c}
        \gamma_{A\setminus J_A} & \mathbf 0 & -\gamma_\text{A,c}{\Lambda_m} \\
        \mathbf 0 & \gamma_{B\setminus J_B} & \gamma_\text{B,c} \\
        \hline
        -{\Lambda_m}\gamma_\text{A,c}^T & \gamma_\text{B,c}^T & {\Lambda_m}\gamma_{J_A}\Lambda_m + \gamma_{J_B}
    \end{array} \right).
    \label{eqn: one-way-combined-cov-matrix-definition1}
\end{align}

\textbf{Completing the Stitch}. \texttt{\smc} then performs the stitching process. It takes a covariance matrix \(\gamma_\text{tot}\) and the \(y_\ell \in J\) that we wish to stitch. At the corresponding step, the effective variable associated with \(y_\ell\) is in the final \(2\times2\) block, so that
\[
    \gamma_\text{tot}
    =
    \left( \begin{array}{c|c}
        \gamma_{\bar X} & \gamma_{{\bar X},y_\ell} \\
        \hline
        \gamma_{{\bar X},y_\ell}^T & \gamma_{y_\ell}
    \end{array} \right).
\]

Here \(\gamma_{y_\ell}\) is the $2\times2$ covariance block associated with the effective stitched variable to be eliminated, while $\bar X$ denotes all remaining modes, including any other effective stitched variables that have not yet been eliminated. The subroutine then returns 

\[
    \texttt{\smc}(\gamma_\text{tot}, y_\ell) 
    =
    \lim_{r_{y_\ell}\to\infty}
    \left[
    \gamma_{\bar X} - \gamma_{{\bar X},y_\ell}(\gamma_{y_\ell})^{-1}\gamma_{{\bar X},y_\ell}^T
    \right],
\]
where \(r_{y_\ell}\) is the squeezing parameter associated with the Choi input mode $\underline{i}_{\ell}$ that is removed by the stitching.

\section{Illustrative Examples}
\label{section: examples}
We illustrate the above procedure in various situations, culminating in a Gaussian non-Markovian quantum environment under interactive measurements.

\textbf{Channels with Input}. We begin with a sanity check: a single-mode squeezer $\mathcal B$ acting on the vacuum state $\rho_\mathrm{vac}$ with covariance $\Gamma_{\mathcal A} = I_2$. While this problem is easily solvable using standard methods, its solution via link products helps illustrate some basic concepts.

Let $\mathcal{A}$ be a circuit fragment with no inputs and one output mode $o^A$ that always emits $\rho_\mathrm{vac}$. Its Choi state then is trivially $\rho_\mathrm{vac}$. Let input and output modes of $\mathcal{B}$ be $i^B$ and $o^B$, and let $\mathcal{B}$ be a $P$-quadrature squeezer that acts on covariance matrices as $\Gamma \mapsto F_s \Gamma F_s^T$ where $F_s=\cosh(s)I_2+\sinh(s)\sigma_z$ such that $s$ captures the magnitude of squeezing. We can then obtain the result by computing a $J$-stitching where $J = \{(o^A,i^B)\}$. First, note that the $r'$-squeezed Choi covariance of $\mathcal B$ is then
\begin{align}
\Gamma_\mathcal{B}(r') = \left( \begin{array}{c|c}
       \gamma_{i} & \gamma_{i,o} \\
       \hline
       \gamma_{i,o}^T & \gamma_{o}
    \end{array} \right),
\end{align}
where $\gamma_i = \cosh (r') I_2$, 
$\gamma_{i,o} = \sinh(r') \sigma_z F_s$ and    
$\gamma_{o} = \cosh(r')F^2_s$. To compute the Choi covariance of
$\mathcal A\star_J\mathcal B$, we apply \texttt{\pcm} to
$\Gamma_{\mathcal A}\oplus\Gamma_{\mathcal B}$,
which gives
\begin{align}
\gamma_\text{tot} = \left( \begin{array}{c|c}
        \gamma_{o} & \gamma_{i,o}^T \\
       \hline
        \gamma_{i,o} &  \sigma_z \Gamma_{\mathrm{A}} \sigma_z + \gamma_{i}
    \end{array} \right),
\end{align}
where each sub-block is a $2\times2$ matrix. \texttt{\smc} then gives the output 
\begin{align}
 \lim_{r'\rightarrow \infty} \left[\gamma_o - \gamma_{i,o}^T(\sigma_z \Gamma_\mathcal{A}\sigma_z+\gamma_i)^{-1} \gamma_{i,o}\right] = F^2_s,
\end{align}
aligning precisely with the squeezed vacuum state we would expect. The above equations also generalize readily to arbitrary Gaussian inputs and channels (see Appendix \ref{appendix: state-channel-join}), where the result agrees with prior literature~\cite{giedke2002characterization,fiurasek_gaussian_2002}. 

\textbf{Concatenating Channels}. Our second example illustrates the use of our algorithm to chain two single-qumode channels (Fig.~\ref{fig:link-prod-gen2}). We illustrate key concepts here via two identical single-mode squeezers of magnitude $s$, while Appendix \ref{appendix: general-chanchan-join} details the concatenation of two general Gaussian channels.

\begin{figure}[htp]
    \centering
    \includegraphics[width=\linewidth]{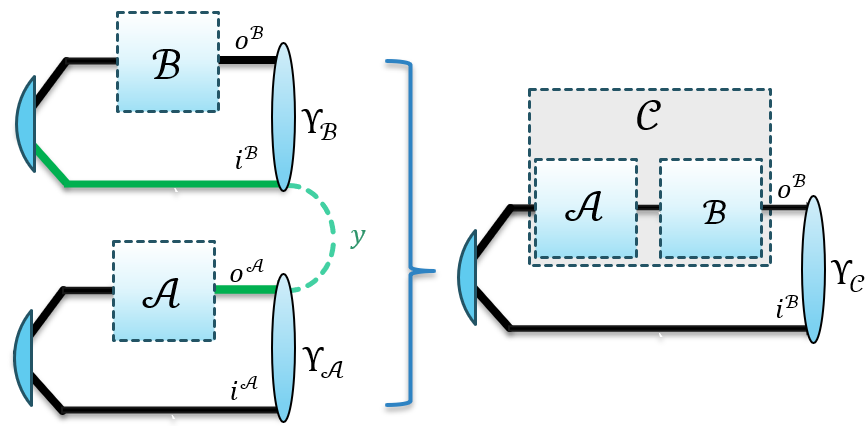}
    \caption{\textbf{Linking two quantum channels}. 
    $\Upsilon_{\mathcal A}(r)$ and $\Upsilon_{\mathcal B}(r')$ are the Choi states of two single-input, single-output circuit fragments, i.e., quantum channels. We can find the Choi states of their concatenation $\Upsilon_{\mathcal C}$ via the mode-stitching $y = (o^{\mathcal{A}},i^{\mathcal{B}})$.}
    \label{fig:link-prod-gen2}
\end{figure}

Let $\mathcal A$ and $\mathcal B$ both be a $P$-quadrature squeezer of identical magnitude $s$. Let $u_r = \cosh(r)$, and $v_r = \sinh(r)$. Then, the individual Choi states have the covariance matrix

\begin{align}
\Gamma_{\mathcal B}(r) = \Gamma_{\mathcal A}(r) = \left( \begin{array}{c|c}
       u_r F_{s}^2 & v_r F_{s}\sigma_z \\
       \hline
       v_r \sigma_z  F_{s} & u_r I_2
    \end{array} \right),
\end{align}

where $F_s$ is as before. Direct application of our algorithm, using $\Gamma_\mathcal A(r)$ and $\Gamma_\mathcal B(r')$, with the total number of stitched pairs, $m = 1$ gives
\begin{align}
    \Gamma_\text{out} &= (u_{r'} F_s^2 \oplus u_r I_2 ) - M^T ( u_{r'} I_2 +u_r F_s^2)^{-1} \nonumber M
    \label{eqn: link-prod-cm-channel}
 \end{align}
 where $M =  ( -v_{r'} F_s, v_r F_s\sigma_z)$. Taking $r'\to\infty$ (see Appendix \ref{appendix: general-chanchan-join} and \ref{appendix:squeezing-chanchan-join}), we obtain

\begin{align}
\Gamma_\text{out} = \begin{pmatrix}
 F_{2s}^2 \cosh r &  F_{2s} \sigma_z  \sinh r\\
 \sigma_z F_{2s} \sinh r   & I_2 \cosh r 
\end{pmatrix}.
\end{align}
This is the Choi covariance of a single squeezing channel with squeezing parameter $2s$, as expected.

\textbf{Modeling Non-Markovianity}. Our next illustration is that of building the Choi covariance of a non-Markovian process. Consider a system $S$ that interacts with a non-Markovian environment $E$ over two distinct time-steps (see Fig.~\ref{fig:process-example}), each with the same interaction process $\mathcal{A}$. Thus, the Choi states of the corresponding circuit fragment $\mathcal{T}$ satisfies
\begin{align}
\Upsilon_{\mathcal T}
=
\operatorname{tr}_{E}
\!\left[
\rho_0
\star_{J_0}
\Upsilon_{\mathcal A}
\star_{J_1}
\Upsilon_{\mathcal A}
\right],
\end{align}
where $\rho^{(E)}_0$ is the initial state of the environment, $\mathcal A$ is the 2-input, 2-output circuit fragment that describes a round of system-environmental interaction and the specific joins are depicted in Fig.~\ref{fig:process-example}.

\begin{figure}[ht]
    %\centering
    \includegraphics[width=0.95\linewidth]{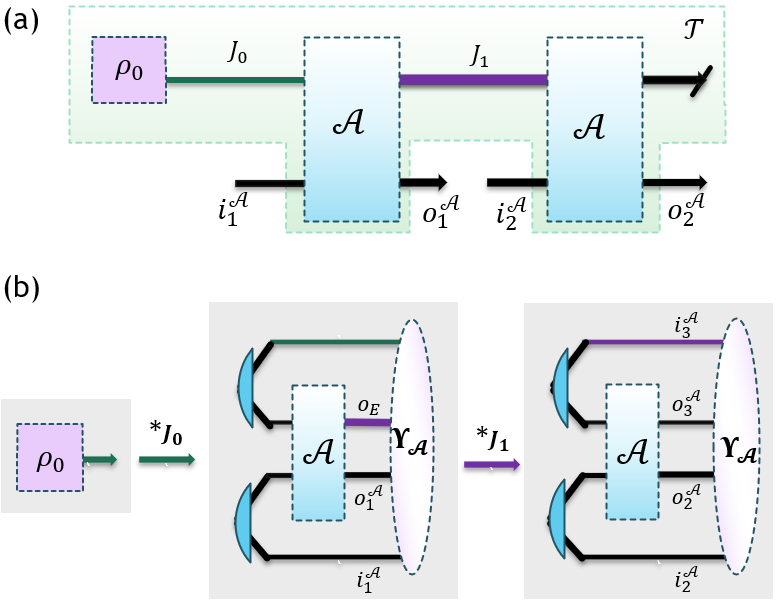}
    \caption{\textbf{A CV Non-Markovian Process}. (a) The quantum circuit fragment $\mathcal{T}$ depicts a non-Markovian process that takes an initial system state as input via mode $i^\mathcal{A}_1$ and describes how it can evolve through two rounds of interactions with some non-Markovian environment initially in state $\rho_0$. This allows for potential agent intervention between rounds. We can build a Choi covariance of $\mathcal{T}$ by taking its constituent circuit fragments and connecting them using mode links $J_0$ (green wire) and $J_1$ (purple wire). (b) depicts the CV Choi state of each circuit fragment: $\rho_0$ (left) is a single-qumode state, while the other two circuit fragments both modeling $\mathcal{A}$ are described by a Gaussian state on $4$ qumodes.}
    \label{fig:process-example}
\end{figure}

Here, we illustrate the scenario where the environment is initially the vacuum and $\mathcal A$ is a beamsplitter with transmittance $\cos^2\theta$. Setting shorthands $C=\cos\theta I_2$, $D=\sin\theta I_2$, $u_r = \cosh(r)$, $v_r = \sinh(r)$, we find in Appendix \ref{appendix:bs-bs-process-join} that the covariance matrix of $\Upsilon_{\mathcal A}$ is 
\begin{align}
    \Gamma_{\mathcal A} = \left(\begin{array}{ccc|c}
u_r I_2 & v_r\sigma_z C & -v_r \sigma_z D & \mathbf 0 \\
v_{r} C \sigma_z & \gamma_{22} & \gamma_{23} & v_{r'} D\sigma_z \\
-v_{r} D\sigma_z & \gamma_{23} & \gamma_{33} & v_{r'} C\sigma_z \\
\hline
\mathbf 0 & v_{r'} \sigma_z D & v_{r'} \sigma_z C & u_{r'}\mathbb I
\end{array}\right),
\end{align}
with $\gamma_{22} = u_r C^2 + u_{r'} D^2$, $\gamma_{23} = [u_{r'}-u_r] CD$, and $\gamma_{33} = u_r D^2 + u_{r'} C^2$, and the mode ordering of $\vec\alpha_{\mathcal A} = (\vec\alpha_{i_3^{\mathcal A}},\vec\alpha_{o_3^{\mathcal A}}, \vec\alpha_{o_2^{\mathcal A}}, \vec\alpha_{i_2^{\mathcal A}})^T$ (see Fig.~\ref{fig:process-example}). Meanwhile, the circuit fragment $\rho_0\star\Upsilon_{\mathcal A}$ has covariance
\begin{align}
    \Gamma_{\rho_0\star \mathcal A} 
    &= \left(\begin{array}{cc|c}
    u_r C^2 + D^2 & v_r C\sigma_z & (u_r  - 1)CD \\
    v_r \sigma_z C & u_r I_2 & v_r \sigma_z D \\
    \hline
    (u_r  - 1)CD & v_r D\sigma_z & u_r D^2 + C^2\\
\end{array}\right),
\end{align}
with the mode ordering of $\vec\alpha_{\rho\star\mathcal A} = (\vec\alpha_{o_1^{\mathcal A}},\vec\alpha_{i_1^{\mathcal A}}, \vec\alpha_{o_E})^T$ using analogous techniques to that of connecting a channel with input. 

To find the resulting Choi covariance, we apply the link product and trace out the environment mode. For expositional clarity and for compactness, we relegate all mathematical evaluation to Appendix \ref{appendix:bs-bs-process-join} and take the two surviving TMSV pairs to
have the same squeezing parameter $r$. Listing the resulting covariance in order $\vec\alpha=(\vec\alpha_{o_1^{\mathcal A}}, \vec\alpha_{i_1^{\mathcal A}}, \vec\alpha_{o_2^{\mathcal A}}, \vec\alpha_{i_2^{\mathcal A}})^T$, we obtain
\begin{align}\nonumber
\Gamma_\mathcal{T} = 
\begin{pmatrix}
\Gamma_{11} & \Gamma_{12} & \Gamma_{13} & \mathbf 0 \\
\Gamma_{12} & u_r I_2 & \Gamma_{23} & \mathbf 0 \\
\Gamma_{13} & \Gamma_{23} & \Gamma_{33} & v_r \cos\theta \sigma_z \\
\mathbf 0 & \mathbf 0 & v_r \cos\theta \sigma_z & u_r I_2
\end{pmatrix}.
\label{eqn:bs-bs-output}
\end{align}
Here, each block matrix is a $2\times 2$ matrix. $\Gamma_{11} = (u_r \cos^2\theta + \sin^2\theta) I_2$, $\Gamma_{12} = v_r \cos\theta \sigma_z$, $\Gamma_{13}=\sin^2\theta \cos\theta (1-u_r) I_2$, $\Gamma_{23}=-v_r \sin^2 \theta \sigma_z$ and $\Gamma_{33}
=
\left[
u_r\cos^2\theta
+
\sin^2\theta
\left(
u_r\sin^2\theta+\cos^2\theta
\right)
\right]I_2$. When $\theta=0$, $\mathcal A$ is the identity channel, and
the output covariance matrix reduces to that of two TMSV
states, as expected.

\begin{figure}
    \centering
    \includegraphics[width=0.5\textwidth]{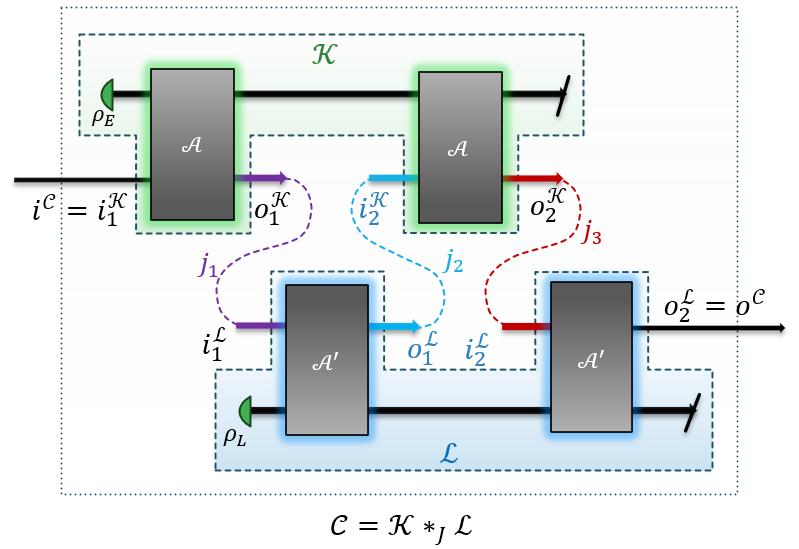}
    \caption{\textbf{Agent-Environmental Interactions over multiple timesteps} can be represented by two interacting quantum circuit fragments which we denote here by $\mathcal{K}$ (green) and $\mathcal{L}$ (blue). Setting $j_1 = (o^\mathcal{K}_1,i^\mathcal{L}_1)$, $j_2 = (o^\mathcal{L}_1,i^\mathcal{K}_2)$ and $j_3 =(o^\mathcal{K}_2,i^\mathcal{L}_2)$ and $J = \{j_1,j_2,j_3\}$, the link product $\mathcal{C} = \mathcal{K} \star_J \mathcal{L}$ then allows us to determine the overall effect of the resulting multi-time agent-environment interaction. The resulting circuit fragment $\mathcal{C}$ is then a quantum channel that maps input mode $i^\mathcal{K}_1$ to $o^{\mathcal{L}}_2$.}
    \label{fig:interaction-example-combs}
\end{figure}

\textbf{Agent-Environment Interactions}. The Gaussian link product also allows us to systematically deduce the result of agent-environment interactions. Here, we begin with the non-Markovian process $\mathcal{K}$ we just studied, consisting of an accessible system $S$ interacting with a non-Markovian environment $E$ (green circuit fragment in Fig.~\ref{fig:interaction-example-combs}). The blue circuit fragment $\mathcal{L}$ then depicts agent interventions on $S$ by application of some quantum process $\mathcal{A}'$ that also involves a memory system $L$. When the two are appropriately connected, they describe two rounds of agent-environmental interactions (see Fig.~\ref{fig:interaction-example-combs}). The link product then allows us to compute the overall effect of these interactions as described by the effective channel $\mathcal{C}$ on the system.  

Appendix \ref{appendix: two-combs-interact} outlines an explicit illustration of this computation using our methods when $\mathcal{A}$ and $\mathcal{A'}$ are each beamsplitters, and both the environment and the agent memory are initially the vacuum. In the special case where $\mathcal{A}$ and $\mathcal{A'}$ both have transmittance $\cos^2 \theta$, the resulting covariance matrix of $\mathcal{K} \star_J \mathcal{L}$ is
\begin{align}
    \Gamma_\mathcal{C}
    =
    \begin{pmatrix}
    \left[a_\theta^2 u_r+(1-a_\theta^2)\right] I_2
    &
    a_\theta v_r\sigma_z\\
    a_\theta v_r\sigma_z
    &
    u_r I_2
    \end{pmatrix},
\end{align}
where $a_\theta=\cos^4\theta-2\sin^2\theta\cos\theta$. This output represents the Choi-covariance of the effective channel $\mathcal{C}$ on the system. Thus, we see that the interaction between the agent and the non-Markovian process acts like an effective attenuating channel on the system with amplitude $a_\theta$.

\section{Discussion}
In a finite-dimensional setting, the quantum circuit fragments and their associated link product are indispensable tools for characterizing non-Markovian processes, agent-environment interactions, multi-time measurements, higher-order quantum operations, and optimizing circuit architectures. Here, we extended this compositional language to continuous-variable systems, whose elementary degrees of freedom are qumodes, such as quantized modes of light. For Gaussian processes, this formalism simplifies substantially: Gaussian circuit fragments admit compact covariance-matrix descriptions, and their link products can be evaluated directly at the covariance-matrix level, avoiding the explicit manipulation of infinite-dimensional density operators.

A natural continuation is to apply this toolkit in domains where finite-dimensional quantum combs have already demonstrated clear usefulness. Adaptive sensing and metrology in non-Markovian environments are compelling candidates~\cite{yang2019memory,kurdzialek2025quantum,liu2024efficient,altherr2021quantum,huang2024exact}, especially given the widespread use of Gaussian probes in CV settings due to their ease of generation and known near-optimality in certain nonadaptive regimes~\cite{nair2020fundamental,yadin2018operational}. Tomography and characterization of non-Markovian noise provide another natural direction~\cite{milz2018reconstructing,white2022non,white2025unifying}, with potential relevance to fault-tolerant CV quantum computation~\cite{menicucci2014fault,fukui2018high}. More fundamentally, the comb formalism enabled quantum uncertainty principles in spatiotemporal regimes~\cite{yunlong2023uncertainty}, and our techniques provide a pathway to discovering continuous-variable analogs. Meanwhile, there has been significant recent interest in transforming one black-box quantum operation into another~\cite{chiribella2008transforming,quintino2019probabilistic}, with the inversion of an unknown unitary serving as a key example~\cite{quintino2019reversing,yoshida2023reversing,mo2025parameterized}. It would certainly be interesting to understand how such a result extends to the CV or Gaussian regime, subject to uniquely CV constraints such as energy, squeezing, or non-classicality.
 
\vspace{0.1cm}
\noindent \emph{{Acknowledgments.}---} This work is supported by the Singapore Ministry of Education Tier 1 Grant RT4/23 and RG91/25, the National Research Foundation through the NRF Investigatorship on Quantum-Enhanced Agents (Grant No. NRF-NRFI09-0010), and the National Quantum Office, hosted in A*STAR, under its Advanced Quantum Algorithms and Solutions Funding Initiative (S25Q9DA001 and S25Q9DA002) and its Centre for Quantum Technologies Funding Initiative (S24Q2d0009).

\bibliography{references}

@article{camasca2021memory,
  title={Memory kernel and divisibility of Gaussian collisional models},
  author={Camasca, Rolando Ramirez and Landi, Gabriel T},
  journal={Physical Review A},
  volume={103},
  number={2},
  pages={022202},
  year={2021},
  publisher={APS}
}

@article{taranto2025higher, title={Higher-order quantum operations}, author={Taranto, Philip and Milz, Simon and Murao, Mio and Quintino, Marco T{\'u}lio and Modi, Kavan}, journal={arXiv preprint arXiv:2503.09693}, year={2025} }

@article{quintino2019probabilistic,
  title={Probabilistic exact universal quantum circuits for transforming unitary operations},
  author={Quintino, Marco T{\'u}lio and Dong, Qingxiuxiong and Shimbo, Atsushi and Soeda, Akihito and Murao, Mio},
  journal={Physical Review A},
  volume={100},
  number={6},
  pages={062339},
  year={2019},
  publisher={APS}
}

@article{fukui2018high,
  title={High-threshold fault-tolerant quantum computation with analog quantum error correction},
  author={Fukui, Kosuke and Tomita, Akihisa and Okamoto, Atsushi and Fujii, Keisuke},
  journal={Physical review X},
  volume={8},
  number={2},
  pages={021054},
  year={2018},
  publisher={APS}
}

@article{menicucci2014fault,
  title={Fault-tolerant measurement-based quantum computing with continuous-variable cluster states},
  author={Menicucci, Nicolas C},
  journal={Physical review letters},
  volume={112},
  number={12},
  pages={120504},
  year={2014},
  publisher={APS}
}

@article{xing2023fundamental,
  title={Fundamental limitations on communication over a quantum network},
  author={Xing, Junjing and Feng, Tianfeng and Fan, Zhaobing and Ma, Haitao and Bharti, Kishor and Koh, Dax Enshan and Xiao, Yunlong},
  journal={arXiv preprint arXiv:2306.04983},
  year={2023}
}

@article{kobayashi2024tensor,
  title={Tensor-network decoders for process tensor descriptions of non-Markovian noise},
  author={Kobayashi, Fumiyoshi and Manabe, Hidetaka and White, Gregory AL and Farrelly, Terry and Modi, Kavan and Stace, Thomas M},
  journal={arXiv preprint arXiv:2412.13739},
  year={2024}
}

@article{ried2015quantum,
  title={A quantum advantage for inferring causal structure},
  author={Ried, Katja and Agnew, Megan and Vermeyden, Lydia and Janzing, Dominik and Spekkens, Robert W and Resch, Kevin J},
  journal={Nature Physics},
  volume={11},
  number={5},
  pages={414--420},
  year={2015},
  publisher={Nature Publishing Group UK London}
}

@article{menicucci2006universal,
  title={Universal quantum computation with continuous-variable cluster states},
  author={Menicucci, Nicolas C and Van Loock, Peter and Gu, Mile and Weedbrook, Christian and Ralph, Timothy C and Nielsen, Michael A},
  journal={Physical review letters},
  volume={97},
  number={11},
  pages={110501},
  year={2006},
  publisher={APS}
}

@article{tan2008quantum,
  title={Quantum illumination with Gaussian states},
  author={Tan, Si-Hui and Erkmen, Baris I and Giovannetti, Vittorio and Guha, Saikat and Lloyd, Seth and Maccone, Lorenzo and Pirandola, Stefano and Shapiro, Jeffrey H},
  journal={Physical review letters},
  volume={101},
  number={25},
  pages={253601},
  year={2008},
  publisher={APS}
}

@article{liu2024efficient,
  title={Efficient tensor networks for control-enhanced quantum metrology},
  author={Liu, Qiushi and Yang, Yuxiang},
  journal={Quantum},
  volume={8},
  pages={1571},
  year={2024},
  publisher={Verein zur F{\"o}rderung des Open Access Publizierens in den Quantenwissenschaften}
}

@article{thompson2018quantum,
  title={Quantum plug n’play: modular computation in the quantum regime},
  author={Thompson, Jayne and Modi, Kavan and Vedral, Vlatko and Gu, Mile},
  journal={New Journal of Physics},
  volume={20},
  number={1},
  pages={013004},
  year={2018},
  publisher={IOP Publishing}
}

@article{chiribella2008transforming,
  title={Transforming quantum operations: Quantum supermaps},
  author={Chiribella, Giulio and D'Ariano, G Mauro and Perinotti, Paolo},
  journal={EPL (Europhysics Letters)},
  volume={83},
  number={3},
  pages={30004},
  year={2008}
}

@article{chiribella2008quantum,
  title={Quantum circuit architecture},
  author={Chiribella, Giulio and D’Ariano, G Mauro and Perinotti, Paolo},
  journal={Physical review letters},
  volume={101},
  number={6},
  pages={060401},
  year={2008},
  publisher={APS}
}

@article{vaidman1994teleportation,
  title={Teleportation of quantum states},
  author={Vaidman, Lev},
  journal={Physical Review A},
  volume={49},
  number={2},
  pages={1473},
  year={1994},
  publisher={APS}
}

@article{nair2020fundamental,
  title={Fundamental limits of quantum illumination},
  author={Nair, Ranjith and Gu, Mile},
  journal={Optica},
  volume={7},
  number={7},
  pages={771--774},
  year={2020},
  publisher={Optical Society of America}
}

@article{furusawa1998unconditional,
  title={Unconditional quantum teleportation},
  author={Furusawa, Akira and S{\o}rensen, Jens Lykke and Braunstein, Samuel L and Fuchs, Christopher A and Kimble, H Jeff and Polzik, Eugene S},
  journal={science},
  volume={282},
  number={5389},
  pages={706--709},
  year={1998},
  publisher={American Association for the Advancement of Science}
}

@article{yang2019memory,
  title={Memory effects in quantum metrology},
  author={Yang, Yuxiang},
  journal={Physical review letters},
  volume={123},
  number={11},
  pages={110501},
  year={2019},
  publisher={APS}
}

@article{huang2024exact,
  title={Exact quantum sensing limits for bosonic dephasing channels},
  author={Huang, Zixin and Lami, Ludovico and Wilde, Mark M},
  journal={PRX Quantum},
  volume={5},
  number={2},
  pages={020354},
  year={2024},
  publisher={APS}
}

@article{benedetti2014non,
  title={Non-Markovianity of colored noisy channels},
  author={Benedetti, Claudia and Paris, Matteo GA and Maniscalco, Sabrina},
  journal={Physical Review A},
  volume={89},
  number={1},
  pages={012114},
  year={2014},
  publisher={APS}
}

@article{yunlong2023uncertainty,
  title = {Quantum Uncertainty Principles for Measurements with Interventions},
  author = {Xiao, Yunlong and Yang, Yuxiang and Wang, Ximing and Liu, Qing and Gu, Mile},
  journal = {Phys. Rev. Lett.},
  volume = {130},
  issue = {24},
  pages = {240201},
  numpages = {7},
  year = {2023},
  month = {Jun},
  publisher = {American Physical Society},
  doi = {10.1103/PhysRevLett.130.240201},
  url = {https://link.aps.org/doi/10.1103/PhysRevLett.130.240201}
}

@article{vasile2011quantifying,
	title = {Quantifying non-{Markovianity} of continuous variable {Gaussian} dynamical maps},
	volume = {84},
	issn = {1050-2947, 1094-1622},
	number = {5},
	journal = {Physical Review A},
	author = {Vasile, Ruggero and Maniscalco, Sabrina and Paris, Matteo G. A. and Breuer, Heinz-Peter and Piilo, Jyrki},
	month = nov,
	year = {2011},
	pages = {052118},
}

@article{diosi2014general,
  title={General non-Markovian structure of Gaussian master and stochastic Schr{\"o}dinger equations},
  author={Di{\'o}si, Lajos and Ferialdi, Luca},
  journal={Physical review letters},
  volume={113},
  number={20},
  pages={200403},
  year={2014},
  publisher={APS}
}

@article{pollock2018non,
  title={Non-Markovian quantum processes: Complete framework and efficient characterization},
  author={Pollock, Felix A and Rodr{\'\i}guez-Rosario, C{\'e}sar and Frauenheim, Thomas and Paternostro, Mauro and Modi, Kavan},
  journal={Physical Review A},
  volume={97},
  number={1},
  pages={012127},
  year={2018},
  publisher={APS}
}

@article{milz2021quantum,
  title={Quantum stochastic processes and quantum non-Markovian phenomena},
  author={Milz, Simon and Modi, Kavan},
  journal={PRX Quantum},
  volume={2},
  number={3},
  pages={030201},
  year={2021},
  publisher={APS}
}

@article{milz2017introduction,
  title={An introduction to operational quantum dynamics},
  author={Milz, Simon and Pollock, Felix A and Modi, Kavan},
  journal={Open Systems \& Information Dynamics},
  volume={24},
  number={04},
  pages={1740016},
  year={2017},
  publisher={World Scientific}
}

@article{fiurasek_gaussian_2002,
    title = {Gaussian transformations and distillation of entangled {Gaussian} states},
    volume = {89},
    issn = {0031-9007, 1079-7114},
    url = {http://arxiv.org/abs/quant-ph/0204069},
    doi = {10.1103/PhysRevLett.89.137904},
    number = {13},
    urldate = {2023-12-29},
    journal = {Physical Review Letters},
    author = {Fiurasek, Jaromir},
    month = sep,
    year = {2002},
    note = {arXiv:quant-ph/0204069},
    pages = {137904},
}

@article{giedke2002characterization,
  title={Characterization of Gaussian operations and distillation of Gaussian states},
  author={Giedke, G{\'e}za and Cirac, J Ignacio},
  journal={Physical Review A},
  volume={66},
  number={3},
  pages={032316},
  year={2002},
  publisher={APS}
}

@article{chiribella2009theoretical,
  title={Theoretical framework for quantum networks},
  author={Chiribella, Giulio and D’Ariano, Giacomo Mauro and Perinotti, Paolo},
  journal={Physical Review A—Atomic, Molecular, and Optical Physics},
  volume={80},
  number={2},
  pages={022339},
  year={2009},
  publisher={APS}
}

@book{serafini2023quantum,
  title={Quantum continuous variables: a primer of theoretical methods},
  author={Serafini, Alessio},
  year={2023},
  publisher={CRC press}
}

@article{tanggara2024strategic,
  title={Strategic code: A unified spatio-temporal framework for quantum error-correction},
  author={Tanggara, Andrew and Gu, Mile and Bharti, Kishor},
  journal={arXiv preprint arXiv:2405.17567},
  year={2024}
}

@article{kurdzialek2025quantum,
  title={Quantum metrology using quantum combs and tensor network formalism},
  author={Kurdzia{\l}ek, Stanis{\l}aw and Dulian, Piotr and Majsak, Joanna and Chakraborty, Sagnik and Demkowicz-Dobrza{\'n}ski, Rafa{\l}},
  journal={New Journal of Physics},
  volume={27},
  number={1},
  pages={013019},
  year={2025},
  publisher={IOP Publishing}
}

@article{quintino2019reversing,
  title={Reversing unknown quantum transformations: Universal quantum circuit for inverting general unitary operations},
  author={Quintino, Marco T{\'u}lio and Dong, Qingxiuxiong and Shimbo, Atsushi and Soeda, Akihito and Murao, Mio},
  journal={Physical review letters},
  volume={123},
  number={21},
  pages={210502},
  year={2019},
  publisher={APS}
}

@article{yoshida2023reversing,
  title={Reversing unknown qubit-unitary operation, deterministically and exactly},
  author={Yoshida, Satoshi and Soeda, Akihito and Murao, Mio},
  journal={Physical Review Letters},
  volume={131},
  number={12},
  pages={120602},
  year={2023},
  publisher={APS}
}

@article{braunstein2005continuous,
title = {Quantum Information with Continuous Variables},
author = {Braunstein, Samuel L. and van Loock, Peter},
journal = {Reviews of Modern Physics},
volume = {77},
pages = {513--577},
year = {2005},
doi = {10.1103/RevModPhys.77.513}
}

@article{mo2025parameterized,
  title={Parameterized quantum comb and simpler circuits for reversing unknown qubit-unitary operations},
  author={Mo, Yin and Zhang, Lei and Chen, Yu-Ao and Liu, Yingjian and Lin, Tengxiang and Wang, Xin},
  journal={npj Quantum Information},
  volume={11},
  number={1},
  pages={32},
  year={2025},
  publisher={Nature Publishing Group UK London}
}

@article{altherr2021quantum,
  title={Quantum metrology for non-Markovian processes},
  author={Altherr, Anian and Yang, Yuxiang},
  journal={Physical Review Letters},
  volume={127},
  number={6},
  pages={060501},
  year={2021},
  publisher={APS}
}

@article{adesso2014continuous,
  title={Continuous variable quantum information: Gaussian states and beyond},
  author={Adesso, Gerardo and Ragy, Sammy and Lee, Antony R},
  journal={Open Systems \& Information Dynamics},
  volume={21},
  number={01n02},
  pages={1440001},
  year={2014},
  publisher={World Scientific}
}

@article{hu2021operator,
  title={Operator transpose within normal ordering and its applications for quantifying entanglement},
  author={Hu, Liyun and Zhang, Luping and Chen, Xiaoting and Ye, Wei and Guo, Qin and Fan, Hongyi},
  journal={Annalen der Physik},
  volume={533},
  number={6},
  pages={2000589},
  year={2021},
  publisher={Wiley Online Library}
}

@inproceedings{gutoski2007toward,
  title={Toward a general theory of quantum games},
  author={Gutoski, Gus and Watrous, John},
  booktitle={Proceedings of the thirty-ninth annual ACM symposium on Theory of computing},
  pages={565--574},
  year={2007}
}

@article{milz2018reconstructing,
  title={Reconstructing non-Markovian quantum dynamics with limited control},
  author={Milz, Simon and Pollock, Felix A and Modi, Kavan},
  journal={Physical Review A},
  volume={98},
  number={1},
  pages={012108},
  year={2018},
  publisher={APS}
}

@article{white2025unifying,
  title={Unifying Non-Markovian Characterization with an Efficient and Self-Consistent Framework},
  author={White, Gregory AL and Jurcevic, Petar and Hill, Charles D and Modi, Kavan},
  journal={Physical Review X},
  volume={15},
  number={2},
  pages={021047},
  year={2025},
  publisher={APS}
}

@article{groeblacher2015observation,
  title={Observation of non-Markovian micromechanical Brownian motion},
  author={Groeblacher, Simon and Trubarov, A and Prigge, N and Cole, GD and Aspelmeyer, M and Eisert, J},
  journal={Nature communications},
  volume={6},
  number={1},
  pages={7606},
  year={2015},
  publisher={Nature Publishing Group UK London}
}

@article{yadin2018operational,
  title={Operational resource theory of continuous-variable nonclassicality},
  author={Yadin, Benjamin and Binder, Felix C and Thompson, Jayne and Narasimhachar, Varun and Gu, Mile and Kim, MS},
  journal={Physical Review X},
  volume={8},
  number={4},
  pages={041038},
  year={2018},
  publisher={APS}
}

@article{chen2011non,
  title={Non-Markovian dynamics of a nanomechanical resonator measured by a quantum point contact},
  author={Chen, Po-Wen and Jian, Chung-Chin and Goan, Hsi-Sheng},
  journal={Physical Review B—Condensed Matter and Materials Physics},
  volume={83},
  number={11},
  pages={115439},
  year={2011},
  publisher={APS}
}

@article{grosshans2003quantum,
title = {Quantum Key Distribution Using Gaussian-Modulated Coherent States},
author = {Grosshans, Fr{'e}d{'e}ric and Van Assche, Gilles and Wenger, J{'e}r{^o}me and Brouri, Rosa and Cerf, Nicolas J. and Grangier, Philippe},
journal = {Nature},
volume = {421},
pages = {238--241},
year = {2003},
doi = {10.1038/nature01289}
}

@article{weedbrook_gaussian_2012,
    title = {Gaussian {Quantum} {Information}},
    volume = {84},
    issn = {0034-6861, 1539-0756},
    url = {http://arxiv.org/abs/1110.3234},
    doi = {10.1103/RevModPhys.84.621},
    language = {en},
    number = {2},
    urldate = {2022-09-06},
    journal = {Reviews of Modern Physics},
    author = {Weedbrook, Christian and Pirandola, Stefano and Garcia-Patron, Raul and Cerf, Nicolas J. and Ralph, Timothy C. and Shapiro, Jeffrey H. and Lloyd, Seth},
    month = may,
    year = {2012},
    note = {arXiv:1110.3234 [quant-ph]},
    pages = {621--669},
}

@article{white2023filtering,
  title={Filtering crosstalk from bath non-Markovianity via spacetime classical shadows},
  author={White, Gregory AL and Modi, Kavan and Hill, Charles D},
  journal={Physical Review Letters},
  volume={130},
  number={16},
  pages={160401},
  year={2023},
  publisher={APS}
}

@article{white2022non,
  title={Non-Markovian quantum process tomography},
  author={White, Gregory AL and Pollock, Felix A and Hollenberg, Lloyd CL and Modi, Kavan and Hill, Charles D},
  journal={PRX Quantum},
  volume={3},
  number={2},
  pages={020344},
  year={2022},
  publisher={APS}
}

@article{holevo2010choi,
  title={The Choi--Jamiolkowski forms of quantum Gaussian channels},
  author={Holevo, Alexander S},
  journal={Journal of Mathematical Physics},
  volume={52},
  number={4},
  year={2011},
  publisher={AIP Publishing}
}

@misc{zhu2026twotoothbosonicquantumcomb,
      title={Two-tooth bosonic quantum comb for temporal-correlation sensing}, 
      author={Shaojiang Zhu and Xinyuan You and Alexander Romanenko and Anna Grassellino},
      year={2026},
      eprint={2601.10916},
      archivePrefix={arXiv},
      primaryClass={quant-ph},
      url={https://arxiv.org/abs/2601.10916}, 
}

@article{PhysRevLett.127.010502,
  title = {Waveform Estimation from Approximate Quantum Nondemolition Measurements},
  author = {Boulebnane, Sami and Woods, Mischa P. and Renes, Joseph M.},
  journal = {Phys. Rev. Lett.},
  volume = {127},
  issue = {1},
  pages = {010502},
  numpages = {6},
  year = {2021},
  month = {Jun},
  publisher = {American Physical Society},
  doi = {10.1103/PhysRevLett.127.010502},
  url = {https://link.aps.org/doi/10.1103/PhysRevLett.127.010502}
}

@article{PhysRevLett.126.230401,
  title = {Experimental Demonstration of Instrument-Specific Quantum Memory Effects and Non-Markovian Process Recovery for Common-Cause Processes},
  author = {Guo, Yu and Taranto, Philip and Liu, Bi-Heng and Hu, Xiao-Min and Huang, Yun-Feng and Li, Chuan-Feng and Guo, Guang-Can},
  journal = {Phys. Rev. Lett.},
  volume = {126},
  issue = {23},
  pages = {230401},
  numpages = {6},
  year = {2021},
  month = {Jun},
  publisher = {American Physical Society},
  doi = {10.1103/PhysRevLett.126.230401},
  url = {https://link.aps.org/doi/10.1103/PhysRevLett.126.230401}
}

@article{PhysRevA.87.012107,
  title = {Optimal estimation of joint parameters in phase space},
  author = {Genoni, M. G. and Paris, M. G. A. and Adesso, G. and Nha, H. and Knight, P. L. and Kim, M. S.},
  journal = {Phys. Rev. A},
  volume = {87},
  issue = {1},
  pages = {012107},
  numpages = {7},
  year = {2013},
  month = {Jan},
  publisher = {American Physical Society},
  doi = {10.1103/PhysRevA.87.012107},
  url = {https://link.aps.org/doi/10.1103/PhysRevA.87.012107}
}

\begin{appendix}
\onecolumngrid

\section{Re-ordering qumodes}
\label{appendix: qumode-reordering}

Recall that we originally set $\hat R = (\hat X_1,\hat P_1,\dots, \hat X_n,\hat P_n)$ and $\vec\alpha = (x_1,p_1,\dots x_n,p_n)^T \in \mathbb R^{2n}$. In this notation, the qumodes are ordered in a specific manner according to the form of $\hat R$ and $\vec\alpha$. As a consequence, the covariance matrix of a Gaussian state will also be of a specific form, e.g., for a multimode Gaussian state whose modes are uncorrelated, it has a covariance matrix given by $\Gamma = \Gamma_1\oplus\Gamma_2\oplus \cdots \oplus \Gamma_n$ where $\Gamma_k$ corresponds to the covariance matrix of the $k$-th mode.

However, it may sometimes be of use to employ a different order of $\hat R$ and $\vec\alpha$. Thus, one can ``re-order" the qumodes (i.e., change the order of $\hat R$ and $\vec\alpha$) - this would be a purely notational operation, but this is useful for e.g., ensuring that \texttt{\pcm} outputs the correct matrix. Being able to re-order the qumodes also allows the algorithm \texttt{multimodeLinkProduct} to return an output that matches our notation convention.

In ``re-ordering" the qumodes, it is important to keep track of the new order of the modes. Here, we first provide a basic example of such re-ordering, before proceeding with a practical and illustrative example appropriate to the context of the paper.

The basic example: in standard channel convention, for the Choi state of a channel with input and output modes $i,o$, our qumodes might take the form of 

\begin{align}
\vec\alpha &= (\vec\alpha_o,\vec\alpha_i)^T = (x_o,p_o,x_i,p_i)^T
\end{align}

(i.e., it has the qumode order of ($o,i$)). If the Choi state is Gaussian, then the covariance matrix of the Choi state will have the submatrices

\begin{align}
\Gamma = \begin{pmatrix}
\gamma_{o} & \gamma_{o,i}\\
\gamma_{o,i}^T & \gamma_{i}
\end{pmatrix}.
\end{align}

If we re-order the qumodes such that the new order is given by ($i,o$)

\begin{align}
\vec\alpha' &= (\vec\alpha_i,\vec\alpha_o)^T = (x_i,p_i,x_o,p_o)^T
\end{align}

then the covariance matrix in this new order is

\begin{align}
\Gamma' = \begin{pmatrix}
\gamma_{i} & \gamma_{o,i}^T\\
\gamma_{o,i} & \gamma_{o}
\end{pmatrix}.
\end{align}

The same concept may be applied for a larger number of qumodes as well. Particularly in the context of this paper, we talk about how we re-order the qumodes such that the modes participating in the J-stitching form the last $2m$ terms of $\vec\alpha$, where $m = |J|$ is the number of qumode-stitched pairs in $J$. For bookkeeping in the algorithm, we have chosen for the qumodes to be arranged in descending order such that last two terms correspond to the $x-,$ and $p-$quadratures of the first ordered pair to be joined. Qumodes not participating in the join constitute the remaining first entries of $\vec\alpha$, such that the new order $\vec\alpha_{\mathcal A}' = (\vec\alpha_{A\setminus J}, \, x_{j=m},p_{j=m},\, \dots,\, x_{j=1},p_{j=1})^T$.

We provide an illustrative example of this here: Consider the process $\mathcal A$ in Fig.~\ref{fig:jcomp} with modes $A = I_\mathcal A \cup O_\mathcal A = \set{i_1^{\mathcal A},o_1^{\mathcal A}, i_2^{\mathcal A}, o_2^{\mathcal A}}$. This circuit is to be $J$-stitched with fragment $\mathcal B$ where $J=\set{(o_1^\mathcal A, i_1^\mathcal B),(o_1^\mathcal B, i_2^\mathcal A)}$. Therefore, the modes of $\mathcal A$ participating in the join are  $J_\mathcal A = \set{o_1^\mathcal A, i_2^\mathcal A}$.

One might naively order the qumodes as $(i_1^{\mathcal A}, o_1^{\mathcal A}, i_2^{\mathcal A}, o_2^{\mathcal A})$, i.e., set $\vec\alpha_A = (x_{i_1^{\mathcal A}},p_{i_1^{\mathcal A}},x_{o_1^{\mathcal A}},p_{o_1^{\mathcal A}}, x_{i_2^{\mathcal A}},p_{i_2^{\mathcal A}},x_{o_2^{\mathcal A}},p_{o_2^{\mathcal A}})^T$. However, in accordance to our requirement for the algorithm, then we may choose to arrange the new qumode order of $(i_1^{\mathcal A}, o_2^{\mathcal A}, i_2^{\mathcal A}, o_1^{\mathcal A})$ such that:

$$
\vec\alpha_A' = (x_{i_1^{\mathcal A}},p_{i_1^{\mathcal A}},x_{o_2^{\mathcal A}},p_{o_2^{\mathcal A}}, x_{i_2^{\mathcal A}},p_{i_2^{\mathcal A}}, x_{o_1^{\mathcal A}},p_{o_1^{\mathcal A}})^T
$$

In this form, the last qumode corresponds to the mode of $\mathcal A$ in the first stitch-pair [$(o_1^\mathcal A, i_1^\mathcal B)$] of the J-stitching, i.e., the qumode $o_1^\mathcal A$. Correspondingly, the penultimate qumode corresponds to the mode of $\mathcal A$ that is in the second stitch-pair [$(o_1^\mathcal B, i_2^\mathcal A)$], i.e., the qumode $i_2^\mathcal A$.
The covariance matrix of $\Upsilon_\mathcal A$ must also have its submatrices arranged to correspond to this new $\vec\alpha_A'$ order. Later examples in this Appendix on using the algorithm will keep track of the qumodes in a similar fashion.

 \section{Partial Transpose and Partial Trace of CV states}
\label{appendix:partial-transpose}
Here we elaborate on two operations that appear in the link product: the partial transpose and the partial trace of CV states.

Firstly, a transpose of all $n$ modes in phase space of a CV state, $\rho^T$, corresponds to the mapping of the phase space vector $\vec\alpha \mapsto -\Lambda_n\vec\alpha$ in the characteristic function~\cite{hu2021operator} where $\Lambda_n = \bigoplus_{i=1}^n \text{diag}(1,-1)$.

On the other hand, the {\it partial transpose}, $\rho^{T_{J}}$, involves the transposition of only a subset of modes. In the context of this paper, $\Upsilon_{\mathcal A}^{T_{J_\mathcal A}}$ is the transposition of $\Upsilon_{\mathcal A}$ on $\bigotimes_{x\in J_\mathcal A}\mathcal H^x$ (and identity elsewhere). In text, we have written this to correspond to the transformation of the characteristic function of $\Upsilon_{\mathcal A}$ as:

\begin{align}
    \chi_{{\Upsilon_\mathcal A}^{T_{J_\mathcal A}}} (\dots, x_j,p_j, \dots) = \chi_{\Upsilon_\mathcal A} (\dots, -x_j, p_j, \dots) \, \forall j \in J_\mathcal A.
\end{align}

We note that for a partial transpose on a single mode $y$, $\Upsilon_{\mathcal A}^{T_y}$, we have

\begin{align}
     \chi_{\Upsilon_\mathcal A^{T_{J_\mathcal A}}}(\dots, x_j,p_j, \dots) = \chi_{\Upsilon_\mathcal A} (\dots, -x_y, p_y, \dots) =  \chi_{\Upsilon_\mathcal A} (\vec\alpha_{A\setminus \set{y}} \oplus -\Lambda_1 \vec\alpha_y)
\end{align}

where $\alpha_y =(x_y, p_y)^T$ corresponding to mode $y$ (i.e., $\vec\alpha_y \mapsto -\Lambda_1 \vec\alpha_y$), $\Lambda_1=\sigma_z=\operatorname{diag}(1,-1)$ and $\vec\alpha_{A\setminus \set{y}}$ corresponds to all other modes of $\Upsilon_\mathcal A$. Then, the transformation of the characteristic function from the partial transpose $\Upsilon_{\mathcal A}^{T_{J_\mathcal A}}$ can be equivalently written as:

\begin{align}
    \chi_{{\Upsilon_A}^{T_{J_\mathcal A}}} (\vec\alpha_{A\setminus J} \oplus \vec\alpha_J)= \chi_{\Upsilon_A} (\vec\alpha_{A\setminus J} \oplus -\Lambda_{m} \vec\alpha_J)
\end{align}

where $\Lambda_{m} = \bigoplus_{j=1}^{m} \text{diag}(1,-1)$, and $m$ is the number of modes of $\mathcal A$ involved in the stitch (i.e., $m=|J_\mathcal A|$).

Practically, for Gaussian states, we can interpret the effect of the partial transpose in terms of the effect on the Gaussian state's covariance matrix. One can rearrange $\vec \alpha_A$ into $\vec\alpha'_A$ whose last $m$ indices correspond to the modes $j \in J_\mathcal{A}$:

$$
\vec\alpha_{\mathcal A}' = (\vec\alpha_{A\setminus J}, \, x_{j=m},p_{j=m},\, \dots,\, x_{j=1},p_{j=1})^T
$$

Here $\vec\alpha_{A\setminus J}$ corresponds to the qumodes of $(I_\mathcal A\cup O_\mathcal A)\setminus J_\mathcal{A}$ (i.e., unlinked modes of $\Upsilon_\mathcal A$). See previous subsection for examples of re-ordering qumodes.

Then, if $\Upsilon_\mathcal A$ is a Gaussian state, the covariance matrix in this basis $\vec\alpha'_A$ can be divided into the following block matrix,

\begin{align}
    \Gamma' = \begin{pmatrix}
        \gamma_{A/J} & \gamma_{A,J} \\
        \gamma_{A,J}^T & \gamma_J
    \end{pmatrix},
\end{align}

where $\gamma_J$ corresponds to the covariance block of stitched modes $J$, $\gamma_{A/J}$ corresponds to the covariance block of modes $A\setminus J=(I_\mathcal A\cup O_\mathcal A) \setminus J_{\mathcal A}$, and $\gamma_{A,J}$ corresponds to the interaction terms between the modes to be joined, and the modes not to be joined. In this new basis of $\vec\alpha'_A$, one finds that $\Upsilon_\mathcal A^{T_{J_\mathcal A}}$ has the covariance matrix

\begin{align}
    \Gamma_{\Upsilon_\mathcal A^{T_{J_\mathcal A}}} &= (\mathbb I \oplus \Lambda_m) \Gamma' (\mathbb I \oplus \Lambda_m) \\
    &= \begin{pmatrix}
        \gamma_{A\setminus J} & \gamma_{A,J} \Lambda_m \\
        \Lambda_m\gamma_{A,J}^T & \Lambda_m\gamma_J \Lambda_m
    \end{pmatrix}.
    \label{eqn: partial-transpose-cm}
\end{align}

where $\mathbb I$ is the identity matrix with the same dimensions as $\gamma_{A\setminus J}$~\cite{serafini2023quantum}.

The form of (\ref{eqn: partial-transpose-cm}) can be obtained by applying the aforementioned transformation of $\vec\alpha_J \mapsto -\Lambda_m \vec\alpha_J$ in characteristic function, $\chi_{\Upsilon_{\mathcal A}}(\vec\alpha) = e^{-\frac{1}{4}\vec\alpha^T\Omega^T \Gamma' \Omega\vec\alpha}e^{i\vec \alpha^T\Omega^T \vec{d}}$, due to the partial transpose and then taking into account $\Lambda\Omega = -\Omega\Lambda$ (where the square matrices $\Omega, \Lambda$ both have the same dimensions) to propagate the effect of $\Lambda$ on the covariance matrix.

Another important operation is the partial trace. \(\operatorname{tr}_J\) denotes the partial trace over the
pairwise-identified stitched Hilbert spaces
\(\mathcal H^{\underline{o}_j}\cong\mathcal H^{\underline{i}_j}\) for all \(j=(\underline{o}_j,\underline{i}_j)\in J\), since through the link product, the stitched spaces are equivalent. We note that a useful property that will be used later relating to the trace and Weyl operators is $\text{tr}[\hat D(\vec\alpha_i)D(\vec\alpha_o)] \propto \delta(\vec\alpha_i+\vec\alpha_o)$~\cite{serafini2023quantum}.

\section{Proof of Result ~\ref{jstitch}}
\label{proof: link-prod-circuitfragments}

We note that \cite{giedke2002characterization,fiurasek_gaussian_2002} provides the form to evaluate the stitching of all output modes of a CV state stitched to all the input modes of a CV channel (for completeness, we present as Lemma 1 below). We therefore build on the work in ~\cite{giedke2002characterization,fiurasek_gaussian_2002} to prove Result \ref{jstitch}, which is a more general form of the results in the aforementioned literature that allows for only a subset of outputs of a first process, $\mathcal A$, to be fed into a subset of inputs of the second process, $\mathcal B$. This result easily generalizes to the case of bi-directional stitching, i.e., a subset of {\it both inputs and outputs} of $\mathcal A$ being mode-stitched with a subset of {\it outputs and inputs} of $\mathcal B$ — and in later examples in the Appendix, we demonstrate such bidirectional stitchings using the link product.

\begin{restatement}{Result \ref{jstitch}}
Consider two CV circuit fragments \(\mathcal A\) and \(\mathcal B\),
with Choi states \(\Upsilon_{\mathcal A}(\mathbf r)\) and
\(\Upsilon_{\mathcal B}(\mathbf s)\), and associated squeezing parameters  $\mathbf r=\{r_i:i\in I_{\mathcal A}\}$ and $\mathbf s=\{s_i:i\in I_{\mathcal B}\}$. Let $\mathcal{C}$ be the resulting composite circuit fragment formed by the unidirectional $J$-stitching some of the output modes of $\mathcal{A}$ (set $J_\mathcal{A} = O_J$) with some of the input modes of $\mathcal{B}$ (set $J_\mathcal{B} = I_J$). Then
\begin{equation}
\Upsilon_{\mathcal A}\star_J \Upsilon_{\mathcal B} =
\lim_{\mathbf s_J\to\infty}
K(\mathbf s_J)\,
\operatorname{tr}_J
\!\left[
\Upsilon_{\mathcal A}(\mathbf r)^{T_{J_{\mathcal A}}}
\Upsilon_{\mathcal B}(\mathbf s)
\right],
\end{equation}
where $K(\mathbf s_J)
=
\prod_{i\in I_J}\cosh^2 (s_i/2)$ is the constant of normalization, the limit is shorthand for taking $s_i \rightarrow \infty$ for all modes $i \in I_J$. Here, the partial transpose \(T_{J_{\mathcal A}}\) is taken on the modes of \(\Upsilon_{\mathcal A}\) participating in the stitching,
and \(\operatorname{tr}_J\) denotes the partial trace.
\end{restatement}

What we wish to prove, therefore, is that the link product (J-stitching) of two circuit fragments, $\mathcal A$, $\mathcal B$, on modes $O_J$ to $I_J$, can be evaluated with the formula in the theorem. In  order to achieve this, we first introduce the following two lemmas.

\subsection{Lemma 1}

\begin{lemma}[Stitching a state into a channel with link product (all outputs to all inputs)]
Consider a CV channel $\mathcal A$ with input and output qumodes $I_\mathcal A$ and $O_\mathcal A$ respectively, with Choi state $\Upsilon_\mathcal A(\mathbf r)$ on $\mathcal D(\bigotimes_{i\in {I_\mathcal A}} \mathcal H^i \otimes \bigotimes_{o\in {O_\mathcal A}} \mathcal H^o)$. Consider a quantum state $\rho$, with output qumodes $O_\rho$, in the space $\mathcal D(\bigotimes_{o\in {O_\rho}} \mathcal H^o)$ that is equivalent to $\mathcal D(\bigotimes_{i\in {I_\mathcal A}} \mathcal H^i)$. Then, the J-stitching of $\mathcal A$ and $\rho$ on the set of modes $J=\set{(\underline{o}_\ell, \underline{i}_\ell):\underline{i}_\ell \in I_{\mathcal A}, \underline{o}_\ell \in O_\rho}$, such that all the outputs of $\rho$ are fed into all the inputs of $\mathcal A$, can be evaluated with
    \begin{align}
\rho \star \Upsilon_{\mathcal A} &:= \mathcal A (\rho) \\
 &= \lim_{\mathbf r\to\infty} K(\mathbf r) \operatorname{tr}_{O_\rho}[\rho^{T} \Upsilon_\mathcal A(\mathbf r)]
\end{align}

$\rho \star \Upsilon_{\mathcal A}= \mathcal A(\rho)$ is clear from the definition of the link product, where $\rho$ and $\mathcal A$ are circuit fragments, and all the outputs of $\rho$ are fed into all the inputs of $\mathcal A$. Here, $K(\mathbf r)$ is the normalization constant, and $\text{tr}_{O_\rho}$ is the partial trace on $\bigotimes_{o\in O_\rho} \mathcal H^{o}\cong\bigotimes_{i\in I_\mathcal A} \mathcal H^{i}$.
\label{lemma:1}
\end{lemma}

\begin{proof}
The action of $\mathcal A(\rho)$ where $o_j \in O_{\rho}$ can be expressed in terms of $\Upsilon_{\mathcal A}$ as the trace relation \cite{giedke2002characterization}

\begin{align}
\mathcal A(\rho) \propto \lim_{\mathbf r\to\infty}\text{tr}_{O_\rho}[\rho^{T} \Upsilon_\mathcal A(\mathbf r)] \label{eqn:Aofrho}
\end{align}

where we have included explicitly the dependence on $\mathbf r$. This inclusion of the $\mathbf r \to\infty$ dependence is later motivated in the example in \cite{giedke2002characterization} in order for the join to yield the correct output state.

Below, we evaluate the proportionality constant (also referred to as the normalization constant here) by considering the case of the link product between a general state, $\rho$ and the Choi state of an identity channel. 

Consider a general state in the Fock basis,

\begin{align}
    \rho^T = \sum_{l,m} \lambda_{l,m} \ket l\bra m_i,
\end{align}

and the Bell state (Choi state of an identity channel),

\begin{align}
    \Phi(r) = \frac{1}{\cosh^2 (r/2)} \sum_{j,k } \tanh^k (r/2) \tanh^j (r/2) \ket k_i\ket k_o \bra j_i \bra j_o .
\end{align}

We know that the link product of these two states should return $\rho$, since $\mathcal I (\rho) = \rho \star \Phi(\mathbf r) = \rho$. To obtain the normalization constant, we evaluate the link product given by $\rho \star \Phi(\mathbf r) = \lim_{r\to\infty} K(r) \text{tr}_i [\rho^T \Phi(r)]$, by using $\braket{m|k} =\delta_{mk}$ and $\text{tr}_i[\ket l\bra j_i] = \text{tr}_i[\braket {j|l}_i] = \delta_{lj}$:

\begin{align}
    \lim_{r\to\infty}K(r)\ \text{tr}_i [\rho^T \Phi(r)] &= \lim_{r\to\infty} K(r)\ \text{tr}_i \left[\sum_{j,k,l,m} \frac{\tanh^k (r/2) \tanh^j (r/2)}{\cosh^2(r/2)} \lambda_{lm} \ket l \braket{m|k}\bra j_i \otimes \ket k \bra j_o \right] \\
    &= \lim_{r\to\infty} K(r)\ \sum_{l,m}  \frac{\tanh^m (r/2) \tanh^l (r/2)}{\cosh^2(r/2)} \lambda_{lm} \ket m \bra l
\end{align}

Since $\tanh r\to1$ as $r\to\infty$, then for $\lim_{r\to\infty} K(r) \text{tr}_i [\rho^T \Phi(r)] = \rho$ to be true, we need

\begin{align}
    K(r) =\cosh^2(r/2).
\end{align}

Finally, we obtain the normalization factor for Eq.~\ref{eqn:Aofrho} as above, because $\mathcal{A}$ is trace preserving. Evidently, in the case where we have more than one mode-stitching (and thus more than 1 Bell state), the normalization factor will simply become $K(\mathbf r) = \prod_j \cosh^2(r_j/2)$ where $r_j$ denotes the squeezing parameter of the $j$-th Bell pair involved in the join.

\end{proof}

It also becomes obvious through this calculation why the limit of $r\to\infty$ must be introduced for the link product: if $r\not\to \infty$, the shape of our resulting state will be affected due to the $\tanh r$ terms, i.e., the link product may yield a different state altogether. This is what we mean when we say that not all information passes through the join.

\subsection{Lemma 2}

We note that Lemma \ref{lemma:1} would still apply to the most general case where only a subset $O_{\rho}$ is fed into a subset of $I_{\mathcal A}$ by considering instead $\mathcal A'(\rho')$ where $\rho' = \rho \otimes \Phi(\mathbf r')$ and $\mathcal A' (\cdot) = \mathcal I \otimes \mathcal A (\cdot)$. We elucidate an explicit version of this case in the following:

\begin{lemma}[Stitching a state into a channel with link product formula (partial outputs to partial inputs)]
Consider a CV channel $\mathcal A$ with input and output qumodes $I_\mathcal A$ and $O_\mathcal A$ respectively with Choi state $\Upsilon_\mathcal A(\mathbf r)$ on $\mathcal D(\bigotimes_{i\in {I_\mathcal A}} \mathcal H^i \otimes \bigotimes_{o\in {O_\mathcal A}} \mathcal H^o)$, and a quantum state $\rho$ with output qumodes $O_\rho$  in the space $\mathcal D(\bigotimes_{o\in {O_\rho}} \mathcal H^o)$. $\rho$ is to be J-stitched to $\mathcal A$ such that some of the output modes of $\rho$, $O_J \subset O_\rho$, are fed into some of the inputs of $\mathcal A$, $I_J \subset I_{\mathcal A}$, where $O_J = \set{\underline{o}_\ell: (\underline{o}_\ell, \underline{i}_\ell) \in J}$ and $I_J = \set{\underline{i}_\ell: (\underline{o}_\ell, \underline{i}_\ell) \in J}$. 

Then, the J-stitching of $\mathcal A$ and $\rho$ on the set of modes $J$, can be evaluated with

\begin{align}
\rho \star_J \Upsilon_{\mathcal A} &= \lim_{\mathbf r_{J}\to\infty} K(\mathbf r_J)  \operatorname{tr}_J[\rho^{T_J}\Upsilon_\mathcal A(\mathbf r)]
\end{align}

where $K(\mathbf r_J) = \prod_{j \in I_J} \cosh^2(r_j/2)$ is the normalization constant. $\text{tr}_J$ denotes the partial trace over the
pairwise-identified stitched Hilbert spaces \(\mathcal H^{\underline{o}_\ell}\cong\mathcal H^{\underline{i}_\ell}\) for all \((\underline{o}_\ell,\underline{i}_\ell)\in J\). $T_J$ corresponds to the partial transpose of the modes of $\rho$ participating in the stitching.
\label{lemma:2}
\end{lemma}

\begin{proof}
\begin{figure}
    \centering
    \includegraphics[width=0.6\linewidth]{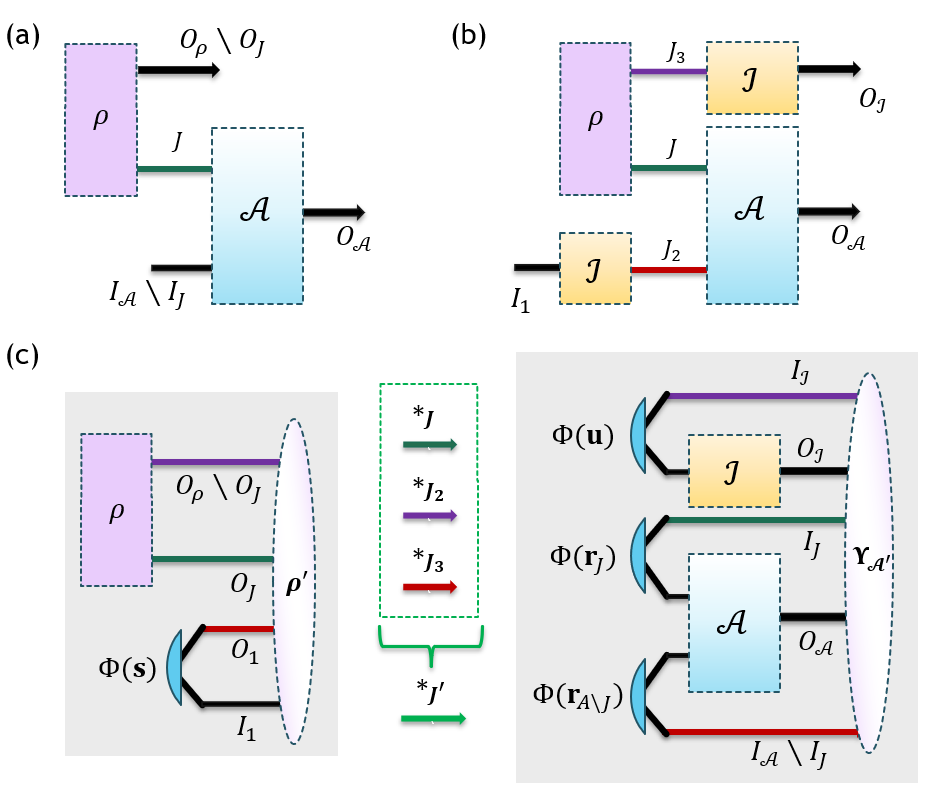}
    \caption{\textbf{Recasting Lemma 2.} $\rho\star_J \Upsilon_{\mathcal A}$, which involves only a subset of modes of $\rho$ stitched to a subset of modes of $\mathcal A$ (a) can be recast as a full all-output to all-input $J'$-stitching of $\rho\otimes \mathcal I$ (whose Choi state is $\rho'$) and $\mathcal A'(\cdot) = \mathcal I \otimes \mathcal A(\cdot)$ (whose Choi state is $\Upsilon_{\mathcal A'}$), such that we can apply Lemma 1 to evaluate this partial stitch (b). In this case, $J'$ constitutes three `smaller' stitches: our original $J$ (green line), $J_2$ (purple line) and $J_3$ (red line). This $J'$-stitching can be explicitly expressed as the link product $\rho'\star_{J'} \Upsilon_{\mathcal A'}$ (c). We note that each horizontal line represents a {\it set} of modes, rather than a single mode.}
    \label{fig:appendix-lemma-2-illus}
\end{figure}

In Fig.~\ref{fig:appendix-lemma-2-illus}, we illustrate how $\rho\star_J \Upsilon_{\mathcal A}$ (which involves the stitching of a subset of output modes of $\rho$ to a subset of input modes of $\Upsilon_\mathcal A$) can be recast to a full all-output to all-input $J'$-stitching of $\rho'\star_{J'} \Upsilon_{\mathcal A'}$, such that we can invoke Lemma 1 to evaluate the {\it original} $J$-stitching. To do so, we introduce two sets of Bell pairs: (1) $\Phi(\mathbf s)$ with input and output modes $I_1, O_1$, which acts of the Choi state to the identity channel in $\rho\otimes \mathcal I$, such that for $\rho\otimes \mathcal I$, its Choi state is $\rho' = \rho \otimes \Phi(\mathbf s)$, and (2) $\Phi(\mathbf u)$ with input and output modes $I_{\mathcal I},O_{\mathcal I}$, which acts as the Choi state of the identity channel $\mathcal I$ such that for $\mathcal A'(\cdot) =\mathcal I \otimes \mathcal A (\cdot)$, its Choi state is given by $\Upsilon_{\mathcal A'} = \Phi(\mathbf u) \otimes \Upsilon_{\mathcal A}$. 

$\rho'$ is thus in the space $\mathcal D(\bigotimes_{o\in {O_\rho}} \mathcal H^o \otimes \bigotimes_{o\in {O_1}} \mathcal H^o \otimes \bigotimes_{i\in {I_1}} \mathcal H^i)$ and $\Upsilon_{\mathcal A'}$ is on the space $\mathcal D(\bigotimes_{i\in {I_{\mathcal I}}} \mathcal H^i \otimes \bigotimes_{o\in {O_{\mathcal I}}} \mathcal H^o \otimes \bigotimes_{i\in {I_\mathcal A}} \mathcal H^i \otimes \bigotimes_{o\in {O_\mathcal A}} \mathcal H^o)$.

To evaluate the $J$-stitch, we consider the new $J'$-stitching between $\rho'$ and $\Upsilon_{\mathcal A'}$, as per Fig.~\ref{fig:appendix-lemma-2-illus}. In particular, for the $J'$-stitching, we have that the output modes to be stitched, $O_{J'} = (O_\rho \cup O_1)$ and the input modes to be stitched, $I_{J'} = (I_\mathcal A \cup I_{\mathcal I})$. We also note that the $J'$-stitching can be `broken down' into three stitches: $J$ (the original $J$-stitching), $J_2$ (stitching between $O_1$ and $(I_{\mathcal A} \setminus I_J)$) and $J_3$ (stitching between $(O_\rho \setminus O_J)$ and $I_\mathcal{I}$).

Additionally, given $\Upsilon_\mathcal A(\mathbf r)$, we can explicitly rewrite $\mathbf r = \mathbf r_J \oplus \mathbf r_{A\setminus J}$ where $\mathbf r_J$ corresponds to the squeezing parameters of the Bell states on $\mathcal D(\bigotimes_{i\in I_J} \mathcal H^{i})$ (involved in the $J$-stitching), and $\mathbf r_{A\setminus J}$ corresponds to the squeezing parameters of the Bell pairs associated with the unstitched input modes of $\mathcal A$, i.e., $I_{\mathcal A} \setminus I_J$. 

Then, invoking Lemma \ref{lemma:1}, with $\text{tr}_{{J'}}$ as the partial trace denotes the over the pairwise-identified stitched Hilbert spaces \(\mathcal H^{\underline{o}_\ell}\cong\mathcal H^{\underline{i}_\ell}\) for all $(\underline{o}_j, \underline{i}_\ell) \in J'$:

\begin{align}
\rho' \star_{J'} \Upsilon_{\mathcal A'} &= \mathcal A' (\rho') \\
 &= \lim_{\mathbf r_{A\setminus J},\mathbf u,\mathbf r_J\to\infty} K(\mathbf r_{A\setminus J}\oplus \mathbf u \oplus \mathbf r_J)\ \text{tr}_{J'}[(\rho\otimes\Phi(\mathbf s))^{T} (\Phi(\mathbf u)\otimes\Upsilon_\mathcal A(\mathbf r_{A\setminus J} \oplus \mathbf r_J))] \\
 &= \lim_{\mathbf r_J \to\infty} K(\mathbf r_J)\ \text{tr}_J[ \lim_{\mathbf u \to\infty} K(\mathbf u)\ \text{tr}_{J_3} [\lim_{\mathbf r_{A\setminus J} \to\infty} K(\mathbf r_{A\setminus J})\ \text{tr}_{J_2}[(\rho^T\otimes\Phi(\mathbf s)^T)(\Phi(\mathbf u)\otimes\Upsilon_\mathcal A(\mathbf r_{A\setminus J}\oplus \mathbf r_J))]]] \\
 &= \lim_{\mathbf r_J \to\infty}K(\mathbf r_J)\ \text{tr}_J[ \lim_{\mathbf u \to\infty} K(\mathbf u)\ \text{tr}_{J_3}[\rho^T(\Phi(\mathbf u)\otimes\Upsilon_\mathcal A(\mathbf s\oplus \mathbf r_J))]] \\
&= \lim_{\mathbf r_J \to\infty} K(\mathbf r_J)\ \text{tr}_J[ \lim_{\mathbf u \to\infty} K(\mathbf u)\ \text{tr}_{J_3}[\rho^{T_{J_3} T_J}(\Phi(\mathbf u)\otimes\Upsilon_\mathcal A(\mathbf s\oplus \mathbf r_J))]] \\
&= \lim_{\mathbf r_J \to\infty} K(\mathbf r_J)\ \text{tr}_J[ \rho^{T_J}\Upsilon_\mathcal A(\mathbf s\oplus \mathbf r_J)] \\
 &= \lim_{\mathbf r_J\to\infty} K(\mathbf r_J)\ \text{tr}_J[\rho^{T_J} \Upsilon_\mathcal A(\mathbf r)],
\end{align}

where $\text{tr}_{{J_i}}$ for $i=2$ and $i=3$ are the partial traces denotes the over the pairwise-identified stitched Hilbert spaces \(\mathcal H^{\underline{o}_\ell}\cong\mathcal H^{\underline{i}_\ell}\) for all $(\underline{o}_j, \underline{i}_\ell) \in J_i$, and $T_{J_3}$ is the partial transpose of $\rho$ on the modes $O_\rho\setminus O_J$. We reiterate that $\mathbf r_J = \set{r_j : j \in I_J}$, and $\mathbf r_J \to\infty$ means $r_j\to\infty$ for all $j\in I_J$.

We elucidate what is happening in this proof: (C11) to (C12) is a direct application of Lemma 1. From (C12) to (C13), we explicitly split the partial trace $\trace_{J'}$ into the partial traces $\trace_{J}, \trace_{J_2}, \trace_{J_3}$, and bring the limits in where possible (due to its independence from the partial trace it is brought into). To reach the final line, we evaluate each subsystem partial trace and limit independently (first $J_2$, then $J_3$) until we are only left with that of the J-stitching. To do so, we observe three things, as follows. 

First, the partial trace and partial transpose operations on $\rho\otimes\Phi(\mathbf s)$ involve the trace operation on the modes $O_J \cup (O_\rho \setminus O_J) \cup O_{1}$.

Second, we can evaluate the partial transpose, partial trace (and take the limit) of the subsets $O_1$ (corresponding to the $J_2$-stitch) and $O_\rho\setminus O_J$ (corresponding to the $J_3$-stitch) separately first, since they are all independent and have no causal relation to each other. In doing so, we leave only the partial trace and partial transpose on the subspace $\bigotimes_{j\in O_J} \mathcal H^j$.

Third, any J-stitching of a Choi state (e.g., $\Upsilon_{\mathcal A}(\mathbf r_{A\setminus J}\oplus \mathbf r_J)$) with Bell states (e.g., $\Phi(\mathbf s)$) will simply return the state, $\Upsilon_\mathcal A$, but with a modified set of squeezing parameters where it was stitched (e.g., here we return $\Upsilon_{\mathcal A}(\mathbf s\oplus\mathbf r_J)$ because $I_\mathcal A\setminus I_J$, with dependence $\mathbf r_{A\setminus J}$ was $J_3$-stitched with $O_1$ with dependence $\mathbf s$, thus $\mathbf s$ will `replace' $\mathbf r_{A\setminus J}$ in the Choi state). Relabel $\mathbf s = \mathbf r_{A\setminus J}$, then $\Upsilon_{\mathcal A}(\mathbf s \oplus \mathbf r_J) = \Upsilon_{\mathcal A}(\mathbf r_{A\setminus J} \oplus \mathbf r_J) = \Upsilon_{\mathcal A}(\mathbf r)$.
\end{proof}

\subsection{Result Proof}

We return to the proof for Result 1. The idea is to first construct the channel $\mathcal C$ as per Fig.~\ref{fig:appendix-result1-a}b, then find its Choi state, $\Upsilon_\mathcal C$ as per Fig.~\ref{fig:appendix-result1-a}c, using the definition of the Choi state. We also know that $\Upsilon_\mathcal A \star_J \Upsilon_\mathcal B = \Upsilon_\mathcal C$ by definition. From there, we make use of Lemma 1 to re-express the form of $\Upsilon_\mathcal C$ to obtain the form of Result 1. 

\begin{proof}
The two CV circuit fragments $\mathcal A$ and $\mathcal B$ have with input qumodes $I_\mathcal A, I_\mathcal B$ respectively and output qumodes $O_\mathcal A, O_\mathcal B$ respectively. Consider the J-stitching of a subset of the inputs of $\mathcal B$ and the subset outputs of $\mathcal A$ to a yield a circuit fragment $\mathcal C$ (see Fig.~\ref{fig:appendix-result1-a}a-b) such that its Choi state is, by Definition~\ref{def:link-prod-CV}, given by $\Upsilon_\mathcal C = \Upsilon_\mathcal A \star_J \Upsilon_\mathcal B$ .

\begin{figure}
    \centering
    \includegraphics[width=0.95\linewidth]{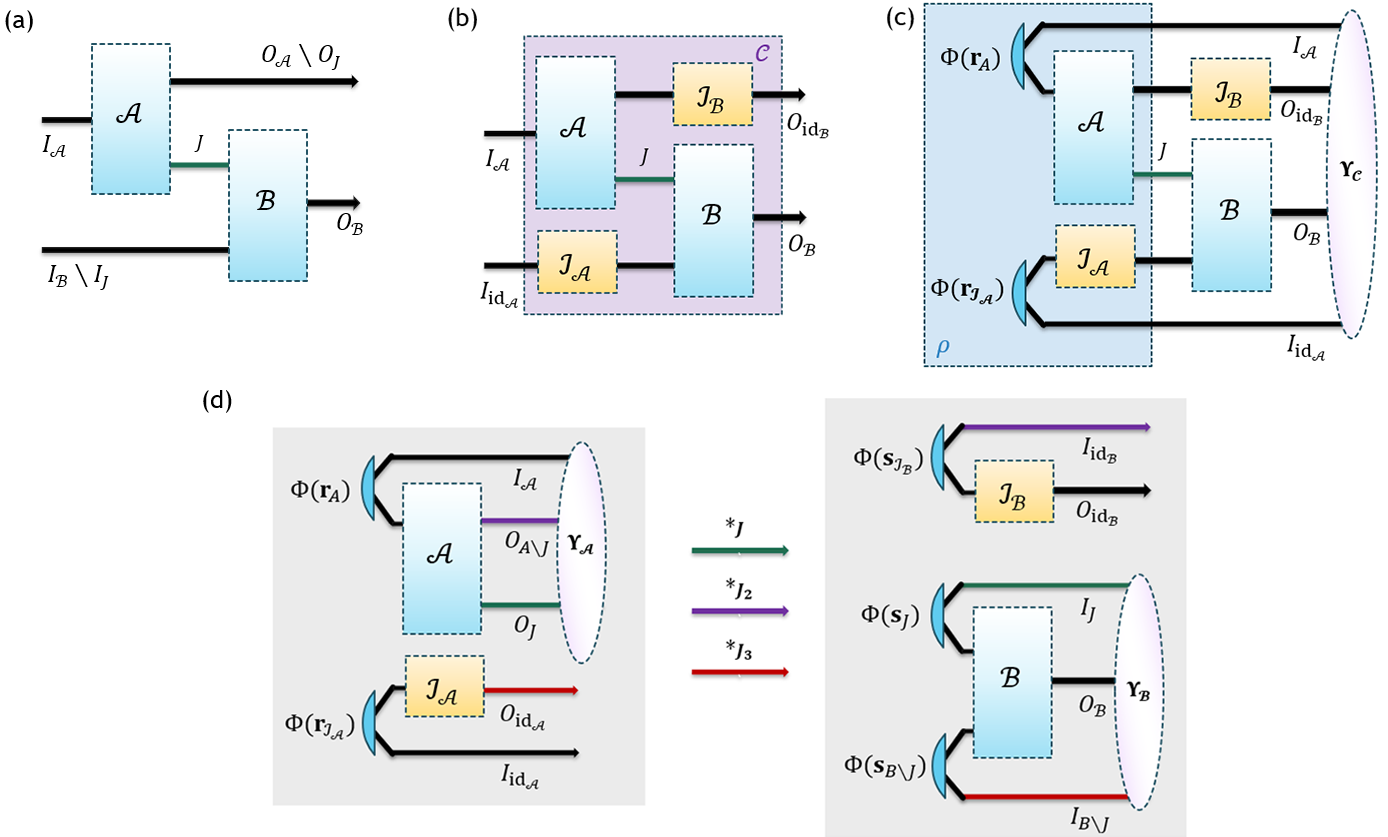}
    \caption{\textbf{Recasting Result 1.} Given two circuit fragments $\mathcal A$ and $\mathcal B$ that are to be $J$-stitched (a), they can be recast into $\mathcal C$ (b). The Choi state of $\mathcal C$ can be found (c) and re-expressed as the state of $\rho$ (denoted in the blue box) as an input to $\mathcal B \otimes \mathcal I_{\mathcal B}$. Equivalently, $\Upsilon_{\mathcal C}$ can be obtained by evaluating the stitchings in (d). We note here that all the lines can denote a collection of modes, rather than a single mode.}
    \label{fig:appendix-result1-a}
\end{figure}

Equivalently for such a J-stitching, we can consider that $\mathcal C = (\mathcal A \otimes \mathcal I_A ) \circ (\mathcal B \otimes \mathcal I_B)$, where $\mathcal I_A$ and $\mathcal I_B$  has the input and output modes $I_{\text{id}_A}, O_{\text{id}_A}$ and $I_{\text{id}_B}, O_{\text{id}_B}$ respectively, as per Fig.~\ref{fig:appendix-result1-a}b. Plainly, $\mathcal I_A$ acts on the input modes of $\mathcal B$ that are not involved in the J-stitching. Similarly, $\mathcal I_B$ acts on the output modes of $\mathcal A$ that are not involved in the J-stitching. 

 In this proof, we will consider the Choi states of each of these circuit fragments (see Fig.~\ref{fig:appendix-result1-a}d).To help clarify notation, we start first with defining and providing the physical understanding of the notation that will be used. Such labels are re-iterated in Figure \ref{fig:appendix-result1-a}d.

For the Choi state of $\mathcal A \otimes \mathcal I_A$:
\begin{itemize}
    \item $O_J$ are the output modes of $\mathcal A$ involved in the J-stitching
    \item $O_{A \setminus J} = \set{o_k : (o_k \in O_\mathcal A)\wedge (o_k \notin O_J)}$ are the output modes of $\mathcal A$ not involved in the J-stitching
    \item $O_{\text{id}_A}$ are the output modes of $\mathcal I_A$ (to be connected to $I_{B\setminus J}$ defined below)
    \item $I_{\text{id}_A}$ are the input modes of $\mathcal I_A$
\end{itemize}

Similarly, for the Choi state of $\mathcal B \otimes \mathcal I_B$:
\begin{itemize}
    \item $I_J$ are the input modes of $\mathcal B$ involved in the J-stitching
    \item $I_{B\setminus J} = \set{i_k : (i_k \in I_\mathcal B)\wedge (i_k \notin I_J)}$ are the input modes of $\mathcal B$ not involved in the J-stitching
    \item $O_{\text{id}_B}$ are the output modes of $\mathcal I_B$
    \item $I_{\text{id}_B}$ are the input modes of $\mathcal I_B$ (to be connected to $O_{A\setminus J}$ defined above)
\end{itemize}

 We note that in such a stitching of $\mathcal A \otimes \mathcal I_{A}$ and $\mathcal B \otimes \mathcal I_B$ as per Fig.~\ref{fig:appendix-result1-a}d, the following spaces become equivalent to each other:
\begin{itemize}
    \item $\mathcal D(\bigotimes_{i\in {I_{\text{id}_B}}} \mathcal H^i)$ is equivalent to $\mathcal D(\bigotimes_{o\in {O_{A \setminus J}}} \mathcal H^o)$ - corresponding to the $J_2$-stitch (purple)
     \item $\mathcal D(\bigotimes_{o\in {O_{\text{id}_A}}} \mathcal H^o)$ is equivalent to $\mathcal D(\bigotimes_{i\in {I_{B\setminus J}}} \mathcal H^i)$ - corresponding to the $J_3$-stitch (red)
    \item $\mathcal D(\bigotimes_{o\in {O_{J}}} \mathcal H^o)$ is equivalent to $\mathcal D(\bigotimes_{i\in {I_J}} \mathcal H^i)$ - corresponding to the original $J$-stitch of $\mathcal A \star_J\mathcal B$ (green)
\end{itemize}

Having cleared up the notation, we proceed with the proof for this result.

Observe that $\mathcal C(\cdot) = \mathcal B \otimes \mathcal I_{B} (\mathcal A \otimes \mathcal I_{A} (\cdot))$. Using Definition \ref{definition: choi-state-cv-fragment}, we find the Choi state of $\mathcal C$ (Fig.~\ref{fig:appendix-result1-a}c):

\begin{align}
\Upsilon_\mathcal C &= \mathcal B \otimes \mathcal I_{B} \otimes \mathcal I_{I_{B'}} \left[\mathcal A \otimes \mathcal I_{A} \otimes \mathcal I_{I_{A'}} [\Phi(\mathbf r_{\mathcal A} \oplus \mathbf r_{\mathcal I_{\mathcal A}})]\right] \\
&= \mathcal B \otimes \mathcal I_{B} \otimes \mathcal I_{I_{B'}} [\mathcal A\otimes \mathcal I_{I_{\mathcal A}}(\Phi(\mathbf r_{\mathcal A})) \otimes (\mathcal I_{A}\otimes \mathcal I_{I_{\text{id}_A}})(\Phi(\mathbf r_{\mathcal I_{\mathcal A}}))] \\
&= \mathcal B \otimes \mathcal I_{B} \otimes \mathcal I_{I_{B'}} [\mathcal A\otimes \mathcal I_{I_{\mathcal A}}(\Phi(\mathbf r_{\mathcal A})) \otimes \Phi(\mathbf r_{\mathcal I_{\mathcal A}})] \\
&= \mathcal B \otimes \mathcal I_{B} \otimes \mathcal I_{I_{B'}} [\Upsilon_{\mathcal A}(\mathbf r_{\mathcal A}) \otimes \Phi(\mathbf r_{\mathcal I_{\mathcal A}})]
\end{align}

where $\mathcal I_{I_{A'}}$ is the identity channel that acts on the modes ${I_\mathcal A \cup I_{\text{id}_A}}$ (i.e., the input arms of the Bell states $\Phi(\mathbf r_{\mathcal A} \oplus \mathbf r_{\mathcal I_{\mathcal A}}) = \Phi(\mathbf r_{\mathcal A}) \otimes \Phi(\mathbf r_{\mathcal I_{\mathcal A}})$) - in Fig.~\ref{fig:appendix-result1-a}d, this identity channel is not drawn for clarity. Similarly, $\mathcal I_{I_{B'}}$ acts on the modes ${I_\mathcal B \cup I_{\text{id}_B}} = {I_{B\setminus J} \cup I_J \cup I_{\text{id}_B}}$.

It is clear that (C21) takes the form of a channel on some state, i.e., $\mathcal T(\rho)$, where here $\mathcal T(\cdot) = \mathcal B \otimes \mathcal I_{B} \otimes \mathcal I_{I_{B'}} (\cdot)$ and $\rho = \Upsilon_{\mathcal A}(\mathbf r_{\mathcal A}) \otimes \Phi(\mathbf r_{\mathcal I_{\mathcal A}})$. From (C21), we invoke Lemma \ref{lemma:1} to yield,

\begin{align}
 \Upsilon_\mathcal C &= \mathcal B \otimes \mathcal I_{B} \otimes \mathcal I_{I_{B'}} [\Upsilon_{\mathcal A}(\mathbf r_{\mathcal A}) \otimes \Phi(\mathbf r_{\mathcal I_{\mathcal A}})] \\
&= \lim_{\mathbf s_{\mathcal B}, s_{\mathcal I_{\mathcal B}} \to\infty} K(\mathbf s_J \oplus \mathbf s_{B\setminus J} \oplus s_{\mathcal I_{\mathcal B}})\ \text{tr}_{J'}[[\Upsilon_{\mathcal B}(\mathbf s_J \oplus \mathbf s_{B\setminus J}) \otimes \Phi(\mathbf s_{\mathcal I_{\mathcal B}})]^T [\Upsilon_{\mathcal A}(\mathbf r_{\mathcal A}) \otimes \Phi(\mathbf r_{\mathcal I_{\mathcal A}})]]
\end{align}

where the Bell states $\Phi(\mathbf s_\mathcal B) = \Phi(\mathbf s_J \oplus \mathbf s_{B\setminus J} ) = \Phi (\mathbf s_J) \otimes \Phi (\mathbf s_{B\setminus J})$ as per Fig.~\ref{fig:appendix-result1-a}d.

From (C22) to (C23), we use the observation that the Choi state of $\mathcal B \otimes \mathcal I_{\mathcal B}$ is $\Upsilon_{\mathcal B \otimes \mathcal I_{\mathcal B}}(\Phi(\mathbf s_\mathcal B \oplus \mathbf s_{\mathcal I_{\mathcal B}})) = \Upsilon_{\mathcal B}(s_\mathcal B) \otimes \Phi(\mathbf s_{\mathcal I_{\mathcal B}})$, since the action of an identity channel on the Bell states $\Phi(\mathbf s_{\mathcal I_{\mathcal B}})$ will simply return the Bell states.

At this stage, the transpose acts on all the modes of $\Upsilon_{\mathcal B}(\mathbf s_{\mathcal B}) \otimes \Phi(\mathbf s_{\mathcal I_{\mathcal B}})$, and not just a transpose on $\mathcal D(\bigotimes_{j\in I_J}\mathcal H^j)$, of which the latter is our desired expression. Therefore wish to resolve the partial trace and partial transpose on the modes that are \emph{not} in $J$ (i.e., modes not involved in the J-stitching) in order to recover the form of our final equation in the result. This is equivalent to evaluating the partial trace and partial transpose of the $J_2$- and $J_3$-stitching first (see Fig.~\ref{fig:appendix-result1-a}d), since they are independent.

To this end, we can make the following observations that will simplify the current expression. Firstly, the $J_2$- and $J_3$-stitching should simply yield the Choi states of $\Upsilon_{\mathcal A}$ and $\Upsilon_{\mathcal B}$ respectively - stitching a channel with the identity channel (whose Choi states are the Bell states) should yield the same channel again. Using Lemma 2, We can therefore express this observation formally in (C24) and (C25):

\begin{align}
\lim_{\mathbf s_{\mathcal I_{\mathcal B}}\to\infty}K(\mathbf s_{\mathcal I_{\mathcal B}})\  \text{tr}_{J_2}[\Phi(\mathbf s_{\mathcal I_{\mathcal B}})^{T_{J_2}}\Upsilon_{\mathcal A}(\mathbf r_{\mathcal A})] =  \lim_{\mathbf s_{\mathcal I_{\mathcal B}} \to\infty}K(\mathbf s_{\mathcal I_{\mathcal B}})\ \text{tr}_{J_2}[\Upsilon_{\mathcal A}(\mathbf r_{\mathcal A})^{T_{J_2}}\Phi(\mathbf s_{\mathcal I_{\mathcal B}})] = \Upsilon_{\mathcal A}(\mathbf r_{\mathcal A}),
\end{align}

where $\Phi(\mathbf s_{\mathcal I_{\mathcal B}})^{T_{J_2}}$ is the Bell states transposed on the modes $I_{\mathcal I_{\mathcal B}}$, while $\Upsilon_{\mathcal A}(\mathbf r_{\mathcal A})^{T_{J_2}}$ is the state $\Upsilon_{\mathcal A}$ partial transposed on the modes $O_{A\setminus J}$. Since they act on the same subspaces, i.e., on the subspaces $\mathcal D(\bigotimes_{i\in {I_{\text{id}_B}}} \mathcal H^i) \cong \mathcal D(\bigotimes_{o\in {O_{A \setminus J}}} \mathcal H^o)$, they are equivalent operations on their respective states. The partial trace $\text{tr}_{J_2}$ also acts on the aforementioned subspace. Similarly, for the $J_3$-stitching,

\begin{align}
\lim_{\mathbf s_{B\setminus J}\to\infty}K(\mathbf s_{B\setminus J})\  \text{tr}_{J_3}[\Phi(\mathbf r_{\mathcal I_{\mathcal A}})^{T_{J_3}}\Upsilon_{\mathcal B}(\mathbf s_{B\setminus J} \oplus \mathbf s_J)] =\lim_{\mathbf s_{B\setminus J}\to\infty} K(\mathbf s_{B\setminus J})\ \text{tr}_{J_3}[\Upsilon_{\mathcal B}(\mathbf s_{B\setminus J} \oplus \mathbf s_J)^{T_{J_3}} \Phi(\mathbf r_{\mathcal I_{\mathcal A}})] = \Upsilon_{\mathcal B}(\mathbf r_{\mathcal I_{\mathcal A}} \oplus \mathbf s_J), 
\end{align}

where $T_{J_3}$ and $\text{tr}_{J_3}$ act on the subspaces $\mathcal D(\bigotimes_{o\in {O_{\text{id}_A}}} \mathcal H^o) \cong \mathcal D(\bigotimes_{i\in {I_{B\setminus J}}} \mathcal H^i)$. Additionally, with the squeezing parameters $\mathbf r_{\mathcal I_{\mathcal A}} = \mathbf s_{B\setminus J}$, we recover our original state of $\Upsilon_B (\mathbf s_{\mathcal B})$. Using these formalized observations, we return to manipulate the expressions in our proof in order to obtain the final form with only the partial trace and transpose of $\text{tr}_J$ and $T_J$:

\begin{align}
 \Upsilon_\mathcal C &= \lim_{\mathbf s_{\mathcal B}, s_{\mathcal I_{\mathcal B}} \to\infty} K(\mathbf s_J \oplus \mathbf s_{B\setminus J} \oplus s_{\mathcal I_{\mathcal B}})\ \text{tr}_{J'}[[\Upsilon_{\mathcal B}(\mathbf s_J \oplus \mathbf s_{B\setminus J}) \otimes \Phi(\mathbf s_{\mathcal I_{\mathcal B}})]^T [\Upsilon_{\mathcal A}(\mathbf r_{\mathcal A}) \otimes \Phi(\mathbf r_{\mathcal I_{\mathcal A}})]] \\
&= \lim_{\mathbf s_{\mathcal B}, s_{\mathcal I_{\mathcal B}} \to\infty} K(\mathbf s_J \oplus \mathbf s_{B\setminus J} \oplus s_{\mathcal I_{\mathcal B}})\ \text{tr}_{J'}[\Upsilon_{\mathcal B}(\mathbf s_J \oplus \mathbf s_{B\setminus J})^{T}\Phi(\mathbf r_{\mathcal I_{\mathcal A}}) \otimes \Phi(\mathbf s_{\mathcal I_{\mathcal B}})^{T} \Upsilon_{\mathcal A}(\mathbf r_{\mathcal A})] \\
&= \lim_{\mathbf s_{\mathcal B}, s_{\mathcal I_{\mathcal B}} \to\infty} K(\mathbf s_J \oplus \mathbf s_{B\setminus J} \oplus s_{\mathcal I_{\mathcal B}})\ \text{tr}_{J'}[[\Upsilon_{\mathcal B}(\mathbf s_J \oplus \mathbf s_{B\setminus J})^{T_{J_3}}\Phi(\mathbf r_{\mathcal I_{\mathcal A}})]^{T_J} \otimes \Phi(\mathbf s_{\mathcal I_{\mathcal B}})^{T_{J_2}} \Upsilon_{\mathcal A}(\mathbf r_{\mathcal A})] \\
&= \lim_{\mathbf s_J\to\infty} K(\mathbf s_J)\ \text{tr}_{J}[\lim_{\mathbf s_{B\setminus J}\to\infty} K(\mathbf s_{B\setminus J})\ \text{tr}_{J_3}[\Upsilon_{\mathcal B}(\mathbf s_J \oplus \mathbf s_{B\setminus J})^{T_{J_3}}\Phi(\mathbf r_{\mathcal I_\mathcal A})]^{T_J} \otimes \lim_{\mathbf s_{\mathcal I_{\mathcal B}}\to\infty} K(\mathbf s_{\mathcal I_{\mathcal B}})\ \text{tr}_{J_2}[\Phi(\mathbf s_{\mathcal I_{\mathcal B}})^{T_{J_2}} \Upsilon_{\mathcal A}(\mathbf r_{\mathcal A})]] \\
&= \lim_{\mathbf s_J\to\infty} K(\mathbf s_J)\ \text{tr}_{J}[\Upsilon_{\mathcal B}(\mathbf r_{\mathcal I_{\mathcal A}} \oplus \mathbf s_{J})^{T_J} \Upsilon_\mathcal A(\mathbf r_{\mathcal A})] \\
&= \lim_{\mathbf s_J\to\infty} K(\mathbf s_J)\ \text{tr}_{J}[\Upsilon_{\mathcal B}(\mathbf s)^{T_J} \Upsilon_\mathcal A(\mathbf r)] = \lim_{\mathbf s_J\to\infty} K(\mathbf s_J)\ \text{tr}_{J}[\Upsilon_\mathcal A(\mathbf r)^{T_J} \Upsilon_{\mathcal B}(\mathbf s)]
\end{align}

where $T_J$ is the partial transpose on $\bigotimes_{i\in {I_J}} \mathcal H^i$ (which is the space equivalent to $\bigotimes_{o\in {O_{J}}} \mathcal H^o$ in the stitching). From (C29) to (C30), we used our formalized observations using Lemma 2 to simplify. From (C30) to (C31), we relabel the squeezing parameters, dropping the subscripts for brevity. In (C31), both orders of expression are equivalent since $\Upsilon_{\mathcal B}\star_J \Upsilon_{\mathcal A} = \Upsilon_{\mathcal A}\star_J \Upsilon_{\mathcal B}$ and $T_J$ acts on the same subspaces.

\end{proof}

\section{Gaussian Link Product Integral}
\label{proof: gaussian-link-prod-int}

Recall that a Gaussian state $\rho_G$ can be expressed in terms of its characteristic function, as in (\ref{eqn: characteristic-func-gaussian-form}):

\begin{align}
     \chi_{\rho_G}(\vec\alpha) = e^{-\frac{1}{4}\vec\alpha^T\Omega^T \Gamma \Omega\vec\alpha}e^{i\vec \alpha^T\Omega^T \vec{d}}
 \end{align}
 where $\vec{d} = (d_1,d_2,\ldots)^T$ is the vector of quadrature means and $\Gamma$ (with matrix elements $\Gamma_{ij}$) is the associated covariance matrix~\cite{adesso2014continuous}.

A Gaussian state $\rho_G$ can also be expressed as a Gaussian integral of its characteristic function over its $n$ qumodes,~\cite{serafini2023quantum}

\begin{align}
    \rho_G 
    &= \frac{1}{(2\pi)^{n}} \int d\vec\alpha \,e^{-\frac{1}{4}\vec\alpha^T\Omega^T \Gamma \Omega\vec\alpha}e^{i\vec \alpha^T\Omega^T \vec{d}} \hat D(\vec\alpha) \\
    &= \frac{1}{(2\pi)^n}
    \label{eqn: gaussian-state-characteristic-func} \int d\vec\alpha \,\chi_{\rho_G}(\vec\alpha) \hat D(\vec\alpha)
\end{align}

where $\hat D(\vec \alpha) = \exp(i \hat{R} \Omega\vec\alpha)$ are the displacement operators, where $\hat R = (\hat X_1,\hat P_1,\dots, \hat X_n,\hat P_n)$ and $\vec\alpha = (x_1,p_1,\dots x_n,p_n)^T \in \mathbb R^{2n}$.
 
Therefore, for a circuit fragment $\mathcal A$ with $n_A$-input and $k_A$-output modes, its $(n_A+k_A)$-mode Choi state $\Upsilon_{\mathcal A}$, has a characteristic function, $\chi_{\Upsilon_{\mathcal A}}(\vec\alpha_A)$, where we have denoted specifically a subscript of $A$ on $\vec\alpha_A$ to identify the modes corresponding to $\mathcal A$. Similarly, the $n_B$-input and $k_B$-output fragment $\mathcal B$ has the $(n_B+k_B)$-mode Choi state $\Upsilon_{\mathcal B}$ with a characteristic function given by $\chi_{\Upsilon_{\mathcal B}}(\vec\alpha_B)$.
 
 Before moving further, we first elaborate further on some notation that will be used. For the purpose of evaluating the integrals in this and the following subsections, we can explicitly re-express $\vec\alpha_A = \vec\alpha_{A\setminus J} \oplus \vec\alpha_{J_A}$ and $\vec\alpha_B = \vec\alpha_{B\setminus J} \oplus \vec\alpha_{J_B}$. The quadratures $\vec\alpha_{A\setminus J}$ and $\vec\alpha_{B\setminus J}$ are associated with the set of unconnected modes $(I_\mathcal A \cup O_\mathcal A) \setminus J_\mathcal{A}$ of $\Upsilon_{\mathcal A}$ and $(I_\mathcal B \cup O_\mathcal B) \setminus J_\mathcal{B}$ of $\Upsilon_{\mathcal B}$ respectively. Meanwhile, $\vec\alpha_{J_A}$ and $\vec\alpha_{J_B}$ is associated with the set of connected modes ($J_\mathcal{A}$ and $J_\mathcal{B}$) of $\Upsilon_{\mathcal A}$ and $\Upsilon_{\mathcal B}$ respectively.
 
Then, using (\ref{eqn: gaussian-state-characteristic-func}) we have

\begin{align}
    \Upsilon_{\mathcal A}^{T_{J_{\mathcal A}}} &= \frac{1}{(2\pi)^{n_A+k_A}} \int d\vec\alpha \, \chi_{\Upsilon_{\mathcal A}^{T_{J_{\mathcal A}}}}(\vec\alpha_{A\setminus J} \oplus \vec\alpha_{J_A}) \hat D(\vec\alpha_{A\setminus J})\otimes \hat D(\vec\alpha_{J_A}) \\
    &= \frac{1}{(2\pi)^{n_A+k_A}} \int d\vec\alpha \, \chi_{\Upsilon_{\mathcal A}}(\vec\alpha_{A\setminus J}\oplus-\Lambda\vec\alpha_{J_A}) \hat D(\vec\alpha_{A\setminus J})\otimes \hat D(\vec\alpha_{J_A}).
\end{align}

In the first line, we see that $\Upsilon_{\mathcal A}^{T_{J_{\mathcal A}}}$ experiences a partial transpose on the modes $J_{\mathcal A}$, and we use the property of partial transpose of CV states (Appendix \ref{appendix:partial-transpose}) to re-express it as the characteristic function of the untransposed $\Upsilon_{\mathcal A}$ in the second line.

Additionally,

\begin{align}
    \Upsilon_{\mathcal B} &= \frac{1}{(2\pi)^{n_B + k_B}} \int d\vec\alpha \, \chi_{\Upsilon_{\mathcal B}}(\vec\alpha_{B\setminus J}\oplus\vec\alpha_{J_B}) \hat D(\vec\alpha_{B\setminus J})\otimes \hat D(\vec\alpha_{J_B}). 
\end{align}

Then inserting these into Result~\ref{jstitch}, one obtains 

\begin{align}
    &\Upsilon_{\mathcal A}(\mathbf r) \star_J \Upsilon_{\mathcal B} (\mathbf s) \propto \lim_{\mathbf s_J \to\infty} \trace_J [\Upsilon_{\mathcal A}^{T_{J_\mathcal A}} \Upsilon_{\mathcal B}] \\
    &\propto \lim_{\mathbf s_J \to\infty} \int d\vec{\alpha} \, \chi_{\Upsilon_{\mathcal A}}(\vec\alpha_{A \setminus J}\oplus-\Lambda\vec\alpha_{J_A}) \chi_{\Upsilon_{\mathcal B}}(\vec\alpha_{B \setminus J}\oplus \vec\alpha_{J_B}) \,\, \hat D(\vec\alpha_{A \setminus J}) \otimes \text{tr}_J [ \hat D(\vec\alpha_{J_A}) \hat D(\vec\alpha_{J_B})]\otimes \hat D(\vec\alpha_{B \setminus J}) \\
    &\propto \lim_{s_J \to\infty} \iiint d\vec\alpha_{A\setminus J} \, d\vec\alpha_{B\setminus J} \, d\vec\alpha_{J_B} \ \chi_{\Upsilon_{\mathcal A}}(\vec\alpha_{A \setminus J}\oplus\Lambda\vec\alpha_{J_B}) \chi_{\Upsilon_{\mathcal B}}(\vec\alpha_{B \setminus J}\oplus \vec\alpha_{J_B}) \,\, \hat D({\vec\alpha_{A \setminus J}}) \otimes \hat D({\vec\alpha_{B \setminus J}}) \\
    &\propto \lim_{s_J \to\infty} \int d\vec\alpha_{J_B} \iint d\vec\alpha_{A\setminus J} \, d\vec\alpha_{B\setminus J} \ \chi_{\Upsilon_{\mathcal A}}(\vec\alpha_{A \setminus J}\oplus\Lambda\vec\alpha_{J_B}) \chi_{\Upsilon_{\mathcal B}}(\vec\alpha_{B \setminus J}\oplus \vec\alpha_{J_B}) \,\, \hat D({\vec\alpha_{A \setminus J}}) \otimes \hat D({\vec\alpha_{B \setminus J}}),
    \label{eqn: gaussian-linkprod-integral-eval}
\end{align}

where from the (D7) to (D8), we use the property that $\trace_J{[\hat D(\vec\alpha_{J_A})\hat D(\vec\alpha_{J_B})]} \propto \delta(\vec\alpha_{J_A} + \vec\alpha_{J_B})$ to replace $\vec\alpha_{J_A} = -\vec\alpha_{J_B}$~\cite{serafini2023quantum}.

In (D10), what results is a multimode (multivariate) Gaussian integral on $\vec\alpha_{J_B}$ - the integral has been separated out in (\ref{eqn: gaussian-linkprod-integral-eval}) to denote this. By evaluating this integral on $\vec\alpha_{J_B}$, one obtains the Choi state resulting from the $J$-stitching - therefore allowing the Gaussian link product to be evaluated by evaluating the Gaussian integral in (\ref{eqn: gaussian-linkprod-integral-eval}), and the normalization constant can be found by taking the trace of the output state. Equivalently, the first and second moments of the output state can be directly read off from the form of the characteristic function after evaluating the aforementioned Gaussian integral.

For brevity, we may also choose to drop the subscript of $B$ in (D10) such that $\vec\alpha_{J_B} = \vec\alpha_J$.

\section{Full Algorithm for Gaussian Link Product }
\label{appendix: gaussianlinkprod-vo-means}

In the main text, we presented an algorithm to evaluate the covariance matrices of the link products. Here, we provide a complementary addendum for that algorithm to additionally evaluate the vector of means of the link products.

The vector of means can be evaluated concurrently in the overall algorithm, together with the covariance matrix. To do so, we introduce two additional subroutines: \texttt{PrepVoM}, prepares the vector of means for all $J$-stitchings at once, and \texttt{VConnect}, which returns the output vector of means after completing a single mode-stitch. 

The subroutine \texttt{VConnect} requires the covariance matrices $\gamma_{y_\ell}$ and $\gamma_{{\bar X},y_\ell}$ from the covariance matrix stitch (evaluated within \texttt{MConnect}). Therefore, we can call a modified subroutine, \texttt{MConnectFull}, which also returns $\gamma_{{\bar X},y_\ell}$ and $\gamma_{y_\ell}$ in addition to $\gamma_\text{tot}$.

Consider two processes, $\mathcal A, \mathcal B$ with Choi states $\Upsilon_\mathcal A, \Upsilon_\mathcal B$ respectively that we wish to stitch on modes $J$. Let $\vec d_{\mathcal A}, \vec d_{\mathcal B}$ be the vector of means for $\Upsilon_\mathcal A$ and $\Upsilon_\mathcal B$, respectively, and $\Gamma_{\mathcal A}, \Gamma_{\mathcal B}$ be the covariance matrices for $\Upsilon_\mathcal A$ and $\Upsilon_\mathcal B$, respectively. Then, the modified algorithm that will evaluate {\it both} the vector of means and the covariance matrix of the output state from the link product between $\Upsilon_\mathcal A$ and $\Upsilon_\mathcal B$ will include the three new subroutines:

\begin{algorithm}[H]
\SetKwFunction{multiLinkProd}{multimodeLinkProductFull}
\SetKwFunction{combMatrix}{\pcm}
\SetKwFunction{singleLinkProd}{MConnectFull}
\SetKwFunction{combVoM}{PrepVoM}
\SetKwFunction{singleLinkProdVM}{VConnect}
\SetKwProg{myalgo}{Algorithm}{}{}
\myalgo{\multiLinkProd{$\Gamma_{\mathcal A}, \Gamma_{\mathcal B},\vec d_{\mathcal A}, \vec d_{\mathcal B},J$}}{
  \nl $\Gamma =\Gamma_{\mathcal A}{\oplus}\Gamma_{\mathcal B}$ \;
  \nl $\vec d =\vec d_{\mathcal A}{\oplus}\vec d_{\mathcal B}$ \;
  \nl $\gamma_\text{tot} \gets$ \combMatrix{$\Gamma$, $J$} \;
  \nl $\vec d_\text{tot} \gets$ \combVoM{$\vec d$, $J$} \;
  \For{$y_\ell \in J$}{
    $\gamma_\text{tot}, \gamma_{{\bar X},y_\ell}, \gamma_{y_\ell} \gets$ \singleLinkProd{$\gamma_\text{tot}$, $y_\ell$} \;
    $\vec d_\text{tot} \gets$ \singleLinkProdVM{$\vec d_\text{tot}$, $\gamma_{{\bar X},y_\ell}$, $\gamma_{y_\ell}$, $y_\ell$} \;
    }
  \nl \KwRet $\gamma_\text{tot}, \vec d_\text{tot}$\;
}
\end{algorithm}

\textbf{Vector of means preparation.} \texttt{PrepVoM}, works similarly to \texttt{PrepCM}, and converts each physical output-input pair into a corresponding stitched variable. 

We begin with the combined vector of means of the two processes, $\vec d = \vec d_{\mathcal A} \oplus \vec d_{\mathcal B}$. The subroutine \texttt{PrepVoM} first reorders the modes of $\vec d$ in the exact same order as was done in \texttt{PrepCM}: all modes not eliminated by the link product appear first, followed by the $m$ modes of $J_\mathcal{A}$ (i.e., both input and output modes of $\mathcal A$ participating in the join), then the $m$ modes of $J_\mathcal{B}$, both in the reverse order of the pairs in $J$.

Then, $\vec d$ has the form, $\vec d = (\vec d_X,\ \vec d_C)^T$, where $\vec d_C$ is a $4m \times 1$ column vector of means of the $2m$ output-input qumodes involved in the stitch, and $\vec d_X$ is the vector of means of the unstitched modes. In this manner, just like in the main text, $X$ denote all modes not eliminated, and $C$ denotes the collection of mode pairs to be stitched. Then, the subroutine \texttt{PrepVoM} returns

\begin{align}
    \texttt{PrepVoM}(\vec d, J) &= \begin{pmatrix}
        \vec d_X \\
        Q_J^T\vec d_C
    \end{pmatrix}
\end{align}

where $Q_J = (-\Lambda_m, I_{2m})^T$, $I_{2m}$ is the identity matrix of size $2m\times 2m$, and $\Lambda_m = \bigoplus_{\ell=1}^m \mathrm{diag}(1,-1)$. Here, $Q_J^T$ has dimensions $2m\times4m$, so the resulting block $Q_J^T\vec d_C$ has dimensions $2m\times1$ and describes the $m$ effective stitched qumodes that are subsequently eliminated later.

Explicitly, consider that for our two processes $\mathcal A$ and $\mathcal B$, the vector of means are written in the form

\begin{align}
    \vec d_{\mathcal A} = \begin{pmatrix}
        \vec d_{A/J} \\
        \vec d_{J_A}
    \end{pmatrix}, \
    \vec d_{\mathcal B} = \begin{pmatrix}
        \vec d_{B/J} \\
        \vec d_{J_B}
    \end{pmatrix}, 
\end{align}

such that $\vec d_{J_A}$ denotes the vector of means of the modes of $\mathcal A$ being stitched, and $\vec d_{J_B}$ denotes the vector of means of the corresponding modes of $\mathcal B$, both ordered according to the elements of $J$ in reverse order.

Then \texttt{PrepVoM} returns

\begin{align}
    \texttt{PrepVoM}(\vec d, J) &= \begin{pmatrix}
        \vec d_{A/J} \\ \vec d_{B/J} \\ -\Lambda_{m} \vec d_{J_A} + \vec d_{J_B}
    \end{pmatrix}.
\end{align}

\textbf{Completing the stitch.} \texttt{MConnectFull} will evaluate the matrices $\gamma_\text{tot}, \gamma_{{\bar X},y_\ell}$ and $\gamma_{y_\ell}$ in the exact same manner as \texttt{\smc}. The only difference, as mentioned, is that \texttt{MConnectFull} will also return $\gamma_{{\bar X},y_\ell}$ and $\gamma_{y_\ell}$, in addition to $\gamma_\text{tot}$.

In a similar fashion to \texttt{MConnect}, \texttt{VConnect} completes the stitch for each $y_\ell$ (and evaluates the vector of means). It takes a vector $\vec d_\text{tot}$ and $y_\ell\in J$ that we wish to stitch. Because we have similarly listed our $y_\ell$ in a reverse order when re-ordering $\vec d$, then

\begin{align}
\vec d_\text{tot} = (\vec d_{\bar X}, \ \vec d_{y_\ell})^T,
\end{align}

where $\vec d_{y_\ell}$ is the $2\times 1$ column vector of means associated with the effective stitched variable to be eliminated, and $\vec d_{\bar X}$ corresponds to the vector of means of all remaining modes, including any other stitched variables in $J$ that have yet to be eliminated. 

The subroutine will return the vector of means after the $y$-stitch:

\begin{align}
    \texttt{VConnect}(\vec d_\text{tot}, y_\ell) &= \lim_{r_{y_\ell} \to \infty} [\vec d_{\bar X}- \gamma_{{\bar X},y_\ell} \gamma_{y_\ell}^{-1} \vec d_{y_\ell} ]
\end{align}

In the same manner as the covariance matrix, taking the recursion for all $y_\ell \in J$ will evaluate the complete $J$-stitching in terms of the vector of means. Altogether, this modified version of the algorithm will return both of the first two moments, $\vec d_\text{tot}$ and $\gamma_\text{tot}$, of the resulting Choi state from the Gaussian Link product.

\section{Complexity of Computing Gaussian Link Products}
\label{subsec: complexity-analysis}
The algorithm \texttt{multimodeLinkProduct} returns a $O(\bar{m}-2m)\times O(\bar{m}-2m)$ covariance matrix and scales as $O(m\bar{m}^2)$, where $\bar{m}$ is the total number of modes of qumodes describing the circuit fragments to be $J$-stitched, and $m=|J|$ is the total number of mode-stitchings. 

In an example where a process with only one system ($S$) input mode per timestep, such as in Fig. \ref{fig:choi-construct}, and therefore only one input TMSV per timestep, this return a $2n\times2n$ matrix.

The algorithm first calls \texttt{\pcm}, which has complexity $O(\bar{m}^2)$. This subroutine consists of reordering and copying covariance-matrix entries, together with the multiplications $\gamma_{X,C}Q_J$ and $Q_J^T\gamma_CQ_J$. Since $Q_J$ is block diagonal with fixed-size blocks, $\gamma_{X,C}Q_J$ scales as $O(\bar{m}m)$ and $Q_J^T\gamma_CQ_J$ scales as $O(m^2)$. Both are bounded by $O(\bar{m}^2)$ for $m\leq O(\bar{m})$, and hence \texttt{\pcm} scales as $O(\bar{m}^2)$.

The algorithm then loops over $m$ effective stitched variables and applies \texttt{\smc} once per stitching. Each call to \texttt{\smc} scales as $O(\bar{m}^2)$. This is because the subroutine involves: (1) copying or reordering blocks of $\gamma_\text{tot}$ to obtain $\gamma_X,\gamma_{X,y},\gamma_y$, (2) inverting the $2\times2$ matrix $\gamma_y$, (3) multiplying $\gamma_{X,y}\gamma_y^{-1}\gamma_{X,y}^T$, and (4) adding the resulting matrices to return $\gamma_\text{res}$. Matrix addition and copy operations scale as $O(\bar{m}^2)$. The matrix inversion scales as $O(1)$ since $\gamma_y$ is always a $2\times2$ matrix. The matrix multiplication scales as $O(\bar{m}^2)$ because it involves the multiplication of an $O(\bar{m})\times2$ matrix with a $2\times O(\bar{m})$ matrix.

Taking into account the $m$ iterations, the total complexity is therefore
\[
    O(\bar{m}^2)+O(m\bar{m}^2)=O(m\bar{m}^2),
\]
for nonempty $J$. Thus, even for bidirectional stitchings, the algorithm remains polynomial in the number of timesteps and retains the exponential improvement over operating directly via full Choi-state representations.

The complexity still remains the same even when including the evaluation of the vector of means, since it only requires the copying, reordering, and multiplication of matrices of smaller dimensions (column vectors).

\section{Input to a Channel}

Consider Gaussian channel $\mathcal G$ with input mode $i$ and output mode $o$ that acts on a Gaussian state covariance matrix $\Gamma_\text{s}$. $\mathcal G$ will have a $4 \times 4$ Choi-covariance matrix of form
\begin{align}
\Gamma_{\mathcal G}(r) = \left( \begin{array}{c|c}
       \gamma_{i} & \gamma_{i,o} \\
       \hline
       \gamma_{i,o}^T & \gamma_{o}
    \end{array} \right),
\end{align}
Our goal is to find the covariance $\gamma_{\mathrm{res}}$ of the resulting state on input of state with covariance $\Gamma_s$.  \texttt{\pcm} then gives
\begin{align}
\gamma_\text{tot} = \left( \begin{array}{c|c}
        \gamma_{o} & \gamma_{i,o}^T \\
       \hline
        \gamma_{i,o} &  \Lambda \Gamma_\text{s}\Lambda + \gamma_{i}
    \end{array} \right),
\end{align}
where each sub-block is a $2\times2$ matrix, such that $\Lambda =\sigma_z$. Applying \texttt{\smc} for one stitch, we find that we have the output covariance matrix
\begin{align}
\Gamma_\text{out} = \lim_{r\rightarrow \infty} \left[\gamma_o - \gamma_{i,o}	^T(\Lambda \Gamma_\text{s}\Lambda+\gamma_i)^{-1} \gamma_{i,o}\right]
\end{align}
with $r$ is the squeezing of mode $i$, in agreement with previous  literature~\cite{giedke2002characterization,fiurasek_gaussian_2002}. 

\section{Example - Input vacuum state into squeezing channel}
\label{appendix: state-channel-join}

We can verify that the result from the previous section for $\gamma_\text{res}$ indeed recovers the action of $\mathcal G$ on $\rho_\text{in}$ with a specific example. Let $\rho_\text{in}$ be a vacuum state with covariance matrix $\Gamma_\text{in} = I_2$, and $\mathcal G$ the general squeezing channel with squeezing parameter $s$. The state will be squeezed in $\hat X(\theta/2)$; when $\theta = 0$, this corresponds to a squeezing in the $X$-quadrature, and anti-squeezing in the $P$-quadrature. In the main text, we have chose $\theta=\pi$ for the $P$-squeezer.

Then, $\mathcal G$'s Choi states have a covariance 

\begin{align}
\Gamma_{\mathcal G}(r') &= \left( \begin{array}{c|c}
       \cosh(r')F_{s,\theta}^2 & \sinh(r') F_{s,\theta}\Lambda \\
       \hline
       \sinh(r') \Lambda  F_{s,\theta} & \cosh(r') I_2
    \end{array} \right), \\
    &= \left( \begin{array}{c|c}
       \gamma_{o} & \gamma_{i,o}^T \\
       \hline
       \gamma_{i,o} & \gamma_{i}
    \end{array} \right),
\end{align}

where $F_{s,\theta}=\cosh(s)I_2 - \sinh(s) S_\theta, \Lambda=\sigma_z$ and

\begin{align}
S_\theta = \begin{pmatrix} \cos\theta & \sin\theta \\ \sin\theta & -\cos\theta
\end{pmatrix}.
\end{align}

Then, using the link product algorithm (i.e., by substituting the corresponding block matrices from $\Gamma_\mathcal G(r')$ into the result from the previous section,

\begin{align}
\Gamma_\text{out} = \gamma_o - \gamma_{i,o}	^T(\Lambda \Gamma_\text{in}\Lambda+\gamma_i)^{-1} \gamma_{i,o},
\end{align}

and evaluating the output) we will obtain

\begin{align}
\Gamma_\text{out} &= \cosh r' F_{s,\theta}^2 \nonumber \\
&- \sinh r' F_{s,\theta}\Lambda(\cosh r' I_2 + \Lambda\Gamma_\text{in}\Lambda)^{-1} \sinh  r' \Lambda F_{s,\theta} \\
\lim_{r'\to\infty} \Gamma_\text{out} &\approx \lim_{r'\to\infty} \big(\cosh r' F_{s,\theta}^2 - \sinh^2 r' F_{s,\theta}^2 \left[\frac{1}{\cosh r'}-\frac{1}{\cosh^2 r'} \right]\big) \\
&= \lim_{r'\to\infty} \tanh^2 r' F_{s,\theta}^2 = F_{s,\theta}^2
\end{align}

and notably, $F_{s,\theta}^2$ is precisely the covariance matrix of a squeezed single mode vacuum state. Thus, it can be verified that we recover the covariance matrix of the vacuum state through a squeezing channel.

\section{Example - Concatenate 2 single-mode channels (general Gaussian channels)}
\label{appendix: general-chanchan-join}

Here we elaborate further on evaluating a general version of the channel-channel join example referenced in text. Consider two Gaussian channels, $\mathcal A, \mathcal B$ with Choi states $\Upsilon_{\mathcal A}$ (with input and output modes ${i_1^{\mathcal A},o_1^{\mathcal A}}$ respectively) and $\Upsilon_{\mathcal B}$ (with input and output modes ${i_1^{\mathcal B},o_1^{\mathcal B}}$ respectively) as in Fig. \ref{fig:link-prod-gen2}. In the figure, we have dropped the subscript `1' from the labels for visual clarity. Their covariance matrices are given by the following block matrices,

\begin{align}
    \Gamma_{\mathcal A} = \begin{pmatrix}
        \gamma_{i_1^{\mathcal A}} & \gamma_{i_1^{\mathcal A},o_1^{\mathcal A}} \\
        \gamma_{i_1^{\mathcal A},o_1^{\mathcal A}}^T & \gamma_{o_1^{\mathcal A}}
    \end{pmatrix}, \,
    \Gamma_{\mathcal B} = \begin{pmatrix}
        \gamma_{i_1^{\mathcal B}} & \gamma_{i_1^{\mathcal B},o_1^{\mathcal B}} \\
        \gamma_{i_1^{\mathcal B},o_1^{\mathcal B}}^T & \gamma_{o_1^{\mathcal B}}
    \end{pmatrix},
    \label{eqn:cov-matrix-channels}
\end{align}

with submatrices of dimension $2\times 2$. Suppose we wish to connect $o_1^{\mathcal A}$ of $\mathcal A$ to $i_1^{\mathcal B}$ of $\mathcal B$ via the link product using the algorithm. We first carry out \texttt{\pcm} subroutine with $\Gamma_{\mathcal A}, \Gamma_{\mathcal B}$, and obtain

\begin{align}
    \gamma_\text{tot} = \texttt{PrepCM}(\Gamma_{\mathcal A}\oplus\Gamma_{\mathcal{B}}, J) &=  \left( \begin{array}{cc|c}
        \gamma_{i_1^{\mathcal A}} & \mathbf 0 & -\gamma_{i_1^{\mathcal A},o_1^{\mathcal A}}\Lambda \\
        \mathbf 0 & \gamma_{o_1^{\mathcal B}} & \gamma_{i_1^{\mathcal B},o_1^{\mathcal B}}^T \\
        \hline
        -\Lambda\gamma_{i_1^{\mathcal A},o_1^{\mathcal A}}^T & \gamma_{i_1^{\mathcal B},o_1^{\mathcal B}} & \Lambda\gamma_{o_1^{\mathcal A}}\Lambda + \gamma_{i_1^{\mathcal B}}
    \end{array} \right) \\
    &= \left( \begin{array}{c|c}
        \gamma_{X} & \gamma_{X,y} \\
        \hline
        \gamma_{X,y}^T & \gamma_{y}
    \end{array} \right).
\end{align}

where $\Lambda = \sigma_z$. Then, the output Choi covariance from the link product algorithm is given by

\begin{align}
    \Gamma_\text{out} &= (\gamma_{i_1^{\mathcal A}} \oplus \gamma_{o_1^{\mathcal B}}) \nonumber \\&- \begin{pmatrix}-\gamma_{i_1^{\mathcal A},o_1^{\mathcal A}}\Lambda \\ \gamma_{i_1^{\mathcal B},o_1^{\mathcal B}}^T \end{pmatrix} (\Lambda \gamma_{o_1^{\mathcal A}}\Lambda+\gamma_{i_1^{\mathcal B}})^{-1} \begin{pmatrix}-\Lambda\gamma_{i_1^{\mathcal A},o_1^{\mathcal A}}^T & \gamma_{i_1^{\mathcal B},o_1^{\mathcal B}} \end{pmatrix}
    \label{eqn: chanchan-gen-result}
\end{align}

\subsection{Correspondence with Gaussian Link Product Integral}

Going a step further, we can show that the result obtained by using the link product algorithm yields exactly the same result as if we were to evaluate the integral in (\ref{eqn: gaussian-linkprod-integral-eval}), and that this accurately provides an expression for the resulting combination of two channels. Below, we proceed with evaluating the integrals.

Since the integral involves both the covariance matrix and vector of means, to verify and compare, we also evaluate the vector of means in this example. Consider $\Upsilon_{\mathcal A}, \Upsilon_{\mathcal B}$ with vector of means

\begin{align}
    \vec d_{\mathcal A} &= \begin{pmatrix}\vec d_{A \setminus J} \\ \vec d_{J_A}\end{pmatrix} = \begin{pmatrix}\vec d_{i_1^{\mathcal A}} \\ \vec d_{o_1^{\mathcal A}}\end{pmatrix} \\
    \vec d_{\mathcal B} &= \begin{pmatrix}\vec d_{B \setminus J'} \\ \vec d_{J_\mathcal{B}}\end{pmatrix} = \begin{pmatrix}\vec d_{o_1^{\mathcal B}} \\ \vec d_{i_1^{\mathcal B}} \end{pmatrix}
    \label{eqn:channels-vectorofmeans}
\end{align}

Using $\Gamma_{\mathcal A}$ from (\ref{eqn:cov-matrix-channels}) and the result of partial transposing a CV state, (\ref{eqn: partial-transpose-cm}), the covariance matrix of $\Upsilon_{\mathcal A}^{T_{J_{\mathcal A}}}$ is given by

\begin{align}
    \Gamma_{\mathcal A^{T_J}} = 
    \begin{pmatrix}
        \gamma_{i_1^{\mathcal A}} & \gamma_{i_1^{\mathcal A},o_1^{\mathcal A}} \Lambda \\
        \Lambda \gamma_{i_1^{\mathcal A},o_1^{\mathcal A}}^T & \Lambda\gamma_{o_1^{\mathcal A}}\Lambda
    \end{pmatrix}
    \label{eqn:cov-matrix-channels-A-transpose}
\end{align}

Similarly, the vector of means of $\Upsilon_A^{T_{J_\mathcal A}}$ is given by $\vec{d}_{\mathcal A^{T_J}} = (\vec d_{A\setminus J}, \Lambda \vec d_{J_\mathcal{A}})^T$. Then, putting $\Gamma_{\mathcal B}$, $\Gamma_{\mathcal A^{T_J}}$, $\vec d_{\mathcal B}$ and $\vec{{ d_{\mathcal A^{T_J}}}}$ into (\ref{eqn: gaussian-linkprod-integral-eval}), one obtains

\begin{align}
\Upsilon_{\mathcal A} \star \Upsilon_{\mathcal B} &\propto \lim_{r_j\to\infty, j\in I_J} \int d\vec\alpha_J \int d\vec\alpha_{A \setminus J} \int d\vec\alpha_{B \setminus J} \, \hat D(\vec\alpha_{A \setminus J}) \otimes \hat D(\vec\alpha_{B \setminus J}) \nonumber \\
&\hspace{25pt} \exp\left[{-\frac{1}{4} \begin{pmatrix} \vec\alpha_{A \setminus J}^T& \vec\alpha_{B \setminus J}^T \end{pmatrix} \Omega^T \begin{pmatrix} \gamma_{i_1^{\mathcal A}} & \mathbf 0 \\ \mathbf 0 & \gamma_{o_1^{\mathcal B}}  \end{pmatrix} \Omega \begin{pmatrix} \vec\alpha_{A \setminus J} \\ \vec\alpha_{B \setminus J} \end{pmatrix}}\right] \nonumber \\
&\hspace{16pt} \cdot \,\exp\left[{-\frac{1}{4} \begin{pmatrix} \vec\alpha_{A \setminus J}^T& \vec\alpha_{B \setminus J}^T \end{pmatrix} \Omega^T \begin{pmatrix} -\gamma_{i_1^{\mathcal A},o_1^{\mathcal A}}\Lambda \\ \gamma^T_{i^{\mathcal B}_1,o^{\mathcal B}_1} \end{pmatrix} \Omega \vec\alpha_{J}}\right] \nonumber \\
&\hspace{16pt} \cdot \, \exp\left[{-\frac{1}{4} \vec\alpha_{J}^T \Omega^T \begin{pmatrix} -\Lambda\gamma_{i_1^{\mathcal A},o_1^{\mathcal A}}^T & \gamma_{i^{\mathcal B}_1,o^{\mathcal B}_1} \end{pmatrix} \Omega \begin{pmatrix} \vec\alpha_{A \setminus J} \\ \vec\alpha_{B \setminus J} \end{pmatrix}}\right] \nonumber \\
&\hspace{16pt} \cdot \, \exp\left[{-\frac{1}{4} \vec\alpha_{J}^T \Omega^T (\Lambda\gamma_{o_1^{\mathcal A}}\Lambda + \gamma_{i^{\mathcal B}_1}) \Omega \vec\alpha_{J}}\right] \nonumber \\
&\hspace{16pt} \cdot \exp\left[{i(\vec\alpha_{A \setminus J}^T \Omega^T \vec d_{A \setminus J} + \vec\alpha_{B \setminus J}^T \Omega^T \vec d_{B \setminus J} )}\right] \cdot \exp\left[{i\vec\alpha_{J}^T \Omega^T (-\Lambda \vec d_{J_\mathcal{A}} + \vec d_{J_\mathcal{B}}) }\right]. 
\end{align}

Although not noted explicitly for brevity, the dimensions of $\Omega$ changes according to the size of the matrices it is multiplied by.

We note the introduction of the negative signs to $\gamma_{i_1^{\mathcal A},o_1^{\mathcal A}}\Lambda, \Lambda \gamma_{i_1^{\mathcal A},o_1^{\mathcal A}}^T$ and $\Lambda d_{J_\mathcal{A}}$ are due to the substitution of $\vec\alpha_{J_{\mathcal A}}$ with $-\vec\alpha_{J_{\mathcal B}} = -\vec\alpha _J$ in (\ref{eqn: gaussian-linkprod-integral-eval}). The integral with respect to $\vec\alpha_J$ is simply a multi-variate Gaussian integral that can be solved with
$$
\int d\vec{\alpha} \, e^{-\vec\alpha^T A \vec\alpha - \vec\alpha^T b} = \frac{\pi^{d/2}}{\sqrt{\det A}}  e^{\frac{1}{4} b^T A^{-1} b},
$$

where $A$ is a $d\times d$ matrix. Additionally, we note that the terms

\begin{align}
    { \begin{pmatrix} \vec\alpha_{A \setminus J}^T& \vec\alpha_{B \setminus J}^T \end{pmatrix} \Omega^T \begin{pmatrix} -\gamma_{i_1^{\mathcal A},o_1^{\mathcal A}}\Lambda \\ \gamma^T_{i^{\mathcal B}_1,o^{\mathcal B}_1} \end{pmatrix} \Omega \vec\alpha_{J}} &= \vec\alpha_{J}^T \Omega^T \begin{pmatrix} -\Lambda\gamma_{i_1^{\mathcal A},o_1^{\mathcal A}}^T & \gamma_{i^{\mathcal B}_1,o^{\mathcal B}_1} \end{pmatrix} \Omega \begin{pmatrix} \vec\alpha_{A \setminus J} \\ \vec\alpha_{B \setminus J} \end{pmatrix},
\end{align}

allowing us to combine those terms. Therefore, evaluating the integral with respect to $\vec\alpha_J$, one obtains,

\begin{align}
\Upsilon_{\mathcal A} \star \Upsilon_{\mathcal B} &\propto \lim_{\mathbf r_J \to \infty} \int d\vec\alpha_{A \setminus J} \int d\vec\alpha_{B \setminus J} \, \hat D(\vec\alpha_{A \setminus J}) \otimes \hat D(\vec\alpha_{B \setminus J}) \nonumber \\
&\hspace{25pt} \exp\Bigg[-\frac{1}{4} \begin{pmatrix} \vec\alpha_{A \setminus J}^T& \vec\alpha_{B \setminus J}^T \end{pmatrix} \Omega^T  \begin{pmatrix} \gamma_{i_1^{\mathcal A}} & \mathbf 0 \\ \mathbf 0 & \gamma_{o_1^{\mathcal B}}  \end{pmatrix} \Omega \begin{pmatrix}\vec\alpha_{A \setminus J} \\ \vec\alpha_{B \setminus J} \end{pmatrix} \Bigg] \nonumber \\
&\hspace{16pt} \cdot \,\exp\Bigg[\frac{1}{4} \begin{pmatrix} \vec\alpha_{A \setminus J}^T& \vec\alpha_{B \setminus J}^T \end{pmatrix} \Omega^T \begin{pmatrix} -\gamma_{i_1^{\mathcal A},o_1^{\mathcal A}}\Lambda \\ \gamma^T_{i^{\mathcal B}_1,o^{\mathcal B}_1} \end{pmatrix} (\Lambda\gamma_{o_1^{\mathcal A}}\Lambda + \gamma_{i^{\mathcal B}_1})^{-1} \begin{pmatrix} -\Lambda\gamma_{i_1^{\mathcal A},o_1^{\mathcal A}}^T & \gamma_{i^{\mathcal B}_1,o^{\mathcal B}_1} \end{pmatrix} \Omega \begin{pmatrix} \vec\alpha_{A \setminus J} \\ \vec\alpha_{B \setminus J} \end{pmatrix} \Bigg] \nonumber \\
&\hspace{16pt} \cdot \exp \Big[{i(\vec\alpha_{A \setminus J}^T \Omega^T \vec d_{A \setminus J} + \vec\alpha_{B \setminus J} \Omega^T \vec d_{B \setminus J} )}\Big] \nonumber \\
&\hspace{16pt} \cdot \exp \Big[i\begin{pmatrix} \vec\alpha_{A \setminus J}^T& \vec\alpha_{B \setminus J}^T \end{pmatrix} \Omega^T  \begin{pmatrix} -\gamma_{i_1^{\mathcal A},o_1^{\mathcal A}}\Lambda \\ \gamma^T_{i^{\mathcal B}_1,o^{\mathcal B}_1} \end{pmatrix} (\Lambda\gamma_{o_1^{\mathcal A}}\Lambda + \gamma_{i^{\mathcal B}_1})^{-1}  (\Lambda \vec d_{J_\mathcal{A}} - \vec d_{J_\mathcal{B}}) \Big]
\label{eqn:link-prod-int-evaluated-channels}
\end{align}

Using $\Omega^T = \Omega^{-1}$. Notably, (\ref{eqn:link-prod-int-evaluated-channels}) is precisely of the form (\ref{eqn: gaussian-state-characteristic-func}) for an output state on modes $i_1^{\mathcal A},o_1^{\mathcal B}$, and its covariance matrix can be directly read off and verified to agree with (\ref{eqn: chanchan-gen-result}). We have therefore shown agreement between the evaluation via the integrals and the link product algorithm provided here.

\subsection{Correspondence with State-into-channel Evaluation}

Next, we wish to show that this resulting covariance matrix correctly represents the output state of feeding a single TMSV copy, $\Phi(r)$ through $\mathcal B \circ \mathcal A$. In this case where $\mathcal A, \mathcal B$ are both Gaussian channels, their Choi states are characterized by

\begin{align}
    \vec d_{\mathcal A} &= \begin{pmatrix}\vec{0} \\  \vec\nu_A \end{pmatrix} \\
    \vec d_{\mathcal B} &= \begin{pmatrix}\vec \nu_B \\ \vec{ 0}\end{pmatrix}
\end{align}

and

\begin{align}
\Gamma_{\mathcal A} (r) &= \begin{pmatrix}
        \gamma_{i_1^{\mathcal A}} & \gamma_{i_1^{\mathcal A},o_1^{\mathcal A}} \\
        \gamma_{i_1^{\mathcal A},o_1^{\mathcal A}}^T & \gamma_{o_1^{\mathcal A}}
    \end{pmatrix} = \begin{pmatrix}
\cosh r\mathbb{I} & \sinh r \Lambda X_A \\
\sinh r X_A^T\Lambda & \cosh r X_A^TX_A +Y_A
\end{pmatrix} \label{eqn:channels-cm-full-A} \\
\Gamma_{\mathcal B} (r') &= \begin{pmatrix}
        \gamma_{o_1^{\mathcal B}} & \gamma_{i_1^{\mathcal B},o_1^{\mathcal B}}^T \\
        \gamma_{i_1^{\mathcal B},o_1^{\mathcal B}} & \gamma_{i_1^{\mathcal B}}
    \end{pmatrix} = \begin{pmatrix}
\cosh r' X_B^TX_B +Y_B & \sinh r' X_B^T \Lambda  \\
\sinh r' \Lambda X_B &\cosh r'\mathbb{I}
\end{pmatrix}
\label{eqn:channels-cm-full-B}
\end{align}

where $(X_A,Y_A,\vec\nu_A),(X_B,Y_B,\vec\nu_B)$ are the transformation matrices corresponding to channels $\mathcal A, \mathcal B$ respectively. Here, the mode order for $\vec{d}_{\mathcal A}, \Gamma_{\mathcal A}$ is given by $(i_1^{\mathcal A}, o_1^{\mathcal A})$, while that for $\vec{d}_{\mathcal B}, \Gamma_{\mathcal B}$ is given by $(o_1^{\mathcal B}, i_1^{\mathcal B})$. Therefore, the covariance matrix of the TMSV state through $\mathcal B\circ\mathcal A$ is given by

\begin{align}
\Gamma_{\mathcal B \circ \mathcal A} = \begin{pmatrix}
\cosh r\mathbb{I} & \sinh r \Lambda X_A X_B \\
\sinh r X_B^TX_A^T\Lambda & X_B^T(\cosh r X_A^TX_A +Y_A)X_B+Y_B
\end{pmatrix}
\label{eqn:link-prod-channels-transformation-matrices}
\end{align}

Following the mode order of $(i_1^{\mathcal A}, o_1^{\mathcal B})$. We want to show that when substituting (\ref{eqn:channels-cm-full-A}) and (\ref{eqn:channels-cm-full-B}) into $\Gamma_\text{out}$ in (\ref{eqn: chanchan-gen-result}) it reduces to $\Gamma_{\mathcal B \circ \mathcal A}$ (\ref{eqn:link-prod-channels-transformation-matrices}).

Here, our $\set{r_j : j \in I_J} =  \set{r'}$. Thus, under $r' \to\infty$, we can use the convergence of a geometric series of matrices to take the following approximation:

\begin{align}
\lim_{r'\to\infty}(\Lambda\gamma_{o_1^{\mathcal A}}\Lambda+\gamma_{i_1^{\mathcal B}})^{-1} &= \lim_{r'\to\infty} [\Lambda (\cosh r' X_A^T X_A +Y_A )\Lambda + \cosh r' I_2]^{-1} \\
&\approx \lim_{r'\to\infty} \frac{1}{\cosh r'} \left[I_2 - \frac{\Lambda (\cosh r X_A^T X_A +Y_A )\Lambda}{\cosh r'}\right] \label{eqn:matrix-inversion-approx}
\end{align}

We can then substitute (\ref{eqn:matrix-inversion-approx}) into both conditions. Taking the limit of $r'\to\infty$ we find indeed that for the covariance matrix,

\begin{align}
\lim_{r'\to\infty} \gamma_{i_1^{\mathcal A}}-\gamma_{i_1^{\mathcal A},o_1^{\mathcal A}}\Lambda (\Lambda\gamma_{o_1^{\mathcal A}}\Lambda+\gamma_{i_1^{\mathcal B}})^{-1}\Lambda \gamma_{i_1^{\mathcal A},o_1^{\mathcal A}}^T &= \cosh r I_2, \\
\lim_{r'\to\infty} \gamma_{o_1^{\mathcal B}}-\gamma_{i_1^{\mathcal B},o_1^{\mathcal B}}^T (\Lambda\gamma_{o_1^{\mathcal A}}\Lambda+\gamma_{i_1^{\mathcal B}})^{-1} \gamma_{i_1^{\mathcal B},o_1^{\mathcal B}} \nonumber
&= X_B^T(\cosh r X_A^TX_A +Y_A)X_B+Y_B \\
\lim_{r'\to\infty} \quad\gamma_{i_1^{\mathcal A},o_1^{\mathcal A}} \Lambda (\Lambda\gamma_{o_1^{\mathcal A}}\Lambda+\gamma_{i_1^{\mathcal B}})^{-1} \gamma_{i_1^{\mathcal B},o_1^{\mathcal B}} &= \sinh r \Lambda X_A X_B, 
\end{align}

which are the desired results. We can also consider the vector of means. A TMSV with initial $\vec{d}_\text{s} = \vec 0$ transformed by $\mathcal B\circ\mathcal A$ has an output vector of means,

\begin{align}
    \vec {d}_{\mathcal B\circ\mathcal A} &= (\vec{0}, \, X_B^T\vec\nu_A + \vec\nu_B)^T \label{eqn:linkprod-chan-trans-vecofmeans},
\end{align}

The vector of means calculated from the link product, when substituting (\ref{eqn:channels-vectorofmeans}) into (\ref{eqn:link-prod-int-evaluated-channels}) yields

\begin{equation}\label{eqn:vecofmeans-linkprod-form}
    \begin{pmatrix}
        \vec{d}_{i_1^{\mathcal A}} -\gamma_{i_1^{\mathcal A},o_1^{\mathcal A}}\Lambda (\Lambda\gamma_{o_1^{\mathcal A}}\Lambda + \gamma_{i^{\mathcal B}_1})^{-1}(\Lambda \vec{d}_{o_1^{\mathcal A}}-\vec{d}_{i_1^{\mathcal B}}) \\
        \vec{d}_{o_1^{\mathcal B}} + \gamma_{i_1^{\mathcal B},o_1^{\mathcal B}}^T (\Lambda\gamma_{o_1^{\mathcal A}}\Lambda + \gamma_{i^{\mathcal B}_1})^{-1}(\Lambda \vec{d}_{o_1^{\mathcal A}}-\vec{d}_{i_1^{\mathcal B}})
    \end{pmatrix}
\end{equation}

We can show, using (I11)-(I12), that for the terms in (\ref{eqn:vecofmeans-linkprod-form}),

\begin{align}
    \lim_{r'\to\infty} -\gamma_{i_1^{\mathcal A},o_1^{\mathcal A}}\Lambda (\Lambda\gamma_{o_1^{\mathcal A}}\Lambda + \gamma_{i^{\mathcal B}_1})^{-1}(\Lambda \vec d_{o_1^{\mathcal A}}-\vec d_{i_1^{\mathcal B}}) = 0,
\end{align}

which means that first term of (\ref{eqn:vecofmeans-linkprod-form}) correctly reduces to the first term of (\ref{eqn:linkprod-chan-trans-vecofmeans}) to yield $\vec d_{i_1^{\mathcal A}} = 0$.

Given $\vec d_{o_1^{\mathcal B}} = X_B^T\vec d_{o_1^{\mathcal B}}+\vec\nu_B$, the second term of (\ref{eqn:vecofmeans-linkprod-form}):

\begin{align}
    \lim_{r'\to\infty} \vec d_{o_1^{\mathcal B}} + \gamma_{i_1^{\mathcal B},o_1^{\mathcal B}}^T (\Lambda\gamma_{o_1^{\mathcal A}}\Lambda + \gamma_{i^{\mathcal B}_1})^{-1}(\Lambda \vec d_{o_1^{\mathcal A}}-\vec d_{i_1^{\mathcal B}}) &= X_B^T(X_A^T \vec{d}_{o_1^{\mathcal A}}+\nu_A) +\nu_B + X_B^T\vec{d}_{o_1^{\mathcal B}} - \vec{d}_{i^{\mathcal B}_1}\\
    &= X_B^T(X_A^T \vec{d}_{o_1^{\mathcal A}}+\nu_A) +\nu_B \\
    &= X_B^T \nu_A +\nu_B
\end{align}

The additional term $X_B^T\vec{d}_{o_1^{\mathcal B}} - \vec{d}_{i^{\mathcal B}_1}$ from the link product evaluation goes to 0, by recalling that we use a $\Phi(r')$ centered on $\vec d_{\mathcal B} = \vec 0$ such that $X_B^T\vec{d}_{o_1^{\mathcal B}} - \vec{d}_{i^{\mathcal B}_1} = 0$. Similarly, $\vec d_{o_1^{\mathcal A}} = \vec 0$. Consequently, we can see that vector of means evaluated via the link product ultimately yields the correct vector of means (checked against the direct evaluation of $\vec d_{\mathcal B \circ \mathcal A}$ above).

\section{Example - Concatenate 2 single-mode channels (2 squeezing channels)}
\label{appendix:squeezing-chanchan-join}
In the previous section, we evaluated the concatenation of general channels $\mathcal A$ and $\mathcal B$. Here, we set $\mathcal A$ and $\mathcal B$ to be squeezing channels given by their symplectic transformation matrix $F_{s,\theta}=\cosh(s)I_2 - \sinh(s) S_\theta$ and

\begin{align}
S_\theta = \begin{pmatrix} \cos\theta & \sin\theta \\ \sin\theta & -\cos\theta
\end{pmatrix},
\end{align}

where $s$ is the squeezing magnitude and $\theta$ is the phase. In-text, we have set $\theta = \pi$, which corresponds to a squeezing along the $P$-quadrature.

Explicitly, we use \texttt{\pcm}$(\Gamma_{\mathcal A}(r)\oplus\Gamma_{\mathcal{B}}(r'), J=\set{(o_1^{\mathcal A},i_1^{\mathcal B})})$ with the following covariance matrices with $\Lambda = \sigma_z$:

\begin{align}
\Gamma_{\mathcal A}(r) = \Gamma_{\mathcal B}(r) = \left( \begin{array}{c|c}
       \cosh(r) I_2 & \sinh(r) \Lambda F_{s,\theta} \\
       \hline
       \sinh(r) F_{s,\theta} \Lambda   & \cosh(r) F_{s,\theta}^2
    \end{array} \right),
\end{align}

where the covariance matrix has mode ordering corresponding to $(i_1^{\mathcal A},o_1^{\mathcal A})$ for $\Gamma_{\mathcal A}(r)$ and $(i_1^{\mathcal B},o_1^{\mathcal B})$ for $\Gamma_{\mathcal B}(r')$. The subroutine will evaluate $\Gamma=\Gamma_{\mathcal A}(r) \oplus \Gamma_{\mathcal B}(r')$, reorder the modes and evaluate the effective mode-pair stitches such that we get

\begin{align}
    \gamma_\text{tot} &= \texttt{PrepCM}(\Gamma_{\mathcal A} \oplus\Gamma_{\mathcal{B}}, J) = \left( \begin{array}{cc|c}
        \cosh(r) I_2 & \mathbf 0 & \sinh(r)\Lambda F_{s,\theta} \\
        \mathbf 0 & \cosh(r') F_{s,\theta}^2 & -\sinh(r')F_{s,\theta}\Lambda^2 \\
        \hline
        \sinh(r) F_{s,\theta}\Lambda & -\sinh(r') \Lambda^2 F_{s,\theta} & \cosh(r) F_{s,\theta}^2 + \cosh(r') \Lambda^2
    \end{array} \right),
\end{align}

which can be simplified with $\Lambda^2 = I_2$. The output $\gamma_\text{tot}$ has the mode order of $(i_1^{\mathcal A}, o_1^{\mathcal B}, y)$, where $y$ corresponds to the stitched pair (i.e., contains the terms of the modes from $o_1^{\mathcal A}$ and $i_1^{\mathcal B}$ that are to be stitched with each other). As was done in the previous section, take the following approximation as $r'\to\infty$:

\begin{align}
&\lim_{r'\to\infty}{(\cosh(r')I_2 + \cosh(r)F_{s,\theta}^2)}^{-1} \\
&\quad\approx \lim_{r'\to\infty} \frac{1}{\cosh r'} \left[I_2 - \frac{\cosh(r)F^2_{s,\theta}}{\cosh r'}\right]
\end{align}

and evaluate the output covariance matrix, $\Gamma_\text{out}$, as per (I4) in the previous section. This output covariance matrix will be ordered corresponding to the mode order of $(i_1^{\mathcal A}, o_1^{\mathcal B})$. For the final form of $\Gamma_\text{out}$ in-text, we set $\theta=\pi$ for a P-squeezer, and switched our qumode order to $(o_1^{\mathcal B}, i_1^{\mathcal A})$ such that the corresponding $\Gamma_\text{out}$ written in the main text is ordered accordingly:

\begin{align}
    \Gamma_\text{out} &= \begin{pmatrix}
\cosh r F_{s}^4 & \sinh r F_{s}^2 \Lambda \\
 \sinh r \Lambda F_{s} ^2 &  \cosh r I_2
\end{pmatrix} \\
&= \begin{pmatrix}
\cosh r F_{2s}^2 & \sinh r F_{2s} \Lambda \\
 \sinh r \Lambda F_{2s}  &  \cosh r I_2
\end{pmatrix},
\end{align}

and in the last line, we use $F_{s}^2 = F_{2s}$ \cite{weedbrook_gaussian_2012} in the simplification.

\section{Example - Construct CV comb with non-Markovianity}
\label{appendix:bs-bs-process-join}
 Consider the process in Fig.~\ref{fig:process-example} consisting of two identical channels,  $\mathcal A$. Let us construct the process, where $\mathcal A$ is an arbitrary 2-mode Gaussian unitary with transformation matrix such that $\Gamma_\text{output} = T_{\mathcal A}^{\dagger}\Gamma_\text{input}T_{\mathcal A}$:

    \begin{align*}
        T_\mathcal A = \begin{pmatrix}
            \mathcal A_{11} &  \mathcal A_{12} \\
             \mathcal A_{21} &  \mathcal A_{22}
        \end{pmatrix},
    \end{align*}
    
where each block matrix in $T_\mathcal A$ is a $2\times2$ matrix. When $\mathcal A$ is a beam splitter, the result is a simple Gaussian collision model, and the transformation matrix $T_{\mathcal A}$ is given by the following four $2\times 2$ block matrices:

\begin{align*}
    T_{\mathcal A} = \begin{pmatrix}
        \cos\theta I_2 & -\sin\theta e^{i\phi} \\
        \sin\theta e^{-i\phi} & \cos\theta I_2
    \end{pmatrix}.
\end{align*}

As elaborated in text, this is process can be iteratively constructed with the $J_0$- and $J_1$-stitchings. Here, we will proceed assuming the $J_0$ stitching has already been evaluated to give the $\rho_0 \star_{J_0} \mathcal A$ circuit fragment (henceforth written as $\rho_0 \star \mathcal A$ for brevity), such that what follows will demonstrate only the $J_1$-stitching of that fragment with the second fragment, $\mathcal A$.

We label the modes of the first circuit fragment $\rho_0\star\mathcal A$ to be $o_E, o_1^{\mathcal A}, i_1^{\mathcal A}$ as per Fig.~\ref{fig:process-example}. The Choi state of this fragment has an input TMSV, $\Phi(r)$. The second circuit fragment, $\mathcal A$, has the modes $i_2^{\mathcal A}, o_2^{\mathcal A}$, corresponding to the arms of the TMSV state $\Phi(r'')$, and $o_3^{\mathcal A}, i_3^{\mathcal A}$, corresponding to $\Phi(r')$ (see Fig.~\ref{fig:interaction-example-channel}a below).

The stitch of interest is given by $J_1=\set{(o_E, i_3^{\mathcal A})}$, and we will therefore take $r'\to\infty$ later during the evaluation of the link product.

To begin, we take $\rho_0$ to be the vacuum state, such that its covariance matrix is given by $I_2$. The covariance matrix of the TMSV state, $\Gamma_{\Phi}(r)$ is given in (\ref{eqn: covmatrix-tmsv}). We can calculate the covariance matrix of the first circuit fragment, with $\Lambda = \sigma_z$:

\begin{align}
    \Gamma_{\rho_0\star\mathcal A} 
    &= (T_\mathcal{A}^\dagger \oplus I_2)(I_2 \oplus \Gamma_\Phi(r))(T_\mathcal{A} \oplus I_2)\\
    &= \begin{pmatrix}
\cosh(r) \mathcal A_{21}^\dagger \mathcal A_{21} + \mathcal A_{11}^\dagger\mathcal A_{11} & \cosh(r) \mathcal A_{21}^\dagger \mathcal A_{22} + \mathcal A_{11}^\dagger \mathcal A_{12} & \sinh(r) \mathcal A_{21}^\dagger\Lambda \\
\cosh(r) \mathcal A_{22}^\dagger \mathcal A_{21} + \mathcal A_{12}^\dagger \mathcal A_{11} & \cosh(r) \mathcal A_{22}^\dagger\mathcal A_{22} + \mathcal A_{12}^\dagger \mathcal A_{12} & \sinh(r) \mathcal A_{22}^\dagger\Lambda \\
\sinh(r) \Lambda \mathcal A_{21} & \sinh(r) \Lambda \mathcal A_{22} & \cosh(r) I_2
\label{eqn: vac-tmsv-bs}
\end{pmatrix}
\end{align}

where the order of the covariance matrix corresponds to the qumode ordering of $(o_E, o_1^{\mathcal A},i_1^{\mathcal A})$. While mode-reordering occurs during the subroutine \texttt{\pcm} for the algorithm, for the benefit of pedagogical clarity throughout this example, we will explicitly re-order the sub-matrices of $\Gamma_{\mathcal \rho_0\star \mathcal A}$, and use lines to visually separate blocks corresponding to modes that are participate in the $J$-stitch. We will reorder the qumodes according to \texttt{\pcm}, where the qumodes participating in the join will constitute the last terms (and correspond to the bottom-right blocks of the covariance matrix).

For $\Gamma_{\rho_0\star \mathcal A}$, its re-ordered matrix, $\Gamma_{\rho_0\star \mathcal A} '$, follows the qumode ordering of $(o_1^{\mathcal A},i_1^{\mathcal A},o_E)$ to yield:

\begin{align}
    \Gamma_{\rho_0\star \mathcal A} '
    &= \left(\begin{array}{cc|c}
    \cosh(r) \mathcal A_{22}^\dagger\mathcal A_{22} + \mathcal A_{12}^\dagger \mathcal A_{12} & \sinh(r) \mathcal A_{22}^\dagger\Lambda & \cosh(r) \mathcal A_{22}^\dagger \mathcal A_{21} + \mathcal A_{12}^\dagger \mathcal A_{11}  \\
    \sinh(r) \Lambda\mathcal A_{22} & \cosh(r) I_2 & \sinh(r) \Lambda \mathcal A_{21} \\
    \hline
    \cosh(r) \mathcal A_{21}^\dagger \mathcal A_{22} + \mathcal A_{11}^\dagger \mathcal A_{12} & \sinh(r) \mathcal A_{21}^\dagger\Lambda & \cosh(r) \mathcal A_{21}^\dagger \mathcal A_{21} + \mathcal A_{11}^\dagger\mathcal A_{11}\\
\end{array}\right) = \left( \begin{array}{c|c}
        \gamma_{A \setminus J} & \gamma_\text{A,c} \\
        \hline
        \gamma_\text{A,c}^T & \gamma_{J}
    \end{array} \right),
\end{align}

where the division and its corresponding sub-block has been explicitly noted. Since the mode $o_E$ is the mode participating in the stitch, it is positioned as the bottom-right submatrix. 

The second circuit fragment, $\mathcal A$ will have a covariance matrix, corresponding to the qumode order $(i_3^{\mathcal A}, o_3^{\mathcal A}, o_2^{\mathcal A}, i_2^{\mathcal A})$, given by:

\begin{align}
    \Gamma_{\mathcal A} 
    &= (I_2 \oplus T_\mathcal{A}^\dagger \oplus I_2)(\Gamma_\Phi(r') \oplus \Gamma_\Phi(r''))(I_2 \oplus  T_\mathcal{A} \oplus I_2) \nonumber \\
    &= \left(\begin{array}{cccc}
\cosh(r')I_2 & \sinh(r')\Lambda \mathcal A_{11} & \sinh(r') \Lambda \mathcal A_{12} & \mathbf 0 \\
\sinh(r') \mathcal A_{11}^\dagger \Lambda & \cosh(r') \mathcal A_{11}^\dagger \mathcal A_{11} + \cosh(r'') \mathcal A_{21}^\dagger\mathcal A_{21} & \cosh(r') \mathcal A_{11}^\dagger \mathcal A_{12} + \cosh(r'') \mathcal A_{21}^\dagger \mathcal A_{22} & \sinh (r'') \mathcal A_{21}^\dagger\Lambda \\
\sinh(r') \mathcal A_{12}^\dagger\Lambda & \cosh(r') \mathcal A_{12}^\dagger\mathcal A_{11} + \cosh(r'') \mathcal A_{22}^\dagger \mathcal A_{21} & \cosh(r') \mathcal A_{12}^\dagger\mathcal A_{12} + \cosh(r'') \mathcal A_{22}^\dagger\mathcal A_{22} & \sinh(r'') \mathcal A_{22}^\dagger\Lambda \\
\mathbf 0 & \sinh(r'') \Lambda \mathcal A_{21} & \sinh(r'') \Lambda \mathcal A_{22} & \cosh(r'')I_2
\end{array}\right).
\end{align}

Once again, we will explicitly reorder the matrix. The re-ordered matrix, $\Gamma'_{\mathcal A}$ corresponds to the new qumode order of $(o_3^{\mathcal A},o_2^{\mathcal A},i_2^{\mathcal A},i_3^{\mathcal A})$, and is of the form:

\begin{align}
    \Gamma'_{\mathcal A} 
    &= \left(\begin{array}{ccc|c}
 \cosh(r') \mathcal A_{11}^\dagger \mathcal A_{11} + \cosh(r'') \mathcal A_{21}^\dagger\mathcal A_{21} & \cosh(r') \mathcal A_{11}^\dagger \mathcal A_{12} + \cosh(r'') \mathcal A_{21}^\dagger \mathcal A_{22} & \sinh (r'') \mathcal A_{21}^\dagger\Lambda & \sinh(r') \mathcal A_{11}^\dagger \Lambda \\
\cosh(r') \mathcal A_{12}^\dagger\mathcal A_{11} + \cosh(r'') \mathcal A_{22}^\dagger \mathcal A_{21} & \cosh(r') \mathcal A_{12}^\dagger\mathcal A_{12} + \cosh(r'') \mathcal A_{22}^\dagger\mathcal A_{22} & \sinh(r'') \mathcal A_{22}^\dagger\Lambda & \sinh(r') \mathcal A_{12}^\dagger \Lambda \\
\sinh(r'') \Lambda \mathcal A_{21} & \sinh(r'') \Lambda \mathcal A_{22} & \cosh(r'')I_2 & \mathbf 0\\
\hline
\sinh(r') \Lambda \mathcal A_{11}  &  \sinh(r') \Lambda \mathcal A_{12} & \mathbf 0 & \cosh(r') I_2
\end{array}\right), \\
&= \left( \begin{array}{c|c}
        \gamma_{B \setminus J} & \gamma_\text{B,c} \\
        \hline
        \gamma_\text{B,c}^T & \gamma_{J}
    \end{array} \right),
\end{align}

where we have once again noted its division and sub-blocks. The explicit notation of the sub-blocks can be used in reference to the algorithm outlined in the main text to obtain the output of \texttt{\pcm}$(\Gamma_{\rho_0 \star \mathcal A} \oplus \Gamma_{\mathcal A}, J_1 =\set{(o_E,i_3^{\mathcal A})})$.

Then, iterating once in the algorithm for the single-mode stitch in $J_1$, we obtain the output covariance matrix of the form

\begin{align}
\tilde \Gamma_\text{out} = 
\left( \begin{array}{cc|c|cc}
\Gamma_{11} & \Gamma_{12} & \Gamma_{01}^\dagger & \Gamma_{13} & \mathbf 0 \\
\Gamma_{12}^{\dagger} & \Gamma_{22} & \Gamma_{02}^\dagger & \Gamma_{23} & \mathbf 0 \\
\hline
\Gamma_{01} & \Gamma_{02} & \Gamma_{00} & \Gamma_{03} & \Gamma_{04} \\
\hline
\Gamma_{13}^\dagger & \Gamma_{23}^\dagger & \Gamma_{03}^\dagger & \Gamma_{33} & \Gamma_{34} \\
 \mathbf 0 & \mathbf 0 & \Gamma_{04}^\dagger & \Gamma_{34}^\dagger & \Gamma_{44}
\end{array} \right )
\label{eqn: bs-bs-covmatrix-env},
\end{align}

corresponding to qumode ordering $(o_1^{\mathcal A}, i_1^{\mathcal A}, o_3^{\mathcal A}, o_2^{\mathcal A}, i_2^{\mathcal A})$. We have included the lines to denote the environment subsystem (qumode $o_3^{\mathcal A}$) and its corresponding block matrices ($\Gamma_{00}$, and its correlation submatrices $\Gamma_{01}, \Gamma_{02}, \Gamma_{03}, \Gamma_{04}$).

To find the Choi state of the resulting process (i.e., quantum comb), without the environment, we trace out the environment mode by deleting the third row and column of the covariance matrix corresponding to $o_3^{\mathcal A}$. The final covariance matrix of the process is thus given by

\begin{align}
\Gamma_{\mathcal T} = 
\begin{pmatrix}
\Gamma_{11} & \Gamma_{12} & \Gamma_{13} & \mathbf 0 \\
\Gamma_{12}^{\dagger} & \Gamma_{22} & \Gamma_{23} & \mathbf 0 \\
\Gamma_{13}^\dagger & \Gamma_{23}^\dagger & \Gamma_{33} & \Gamma_{34} \\
\mathbf 0 & \mathbf 0 & \Gamma_{34}^\dagger & \Gamma_{44}
\end{pmatrix}.
\end{align}

With $\Xi = \cosh(r)\Lambda \mathcal A_{21}^\dagger \mathcal A_{21}\Lambda + \cosh(r') I_2 +\Lambda \mathcal A_{11}^\dagger\mathcal A_{11}\Lambda$, each block matrix $\Gamma_{kl}$ in $\Gamma_\text{out}$ is a $2\times 2$ matrix with the forms:

\begin{align*}
    \Gamma_{11} &= \cosh(r) \mathcal A_{22}^\dagger \mathcal A_{22} + \mathcal A_{12}^\dagger\mathcal A_{12} - (\cosh(r) \mathcal A_{22}^\dagger \mathcal A_{21}+\mathcal A_{12}^\dagger \mathcal A_{11}) \Lambda \Xi^{-1} \Lambda(\cosh(r) \mathcal A_{21}^\dagger \mathcal A_{22} + \mathcal A_{11}^\dagger \mathcal A_{12}), \\
    \Gamma_{12} &= \sinh(r) \mathcal A_{22}^\dagger\Lambda - (\cosh(r)\mathcal A_{22}^\dagger \mathcal A_{21} + \mathcal A_{12}^\dagger \mathcal A_{11}) \Lambda \Xi^{-1}\Lambda\sinh(r) \mathcal A_{21}^\dagger\Lambda, \\
    \Gamma_{13} &= (\cosh(r) \mathcal A_{22}^\dagger \mathcal A_{21} + \mathcal A_{12}^\dagger \mathcal A_{11})\Lambda \Xi^{-1} \sinh(r') \Lambda \mathcal A_{12}, \\
    \Gamma_{22} &= \cosh(r) I_2 - \sinh^2(r) \Lambda \mathcal A_{21} \Lambda \Xi^{-1} \Lambda \mathcal A_{21}^\dagger \Lambda, \\
    \Gamma_{23} &= \sinh(r)\Lambda \mathcal A_{21} \Lambda \Xi^{-1} \sinh(r')  \Lambda \mathcal A_{12}, \\
    \Gamma_{33} &= \cosh(r'') \mathcal A_{22}^\dagger \mathcal A_{22} + \cosh(r') \mathcal A_{12}^\dagger \mathcal A_{12} - \sinh^2(r') \mathcal A_{12}^\dagger \Lambda \Xi^{-1} \Lambda \mathcal A_{12}, \\
    \Gamma_{34} &= \sinh(r'') \mathcal A_{22}^\dagger \Lambda, \\
    \Gamma_{44} &= \cosh(r'') I_2.
\end{align*}

For explicit evaluation, the terms corresponding to the deleted environment modes are given as follows:

\begin{align*}
\Gamma_{00} &= \cosh(r') \mathcal A_{11}^\dagger\mathcal A_{11} + \cosh(r'') \mathcal A_{21}^\dagger \mathcal A_{21} -\sinh^2(r') \mathcal A_{11}^\dagger\Lambda\Xi^{-1} \Lambda \mathcal A_{11} \\
\Gamma_{01} &= \sinh(r') \mathcal A_{11}^\dagger\Lambda \Xi^{-1} \Lambda (\cosh(r)\mathcal A_{21}^\dagger\mathcal A_{22} + \mathcal A_{11}^\dagger \mathcal A_{12}) \\
\Gamma_{02} &= \sinh(r) \sinh(r') \mathcal A_{11}^\dagger \Lambda \Xi^{-1} \Lambda \mathcal A_{21}^\dagger \Lambda \\
\Gamma_{03} &= \cosh(r') \mathcal A_{11}^\dagger \mathcal A_{12} + \cosh(r'') \mathcal A_{21}^\dagger \mathcal A_{22} - \sinh^2(r') \mathcal A_{11}^\dagger\Lambda \Xi^{-1} \Lambda \mathcal A_{12} \\
\Gamma_{04} &= \sinh(r'') \mathcal A_{21}^\dagger \Lambda 
\end{align*}

If $\mathcal A$ is a beam splitter, then the transformation matrix

\begin{align}
    T_{\mathcal A} =
        \displaystyle \left[\begin{array}{cc|cc}
        \cos{\left(\theta \right)} & 0 & - e^{i \phi} \sin{\left(\theta \right)} & 0\\
        0 & \cos{\left(\theta \right)} & 0 & - e^{- i \phi} \sin{\left(\theta \right)}\\
        \hline
        e^{- i \phi} \sin{\left(\theta \right)} & 0 & \cos{\left(\theta \right)} & 0\\
        0 & e^{i \phi} \sin{\left(\theta \right)} & 0 & \cos{\left(\theta \right)}
        \end{array}\right]
        = \left(\begin{array}{c|c}
            \mathcal A_{11} &  \mathcal A_{12} \\
            \hline
             \mathcal A_{21} &  \mathcal A_{22}
             \end{array}\right),
             \label{eqn: beam splitter transformation matrix}
\end{align}

where $\cos^2(\theta)$ is the transmittivity, and $\phi$ is the phase difference between the transmitted and reflected fields. For a real beam splitter, $\phi=0$. If $\mathcal A$ is a beam splitter with transmittance angle $\theta$, and we apply the limit $r'\to\infty$, we can recover, from $\Gamma_{\mathcal T}$, the following covariance matrix, following mode order $(o_1^{\mathcal A}, i_1^{\mathcal A}, o_2^{\mathcal A}, i_2^{\mathcal A})$:

$$
\Gamma_{\mathcal T} = 
\begin{pmatrix}
(\cosh r \cos^2\theta + \sin^2\theta) I_2 & \sinh r \cos\theta \Lambda & \sin^2\theta \cos\theta (1-\cosh r) I_2 & \mathbf 0 \\
\sinh r \cos\theta \Lambda & \cosh r I_2 & -\sinh r \sin^2 \theta \Lambda & \mathbf 0 \\
\sin^2\theta \cos\theta (1-\cosh r) I_2 & -\sinh r  \sin^2 \theta \Lambda & \cosh r'' \cos^2 \theta I_2 + \sin^2 \theta (\cosh r \sin^2 \theta + \cos^2\theta) I_2 & \sinh r'' \cos\theta \Lambda \\
\mathbf 0 & \mathbf 0 & \sinh r'' \cos\theta \Lambda & \cosh r'' I_2
\end{pmatrix}.
$$

Below, we briefly outline how the above is obtained. $\Gamma_{34},\Gamma_{44}$ both can be read off directly. For all other terms, we substituted the form of $\mathcal A_{ij}$ from (\ref{eqn: beam splitter transformation matrix}) into $\Xi$ and each $\Gamma_{kl}$. Since $\Xi = \cosh(r) \sin^2(\theta) I_2 + \cosh(r') I_2 + \cos^2(\theta)I_2$, under the limit $r'\to\infty$, we can use the following approximation

\begin{align}
    \Xi^{-1} &= \frac{1}{\cosh(r')} \left(1+\frac{\cosh(r) \sin^2\theta}{\cosh(r')}+\frac{\cos^2\theta}{\cosh(r')}\right)^{-1}I_2 \\
    &\approx \frac{1}{\cosh(r')} \left(1-\frac{\cosh(r) \sin^2\theta}{\cosh(r')}-\frac{\cos^2\theta}{\cosh(r')}\right)I_2
\end{align}

Insert this approximation to each $\Gamma_{kl}$, we can recover the appropriate terms. The trick is to use the approximate form of $\Xi^{-1}$ in either two ways: (1) to the first order if the other terms multiplied into $\Xi^{-1}$ contains no $r'$ dependence, such that the terms go to zero in the limit, and (2) to higher orders in order to `cancel out' any $\cosh r'$ or $\sinh r'$ factor (e.g., by using trigonometric identity of $\sinh r'/\cosh r' = \tanh r'$, or by using $\sinh r' = \cosh r'$ in the limit of $r'\to\infty$) such that they do not survive and blow up in the limit of $r'\to\infty$.

For example, in the first use case, we show this evaluation for $\Gamma_{22}$. We can show that $\cosh(r) I_2 - \sinh^2(r) \Lambda \mathcal A_{21} \Lambda \Xi^{-1} \Lambda \mathcal A_{21}^\dagger \Lambda$ will recover $\cosh r I_2$:

\begin{align}
    \Gamma_{22} &= \cosh(r) I_2 - \sinh^2(r) \Lambda \mathcal A_{21} \Lambda \Xi^{-1} \Lambda \mathcal A_{21}^\dagger \Lambda \\
    &\approx \cosh(r)I_2 - \frac{\sinh^2(r) \sin^2\theta}{\cosh(r')} \left(1-\frac{\cosh(r) \sin^2\theta}{\cosh(r')}-\frac{\cos^2\theta}{\cosh(r')}\right)I_2 \\
    &\approx \cosh(r)I_2 - \frac{\sinh^2(r) \sin^2\theta}{\cosh(r')} I_2 \\
    &=_{r'\to\infty} \cosh(r)I_2
\end{align}

In the limit of $r'\to\infty$, $\cosh(r')\to\infty$ and thus only $\cosh(r)I_2$ remains.

For the second use case, we can evaluate
\begin{align}
    \Gamma_{23} &= \sinh(r)\Lambda \mathcal A_{21} \Lambda \Xi^{-1} \sinh(r')  \Lambda \mathcal A_{12} \\
    &\approx \sinh(r)\Lambda \mathcal A_{21} \Lambda \frac{\sinh(r')}{\cosh(r')}   \Lambda \mathcal A_{12} \\
    &\approx \sinh(r) \tanh(r') \Lambda \mathcal A_{21}A_{12} = \sinh(r) \tanh(r') \Lambda (-\sin^2\theta) \\
    &=_{r'\to\infty} -\sinh(r) \sin^2\theta \Lambda
\end{align}

All the $\Gamma_{kl}$ may be evaluated in a similar manner. In order to simplify and make more compact the final expression of $\Gamma_{\mathcal T}$ in the main body text, we elected to set $r'' = r$, such that $u_{r''} = u_r$ and $v_{r''}=v_r$  - this corresponds to the case where both surviving TMSV states have the same squeezing.

Separately, the covariance matrix can be evaluated by applying two iterations of $\mathcal A$ as per Fig.~\ref{fig:process-example}. This corresponds to evaluating the covariance matrix of the state given by $\trace_E[\mathcal A^{(ES_2)}\circ \mathcal A^{(ES_1)}(\rho_0^{(E)} \otimes \Phi^{(S_1L_1)}(r)\otimes \Phi^{(S_2L_2)}  (r))]$, (where we explicitly denote the subsystems they act on for clarity). Doing so, we obtain a covariance matrix of the same form. This verifies that the result of our link product indeed corresponds to the covariance of the Choi matrix of the process pictured in Fig.~\ref{fig:process-example}a.

\section{Example - Interacting a comb with a channel}
\label{appendix: comb-channel-interact}

\begin{figure}[tp]
    \centering
    \includegraphics[width=0.95\linewidth]{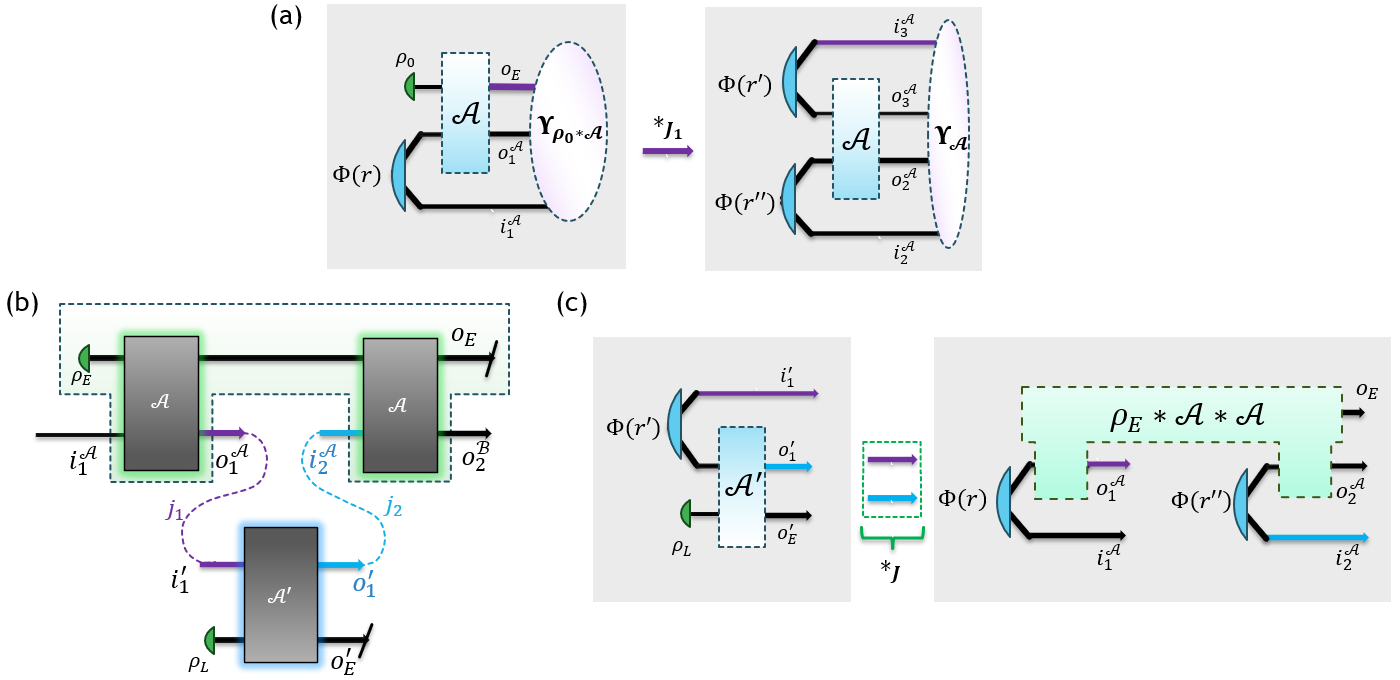}
    \caption{(a) \textbf{Building non-Markovianity.} In this section, we expound on how to evaluate the stitching of $o_E$ to $i_3^{\mathcal A}$. The mode labels and TMSV squeezing dependences are illustrated in the figure. (b) \textbf{Interacting a comb with a channel.} Here, we link the process from the previous example to a multimode channel via two mode-stitches, $j_1 = (o_1^{\mathcal A},i_1')$ and $j_2= (o_1',i_2^{\mathcal A})$. As per all the previous examples, for all information to pass through the join, $r',r''\to\infty$, corresponding to $I_J = \set{i_1', i_2^{\mathcal A}}$. The Choi states of the circuit fragments involved are given in (c). }
    \label{fig:interaction-example-channel}
\end{figure}

In this example, we demonstrate how two processes may be connected using the link product algorithm. We begin with the previous example with two connected beam splitters. Consider now that, as per Fig.~\ref{fig:interaction-example-channel}b, we wish to connect a beam splitter (i.e., a single-timestep process) $\mathcal A'$ with transmittance angle $\alpha$ (with input $\rho_L$) i.e. transmittivity $\cos^2\alpha$, to the process we just evaluated in the previous section. We maintain the same mode labels as the previous section for its circuit fragment in this example, except $o_3^{\mathcal A} = o_E$. For the second circuit fragment, $\rho_L \star \mathcal A'$, we will use the mode labels $o_E', o_1', i_1'$. For both these fragments, $\rho_0 = \rho_L = \rho_E$ are taken to be the vacuum state such that its covariance matrix is $I_2$. This circuit stitching is now {\it bidirectional}, with $J = \{(o_1^{\mathcal A}, i_1'), (o_1',i_2^{\mathcal A})\}$.

We start with the resulting covariance matrix of the beam splitter-beam splitter process calculated in the previous section, given by $\tilde\Gamma_\text{out}$ in (\ref{eqn: bs-bs-covmatrix-env}). In this example, we also elect to also explicitly reorder the covariance matrices (as would have been done within \texttt{\pcm}) for pedagogical clarity. We reorder $\Gamma_{\rho_E\star\mathcal A\star \mathcal A}$ to correspond to mode order $(o_E,i_1^{\mathcal A},o_2^{\mathcal A},i_2^{\mathcal A},o_1^{\mathcal A})$:

\begin{align}
\Gamma_{\rho_E\star\mathcal A\star \mathcal A} = 
\left( \begin{array}{ccc|cc}
\Gamma_{00} & \Gamma_{02} & \Gamma_{03} & \Gamma_{04} & \Gamma_{01} \\
\Gamma_{02}^\dagger & \Gamma_{22} & \Gamma_{23} & \mathbf 0 & \Gamma_{12}^\dagger \\
\Gamma_{03}^\dagger &\Gamma_{23}^{\dagger} & \Gamma_{33} & \Gamma_{34} & \Gamma_{13}^\dagger \\
\hline
\Gamma_{04}^\dagger & \mathbf 0 & \Gamma_{34}^\dagger & \Gamma_{44} & \mathbf 0 \\
\Gamma_{01}^\dagger & \Gamma_{12} & \Gamma_{13} & \mathbf 0 & \Gamma_{11}
\end{array} \right )
\end{align}

where the form of each term is given by the previous section. The lines represent the division of the covariance matrix according to the modes participating in the join. The submatrices of modes that participate in the join are in the bottom-right.

For the second process, the single beam splitter we wish to connect, we can use the form of $\Gamma_{\rho_0 \star \mathcal A}$ from (\ref{eqn: vac-tmsv-bs}). We denote the subdivision to yield the covariance matrix for our second fragment, $\rho_L \star \mathcal A'$, [following mode order $(o_E',o_1',i_1')$]:

\begin{align}
    \Gamma_{\rho_L\star \mathcal A'} 
    &= \left( \begin{array}{c|cc}
\cosh(r') \mathcal A_{21}'^\dagger \mathcal A_{21}' + \mathcal A_{11}'^\dagger\mathcal A_{11}' & \cosh(r') \mathcal A_{21}'^\dagger \mathcal A_{22}' + \mathcal A_{11}^\dagger \mathcal A_{12}' & \sinh(r') \mathcal A_{21}'^\dagger\Lambda \\
\hline
\cosh(r') \mathcal A_{22}'^\dagger \mathcal A_{21}' + \mathcal A_{12}'^\dagger \mathcal A_{11}' & \cosh(r') \mathcal A_{22}'^\dagger\mathcal A_{22}' + \mathcal A_{12}'^\dagger \mathcal A_{12}' & \sinh(r') \mathcal A_{22}'^\dagger\Lambda \\
\sinh(r') \Lambda \mathcal A_{21}' & \sinh(r') \Lambda \mathcal A_{22}' & \cosh(r') I_2
\end{array} \right) = \left( \begin{array}{c|cc}
\gamma_{11} & \gamma_{12} & \gamma_{13} \\
\hline
\gamma_{12}^\dagger & \gamma_{22} & \gamma_{23} \\
\gamma_{13}^\dagger & \gamma_{23}^\dagger & \gamma_{33} \\
\end{array} \right).
\end{align}

We use $\gamma_{ij}$ to simplify the notation in the further parts of this section. We also note that the order of the sub-matrices within the bottom-right block matrices of $\Gamma_{\rho_E\star \mathcal A \star \mathcal A}$ and $\Gamma_{\rho_L\star \mathcal A'}$ are arranged according to their stitch-pairs: for the penultimate diagonal elements in their respective matrices, $\Gamma_{44}$ and $\gamma_{22}$, the mode corresponding to $\Gamma_{44}$ ($i_2^{\mathcal A}$) is to be joined to the mode corresponding to $\gamma_{22}$ ($o_1'$). Similarly, for the last entries, $\Gamma_{11}$ (corresponding to mode $o_1^{\mathcal A}$) and $\gamma_{33}$ (corresponding to mode $i_1'$), their corresponding modes are to be joined to each other.

Using \texttt{\pcm}($\Gamma_{\rho_L\star \mathcal A'} \oplus \Gamma_{\rho_E\star\mathcal A\star \mathcal A}, J=\set{(o_1^{\mathcal A},i_1'),(o_1',i_2^{\mathcal A})}$), we obtain the covariance matrix of the form (with $\Lambda = \sigma_z$):

\begin{align}
\gamma_\text{tot} = 
\left( \begin{array}{cccc|cc}
\gamma_{11} & \mathbf 0 & \mathbf 0 & \mathbf 0 & -\gamma_{12}\Lambda & -\gamma_{13}\Lambda \\
\mathbf 0 & \Gamma_{00} & \Gamma_{02} & \Gamma_{03} & \Gamma_{04} & \Gamma_{01} \\
\mathbf 0 & \Gamma_{02}^\dagger & \Gamma_{22} & \Gamma_{23} & \mathbf 0 & \Gamma_{12}^\dagger \\
\mathbf 0 & \Gamma_{03}^\dagger &\Gamma_{23}^{\dagger} & \Gamma_{33} & \Gamma_{34} & \Gamma_{13}^\dagger \\
\hline
-\Lambda\gamma_{12}^\dagger & \Gamma_{04}^\dagger & \mathbf 0 & \Gamma_{34}^\dagger & \Gamma_{44} + \Lambda \gamma_{22} \Lambda & \Lambda \gamma_{23} \Lambda  \\
-\Lambda\gamma_{13}^\dagger & \Gamma_{01}^\dagger & \Gamma_{12} & \Gamma_{13} & \Lambda \gamma_{23}^\dagger \Lambda  & \Gamma_{11} + \Lambda \gamma_{33} \Lambda 
\end{array} \right )
\end{align}

Explicitly, this ordering follows the mode order $(o_E', o_E, i_1^{\mathcal A},o_2^{\mathcal A}, o_1'/i_2^{\mathcal A}, o_1^{\mathcal A}/i_1')$ - where we note the stitched modes together. We then follow the iterations of the link product algorithm. In the first iteration, for \texttt{\smc}, the matrix is subdivided as 

\begin{align}
\gamma_\text{tot} &= 
\left( \begin{array}{ccccc|c}
\gamma_{11} & \mathbf 0 & \mathbf 0 & \mathbf 0 & -\gamma_{12}\Lambda & -\gamma_{13}\Lambda \\
\mathbf 0 & \Gamma_{00} & \Gamma_{02} & \Gamma_{03} & \Gamma_{04} & \Gamma_{01} \\
\mathbf 0 & \Gamma_{02}^\dagger & \Gamma_{22} & \Gamma_{23} & \mathbf 0 & \Gamma_{12}^\dagger \\
\mathbf 0 & \Gamma_{03}^\dagger &\Gamma_{23}^{\dagger} & \Gamma_{33} & \Gamma_{34} & \Gamma_{13}^\dagger \\
-\Lambda\gamma_{12}^\dagger & \Gamma_{04}^\dagger & \mathbf 0 & \Gamma_{34}^\dagger & \Gamma_{44} + \Lambda \gamma_{22} \Lambda & \Lambda \gamma_{23} \Lambda  \\
\hline
-\Lambda\gamma_{13}^\dagger & \Gamma_{01}^\dagger & \Gamma_{12} & \Gamma_{13} & \Lambda \gamma_{23}^\dagger \Lambda  & \Gamma_{11} + \Lambda \gamma_{33} \Lambda 
\end{array} \right ) = \left( \begin{array}{c|c}
        \gamma_{\bar X}^{(1)} & \gamma_{\bar X,y=1} \\
        \hline
        \gamma_{\bar X,y=1}^{\dagger} & \gamma_{y=1}
    \end{array} \right).
\end{align}

such that the output after this first iteration of \texttt{\smc} is given by $ \gamma_\text{res}(y=1) = \gamma_{\bar X^{(1)}} - \gamma_{\bar X,y=1} (\gamma_{y=1})^{-1} \gamma_{\bar X,y=1}^\dagger$. This corresponds to the output after the first mode-stitching of modes $(o_1^{\mathcal A},i_1')$.

This resulting covariance matrix, $\gamma_\text{res}(y=1)$, will now be used as the new $\gamma_\text{tot}$ in the second iteration. Therefore, this new covariance matrix will be subdivided as

\begin{align}
\gamma_\text{res}(y=1) &= 
\left( \begin{array}{cccc|c}
\gamma_{11} & \mathbf 0 & \mathbf 0 & \mathbf 0 & -\gamma_{12}\Lambda\\
\mathbf 0 & \Gamma_{00} & \Gamma_{02} & \Gamma_{03} & \Gamma_{04} \\
\mathbf 0 & \Gamma_{02}^\dagger & \Gamma_{22} & \Gamma_{23} & \mathbf 0 \\
\mathbf 0 & \Gamma_{03}^\dagger &\Gamma_{23}^{\dagger} & \Gamma_{33} & \Gamma_{34} \\
\hline
-\Lambda\gamma_{12}^\dagger & \Gamma_{04}^\dagger & \mathbf 0 & \Gamma_{34}^\dagger & \Gamma_{44} + \Lambda \gamma_{22} \Lambda  \\
\end{array} \right ) \\
&\ -
\left( \begin{array}{cccc|c}
    \gamma_{13}\Lambda\Xi_{(1)}^{-1}\Lambda \gamma_{13}^\dagger & -\gamma_{13}\Lambda\Xi_{(1)}^{-1}\Gamma_{01}^\dagger & -\gamma_{13}\Lambda\Xi_{(1)}^{-1}\Gamma_{12} & -\gamma_{13}\Lambda\Xi_{(1)}^{-1}\Gamma_{13} &  \gamma_{13}\Lambda\Xi_{(1)}^{-1}\Lambda \gamma_{23}^\dagger \Lambda \\
    -\Gamma_{01}\Xi_{(1)}^{-1}\Lambda \gamma_{13}^\dagger & \Gamma_{01}\Xi_{(1)}^{-1}\Gamma_{01}^\dagger & \Gamma_{01}\Xi_{(1)}^{-1}\Gamma_{12} & \Gamma_{01}\Xi_{(1)}^{-1}\Gamma_{13} &  \Gamma_{01}\Xi_{(1)}^{-1}\Lambda \gamma_{23}^\dagger \Lambda \\
    -\Gamma_{12}^\dagger\Xi_{(1)}^{-1}\Lambda \gamma_{13}^\dagger & \Gamma_{12}^\dagger\Xi_{(1)}^{-1}\Gamma_{01}^\dagger & \Gamma_{12}^\dagger\Xi_{(1)}^{-1}\Gamma_{12} & \Gamma_{12}^\dagger\Xi_{(1)}^{-1}\Gamma_{13} &  \Gamma_{12}^\dagger \Xi_{(1)}^{-1}\Lambda \gamma_{23}^\dagger \Lambda \\
    -\Gamma_{13}^\dagger\Xi_{(1)}^{-1}\Lambda \gamma_{13}^\dagger & \Gamma_{13}^\dagger\Xi_{(1)}^{-1}\Gamma_{01}^\dagger & \Gamma_{13}^\dagger\Xi_{(1)}^{-1}\Gamma_{12} & \Gamma_{13}^\dagger\Xi_{(1)}^{-1}\Gamma_{13} &  \Gamma_{13}^\dagger \Xi_{(1)}^{-1}\Lambda \gamma_{23}^\dagger \Lambda \\
    \hline
    -\Lambda\gamma_{23}\Lambda \Xi_{(1)}^{-1}\Lambda \gamma_{13}^\dagger & \Lambda\gamma_{23}\Lambda \Xi_{(1)}^{-1}\Gamma_{01}^\dagger & \Lambda\gamma_{23}\Lambda \Xi_{(1)}^{-1}\Gamma_{12} & \Lambda\gamma_{23}\Lambda \Xi_{(1)}^{-1}\Gamma_{13} &  \Lambda\gamma_{23}\Lambda  \Xi_{(1)}^{-1}\Lambda \gamma_{23}^\dagger \Lambda
    \end{array} \right) \\
    &= \left( \begin{array}{c|c}
        \gamma_{\bar X}^{(2)} & \gamma_{\bar X,y=2} \\
        \hline
        \gamma_{\bar X,y=2}^{\dagger} & \gamma_{y=2}
    \end{array} \right).
\end{align}

where $\Xi_{(1)}=\Gamma_{11} + \Lambda \gamma_{33} \Lambda$. Thus, the output after the second iteration of \texttt{\smc} will be given by $\gamma_\text{res}(y=2) = \gamma_{\bar X^{(2)}} - \gamma_{\bar X,y=2} (\gamma_{y=2})^{-1} \gamma_{\bar X,y=2}^\dagger$. This will correspond to the second mode-stitching of the modes $(o_1',i_2^{\mathcal A})$, and $\gamma_\text{res}(y=2)$ as outlined here will have a form corresponding to the qumode order $(o_E', o_E, i_1^{\mathcal A},o_2^{\mathcal A})$.

Evaluating the second iteration of \texttt{\smc} to find $\gamma_\text{res}(y=2)$, and discarding the environment subsystems $o_E, o_E'$ (corresponding to the first two rows and first two columns), we obtain the final covariance matrix:

\begin{align}
    \Gamma_\text{tot} &= \begin{pmatrix}
        \Gamma_{22}-\Gamma_{12}^\dagger\Xi_{(1)}^{-1}\Gamma_{12} & \Gamma_{23} - \Gamma_{12}^\dagger\Xi_{(1)}^{-1}\Gamma_{13} \\
        \Gamma_{23}^\dagger - \Gamma_{13}^\dagger\Xi_{(1)}^{-1}\Gamma_{12} & \Gamma_{33} - \Gamma_{13}^\dagger\Xi_{(1)}^{-1}\Gamma_{13}
    \end{pmatrix} \\
    &- \begin{pmatrix}
        \Gamma_{12}^\dagger \Xi_{(1)}^{-1}\Lambda \gamma_{23}^\dagger \Lambda \Xi_{(2)}^{-1} \Lambda\gamma_{23}\Lambda \Xi_{(1)}^{-1}\Gamma_{12} & -\Gamma_{12}^\dagger \Xi_{(1)}^{-1}\Lambda \gamma_{23}^\dagger \Lambda \Xi_{(2)}^{-1} (\Gamma_{34}^\dagger-\Lambda\gamma_{23}\Lambda \Xi_{(1)}^{-1}\Gamma_{13}) \\
        -(\Gamma_{34}-\Gamma_{13}^\dagger \Xi_{(1)}^{-1}\Lambda \gamma_{23}^\dagger \Lambda) \Xi_{(2)}^{-1} \Lambda\gamma_{23}\Lambda \Xi_{(1)}^{-1}\Gamma_{12} & (\Gamma_{34}-\Gamma_{13}^\dagger \Xi_{(1)}^{-1}\Lambda \gamma_{23}^\dagger \Lambda) \Xi_{(2)}^{-1} (\Gamma_{34}^\dagger-\Lambda\gamma_{23}\Lambda \Xi_{(1)}^{-1}\Gamma_{13})
    \end{pmatrix}
\end{align}

where $\Xi_{(2)} = \Gamma_{44} + \Lambda \gamma_{22} \Lambda -\Lambda\gamma_{23}\Lambda  \Xi_{(1)}^{-1}\Lambda \gamma_{23}^\dagger \Lambda $.

To evaluate this final matrix, we have, for $\mathcal A'$ with transmission angle $\alpha$ such that for $\mathcal A'_{ij}$ in the terms $\gamma_{kl}$,

\begin{align}
    T_{\mathcal A'} =
        \displaystyle \left[\begin{array}{cc|cc}
        \cos{\left(\alpha \right)} & 0 & - e^{i \phi} \sin{\left(\alpha \right)} & 0\\
        0 & \cos{\left(\alpha \right)} & 0 & - e^{- i \phi} \sin{\left(\alpha \right)}\\
        \hline
        e^{- i \phi} \sin{\left(\alpha \right)} & 0 & \cos{\left(\alpha \right)} & 0\\
        0 & e^{i \phi} \sin{\left(\alpha \right)} & 0 & \cos{\left(\alpha \right)}
        \end{array}\right]
        = \left(\begin{array}{c|c}
            \mathcal A'_{11} &  \mathcal A'_{12} \\
            \hline
             \mathcal A'_{21} &  \mathcal A'_{22}
             \end{array}\right),
\end{align}

in order to obtain:

\begin{align}
    \Xi_{(1)} &= (\cosh r \cos^2\theta + \sin^2\theta + \cosh r') I_2 \\
    \Xi_{(2)} &\approx ( \cos^2\alpha [\cosh r\cos^2\theta + \sin^2\theta] + \cosh r'' + \sin^2\alpha) I_2
    \end{align}

To obtain $\Xi_{(2)}$, we evaluated $\Lambda \gamma_{22} \Lambda -\Lambda\gamma_{23}\Lambda  \Xi_{(1)}^{-1}\Lambda \gamma_{23}^\dagger \Lambda$ by apply the approximation of $\Xi_{(1)}^{-1}$ in the limit $r'\to\infty$. Under $r'\to\infty$, we have the approximation that
\begin{align}
    \Xi_{(1)}^{-1} &\approx \frac{1}{\cosh(r')} \left[1-\frac{\cosh r \cos^2\theta + \sin^2\theta}{\cosh(r')}\right]I_2,
\end{align}

and under $r''\to\infty$, we have the approximation that
\begin{align}
    \Xi_{(2)}^{-1} &\approx \frac{1}{\cosh r''} \left[1 - \frac{\cos^2\alpha [\cosh r\cos^2\theta + \sin^2\theta] + \sin^2\alpha}{\cosh r''} \right] I_2
\end{align}

Using these approximations, we can show that 

\begin{align}
    \lim_{r',r''\to\infty} \Gamma_\text{tot} =
    \begin{pmatrix}
        \cosh r I_2 & \sinh r (\cos^2\theta \cos\alpha -\sin^2\theta ) \Lambda \\
        \sinh r (\cos^2\theta \cos\alpha -\sin^2\theta ) \Lambda  & D I_2
    \end{pmatrix}
    \label{eqn: process-process-output}
\end{align}

corresponding to qumode order $(i_1^{\mathcal A}, o_2^{\mathcal A})$, where $D$ is given by

\begin{align*}
    D &= \cosh r\sin^4\theta + \sin^2\theta\cos^2\theta + \cosh r \cos^2\alpha \cos^4\theta + \cos^2\theta\sin^2\theta\cos^2\alpha + \cos^2\theta\sin^2\alpha \\
    &+ 2\cos^2\theta\sin^2\theta\cos\alpha - 2\cosh r \cos^2\theta\sin^2\theta\cos\alpha.
\end{align*}

Moreover, by directly evaluating the process by evaluating $(\mathcal A^{(ES_1)} \circ \mathcal A'^{(S_1 L)} \circ \mathcal A^{(ES_1)})[\rho_0^{(E)}\otimes \Phi^{(S_1S_2)}(r) \otimes \rho_{0}^{(L)}]$, (where we explicitly denote the subsystems they act on) we obtain the final form of $\Gamma_\text{tot}$ above, as expected. This verifies that through the link product algorithm, we do indeed obtain the correct covariance matrix.

In the case where $\mathcal A'$ is taken to be the identity channel, the result simplifies to the case of $(\mathcal A\otimes I_2)\circ (\mathcal A\otimes I_2)[\rho_\text{vac} \otimes \Phi(r)]$:

\begin{align}
    \Gamma_\text{out} = \begin{pmatrix}
        [Y^2\cosh r+4\sin^2\theta\cos^2\theta]I_2 & Y\sinh r\Lambda \\
        Y\sinh r\Lambda & \cosh rI_2
    \end{pmatrix}
\end{align}

where $Y=\cos^2 \theta- \sin^2\theta$, which is what we expect - directly calculating $(\mathcal A\otimes I_2)\circ (\mathcal A\otimes I_2)[\rho_\text{vac} \otimes \Phi(r)]$, we get the same covariance matrix.

\subsection{Different Ordering}
What if we had used a different order of $J$-stitching? Previously, we evaluated the mode-stitching of $(o_1^{\mathcal A},i_1')$ first, followed by that of $(o_1',i_2^{\mathcal A})$. If instead, we wanted to evaluate the J-stitching according to $J=\set{(o_1',i_2^{\mathcal A}), (o_1^{\mathcal A},i_1')}$ we will obtain from \texttt{\pcm} a covariance matrix of the form:

\begin{align}
\gamma_\text{tot} = 
\left( \begin{array}{cccc|cc}
\gamma_{11} & \mathbf 0 & \mathbf 0 & \mathbf 0 & -\gamma_{13}\Lambda & -\gamma_{12}\Lambda \\
\mathbf 0 & \Gamma_{00} & \Gamma_{02} & \Gamma_{03} & \Gamma_{01} & \Gamma_{04} \\
\mathbf 0 & \Gamma_{02}^\dagger & \Gamma_{22} & \Gamma_{23} & \Gamma_{12}^\dagger & \mathbf 0 \\
\mathbf 0 & \Gamma_{03}^\dagger &\Gamma_{23}^{\dagger} & \Gamma_{33} & \Gamma_{13}^\dagger & \Gamma_{34} \\
\hline
-\Lambda\gamma_{13}^\dagger & \Gamma_{01}^\dagger & \Gamma_{12} & \Gamma_{13} & \Gamma_{11} + \Lambda \gamma_{33}\Lambda & \Lambda \gamma_{23}^\dagger \Lambda  \\
-\Lambda\gamma_{12}^\dagger & \Gamma_{04}^\dagger & \mathbf 0 & \Gamma_{34}^\dagger & \Lambda \gamma_{23} \Lambda  & \Gamma_{44} + \Lambda \gamma_{22} \Lambda 
\end{array} \right )
\end{align}

Explicitly, this ordering follows the mode order $(o_E', o_E, i_1^{\mathcal A},o_2^{\mathcal A}, o_1^{\mathcal A}/i_1', o_1'/i_2^{\mathcal A})$ instead. This corresponds to a switch between the last two modes in the order of modes. Applying \texttt{\smc} for the first loop, we resolve the stitching for the $o_1^{\mathcal A}/i_1'$ mode first to get

\begin{align}
\gamma_\text{res}(y=1) &= 
\left( \begin{array}{cccc|c}
\gamma_{11} & \mathbf 0 & \mathbf 0 & \mathbf 0 & -\gamma_{13}\Lambda  \\
\mathbf 0 & \Gamma_{00} & \Gamma_{02} & \Gamma_{03} & \Gamma_{01} \\
\mathbf 0 & \Gamma_{02}^\dagger & \Gamma_{22} & \Gamma_{23} & \Gamma_{12}^\dagger \\
\mathbf 0 & \Gamma_{03}^\dagger &\Gamma_{23}^{\dagger} & \Gamma_{33} & \Gamma_{13}^\dagger \\
\hline
-\Lambda\gamma_{13}^\dagger & \Gamma_{01}^\dagger & \Gamma_{12} & \Gamma_{13} & \Gamma_{11} + \Lambda \gamma_{33}\Lambda
\end{array} \right ) \\
&- \left( \begin{array}{cccc|c}
\gamma_{12}\Lambda \Xi_{(1)}^{-1} \Lambda\gamma_{12} & -\gamma_{12}\Lambda \Xi_{(1)}^{-1} \Gamma_{04}^\dagger & \mathbf 0 & -\gamma_{12}\Lambda \Xi_{(1)}^{-1} \Gamma_{34}^\dagger & -\gamma_{12}\Lambda \Xi_{(1)}^{-1} \Lambda\gamma_{23}\Lambda \\
-\Gamma_{04} \Xi_{(1)}^{-1} \Lambda\gamma_{12} & \Gamma_{04} \Xi_{(1)}^{-1} \Gamma_{04}^\dagger & \mathbf 0 & \Gamma_{04} \Xi_{(1)}^{-1} \Gamma_{34}^\dagger & \Gamma_{04} \Xi_{(1)}^{-1} \Lambda\gamma_{23}\Lambda \\
\mathbf 0 & \mathbf 0 & \mathbf 0 & \mathbf 0 & \mathbf 0 \\
-\Gamma_{34} \Xi_{(1)}^{-1} \Lambda\gamma_{12} & \Gamma_{34} \Xi_{(1)}^{-1} \Gamma_{04}^\dagger & \mathbf 0 & \Gamma_{34} \Xi_{(1)}^{-1} \Gamma_{34}^\dagger & \Gamma_{34} \Xi_{(1)}^{-1} \Lambda\gamma_{23}\Lambda \\
\hline
-\Lambda\gamma_{23}^\dagger\Lambda \Xi_{(1)}^{-1} \Lambda\gamma_{12} & \Lambda\gamma_{23}^\dagger\Lambda \Xi_{(1)}^{-1} \Gamma_{04}^\dagger & \mathbf 0 & \Lambda\gamma_{23}^\dagger\Lambda \Xi_{(1)}^{-1} \Gamma_{34}^\dagger & \Lambda\gamma_{23}^\dagger\Lambda \Xi_{(1)}^{-1} \Lambda\gamma_{23}\Lambda \\
\end{array} \right )
\end{align}

where here, $\Xi_{(1)} = \Gamma_{44} + \Lambda \gamma_{22} \Lambda $ instead. Evaluating the form of $\gamma_\text{res}(y=2)$, now corresponding to resolving the stitch of $o_1'/i_2^{\mathcal A}$, and tracing out the environment subsystems, the final covariance matrix:

\begin{align}
    \Gamma_\text{tot} &= 
    \begin{pmatrix}
        \Gamma_{22} & \Gamma_{23} \\
        \Gamma_{23}^\dagger & \Gamma_{33}  - \Gamma_{34} \Xi_{(1)}^{-1} \Gamma_{34}^\dagger
    \end{pmatrix} \\
    &- \begin{pmatrix}
        \Gamma_{12}^\dagger \Xi_{(2)}^{-1} \Gamma_{12} & \Gamma_{12}^\dagger \Xi_{(2)}^{-1} (\Gamma_{13}-\Lambda \gamma_{23}^\dagger \Lambda \Xi_{(1)}^{-1} \Gamma_{34}^\dagger)   \\
         (\Gamma_{13}^\dagger - \Gamma_{34}\Xi_{(1)}^{-1} \Lambda \gamma_{23} \Lambda) \Xi_{(2)}^{-1} \Gamma_{12} & (\Gamma_{13}^\dagger - \Gamma_{34}\Xi_{(1)}^{-1} \Lambda \gamma_{23} \Lambda) \Xi_{(2)}^{-1} (\Gamma_{13}-\Lambda \gamma_{23}^\dagger \Lambda \Xi_{(1)}^{-1} \Gamma_{34}^\dagger)
    \end{pmatrix}
\end{align}

where $\Xi_{(2)} = \Gamma_{11}+\Lambda \gamma_{33}\Lambda - \Lambda\gamma_{23}^\dagger\Lambda \Xi_{(1)}^{-1} \Lambda\gamma_{23}\Lambda$ instead. Taking the approximation $r'' \to\infty$, we get

\begin{align}
    \Xi_{(1)}^{-1} &\approx \frac{1}{\cosh(r'')} \left[1-\frac{\cosh r' \cos^2\alpha + \sin^2\alpha}{\cosh(r'')}\right]I_2,
\end{align}

and under the approximation of $r'\to\infty$:

\begin{align}
    \Xi_{(2)}^{-1} &\approx \frac{1}{\cosh(r')} \left[1-\frac{\cosh r \cos^2\theta + \sin^2\theta}{\cosh(r')}\right]I_2,
\end{align}

Evaluating all the terms, we see that

\begin{align}
    \Gamma_{22} - \Gamma_{12}^\dagger \Xi_{(2)}^{-1} \Gamma_{12} &\approx_{r',r'' \to\infty} \Gamma_{22} = \cosh r I_2 \\
    \Gamma_{23} - \Gamma_{12}^\dagger \Xi_{(2)}^{-1} (\Gamma_{13}-\Lambda \gamma_{23}^\dagger \Lambda \Xi_{(1)}^{-1} \Gamma_{34}^\dagger) &\approx_{r',r'' \to\infty} \Gamma_{23} + \Gamma_{12}^\dagger \Xi_{(2)}^{-1} \Lambda \gamma_{23}^\dagger \Lambda \Xi_{(1)}^{-1} \Gamma_{34}^\dagger \\
    & = \sinh r (\cos^2\theta \cos\alpha - \sin^2\theta)\Lambda \\
    \Gamma_{33}  - \Gamma_{34} \Xi_{(1)}^{-1} \Gamma_{34}^\dagger - (\Gamma_{13}^\dagger - \Gamma_{34}\Xi_{(1)}^{-1} \Lambda \gamma_{23} \Lambda) \Xi_{(2)}^{-1} (\Gamma_{13}-\Lambda \gamma_{23}^\dagger \Lambda \Xi_{(1)}^{-1} \Gamma_{34}^\dagger) &\approx_{r',r'' \to\infty} DI_2
\end{align}

i.e., we will get the exact same result even though we had done the ordering differently - the order of our stitching does not matter, and the order elected in the algorithm is purely for bookkeeping.

\section{Example - Interacting two combs}
In this example, we use the link product to stitch two copies of the process evaluated in Appendix \ref{appendix:bs-bs-process-join} as per Fig.~\ref{fig:interaction-example-combs}: this evaluates the interaction between two combs. In this case, both $\rho_E$ and $\rho_L$ are taken to be the vacuum state. For the first copy of this process, which we label $\mathcal K$, we restate our previous result for the covariance matrix this process (corresponding to qumode order $(o_1^{\mathcal K}, i_1^{\mathcal K}, o_2^{\mathcal K}, i_2^{\mathcal K})$):

\begin{align}
    \Gamma_{\Upsilon_{\mathcal K}} &= 
\begin{pmatrix}
\Gamma_{11} & \Gamma_{12} & \Gamma_{13} & \mathbf 0 \\
\Gamma_{12}^{\dagger} & \Gamma_{22} & \Gamma_{23} & \mathbf 0 \\
\Gamma_{13}^\dagger & \Gamma_{23}^\dagger & \Gamma_{33} & \Gamma_{34} \\
\mathbf 0 & \mathbf 0 & \Gamma_{34}^\dagger & \Gamma_{44}
\end{pmatrix} \\
&= \begin{pmatrix}
(\cosh r \cos^2\theta + \sin^2\theta) I_2 & \sinh r \cos\theta \Lambda & \sin^2\theta \cos\theta (1-\cosh r) I_2 &\mathbf 0  \\
\sinh r \cos\theta \Lambda & \cosh r I_2 & -\sinh r \sin^2 \theta \Lambda & \mathbf 0 \\
\sin^2\theta \cos\theta (1-\cosh r) I_2 & -\sinh r  \sin^2 \theta \Lambda & \cosh r' \cos^2 \theta I_2 + \sin^2 \theta (\cosh r \sin^2 \theta + \cos^2\theta) I_2 & \sinh r' \cos\theta \Lambda \\
\mathbf 0 & \mathbf 0 & \sinh r' \cos\theta \Lambda & \cosh r' I_2
\end{pmatrix}
\end{align}

For the second circuit fragment $\mathcal L$, we can restate the same result but with replaced variables - we change $\theta\to\alpha$ and $r\to s$. The covariance matrix $(o_1^{\mathcal L}, i_1^{\mathcal L}, o_2^{\mathcal L}, i_2^{\mathcal L})$:

\begin{align}
    \Gamma_{\Upsilon_{\mathcal L}} &= 
\begin{pmatrix}
\hat\Gamma_{11} & \hat\Gamma_{12} & \hat\Gamma_{13} & \mathbf 0 \\
\hat\Gamma_{12}^{\dagger} & \hat\Gamma_{22} & \hat\Gamma_{23} & \mathbf 0 \\
\hat\Gamma_{13}^\dagger & \hat\Gamma_{23}^\dagger & \hat\Gamma_{33} & \hat\Gamma_{34} \\
\mathbf 0 & \mathbf 0 & \hat\Gamma_{34}^\dagger & \hat\Gamma_{44}
\end{pmatrix} \\
&= \begin{pmatrix}
(\cosh s \cos^2\alpha + \sin^2\alpha) I_2 & \sinh s \cos\alpha \Lambda & \sin^2\alpha \cos\alpha (1-\cosh s) I_2 & \mathbf 0 \\
\sinh s \cos\alpha \Lambda & \cosh s I_2 & -\sinh s \sin^2 \alpha \Lambda & \mathbf 0 \\
\sin^2\alpha \cos\alpha (1-\cosh s) I_2 & -\sinh s  \sin^2 \alpha \Lambda & \cosh s' \cos^2 \alpha I_2 + \sin^2 \alpha (\cosh s \sin^2 \alpha + \cos^2\alpha) I_2 & \sinh s' \cos\alpha \Lambda \\
\mathbf 0 & \mathbf 0 & \sinh s' \cos\alpha \Lambda & \cosh s' I_2
\end{pmatrix}
\end{align}

\label{appendix: two-combs-interact}
In this $J$-stitching, we have three mode-stitches with $J=\set{(o_1^{\mathcal K}, i_1^{\mathcal L}), (o_1^{\mathcal L}, i_2^{\mathcal K}), (o_2^{\mathcal K}, i_2^{\mathcal L})} = \set{j_1, j_2, j_3}$. The stitch $j_1$ will correspond with taking the limit $s\to\infty$, $j_2$ with $r'\to\infty$ and $j_3$ with $s'\to\infty$. We directly apply the algorithm for three loops (one for each of the 3 stitches) and obtain

\begin{align}
    \Gamma_{\mathcal C} = \lim_{s,r',s'\to\infty} \displaystyle \left[\begin{matrix} \Gamma_{\mathcal C}^{(11)} & \Gamma_{\mathcal C}^{(12)} \\ \Gamma_{\mathcal C}^{(12)T} & \Gamma_{\mathcal C}^{(22)} \end{matrix}\right]
\end{align}

corresponding to mode order ($i_1^{\mathcal K}, o_2^{\mathcal L}$) with $2\times2$ block matrices of the form

\begin{align}
    \Gamma_{\mathcal C}^{(11)} &= -\Gamma_{12}^\dagger\Lambda \Xi_{(1)}^{-1}\Lambda\Gamma_{12} - \Gamma_{12}^\dagger\Lambda \Xi_{(1)}^{-1} \hat\Gamma_{12}^\dagger \Xi_{(2)}^{-1} \hat\Gamma_{12} \Xi_{(1)}^{-1}\Lambda\Gamma_{12} + \Gamma_{22} \nonumber \\
    &- \left(\Gamma_{12}^\dagger\Lambda \Xi_{(1)}^{-1} \Lambda \Gamma_{13} \Lambda -\Gamma_{12}^\dagger\Lambda \Xi_{(1)}^{-1} \hat\Gamma_{12}^\dagger \Xi_{(2)}^{-1} \left(\Lambda \Gamma_{34}^\dagger \Lambda + \hat\Gamma_{12} \Xi_{(1)}^{-1} \Lambda \Gamma_{13} \Lambda \right) - \Gamma_{23} \Lambda\right) \Xi_{(3)}^{-1} \nonumber \\
    &\left(\Lambda \Gamma_{13}^\dagger \Lambda \Xi_{(1)}^{-1} \Lambda\Gamma_{12} - \Lambda \Gamma_{23}^\dagger - \left(-\Lambda \Gamma_{13}^\dagger \Lambda \Xi_{(1)}^{-1} \hat\Gamma_{12}^\dagger + \Lambda \Gamma_{34} \Lambda\right) \Xi_{(2)}^{-1} \hat\Gamma_{12} \Xi_{(1)}^{-1} \Lambda\Gamma_{12}\right) \\
    \Gamma_{\mathcal C}^{(22)} &= - \hat\Gamma_{23}^\dagger \Xi_{(1)}^{-1} \hat\Gamma_{23} + \hat\Gamma_{33} - \left(\hat\Gamma_{13}^\dagger - \hat\Gamma_{23}^\dagger \Xi_{(1)}^{-1} \hat\Gamma_{12}^\dagger\right) \Xi_{(2)}^{-1} \left(- \hat\Gamma_{12} \Xi_{(1)}^{-1} \hat\Gamma_{23} + \hat\Gamma_{13}\right) \nonumber \\
    &- \left(-\hat\Gamma_{23}^\dagger \Xi_{(1)}^{-1} \Lambda \Gamma_{13} \Lambda + \hat\Gamma_{34}^\dagger - \left(\hat\Gamma_{13}^\dagger - \hat\Gamma_{23}^\dagger \Xi_{(1)}^{-1} \hat\Gamma_{12}^\dagger\right) \Xi_{(2)}^{-1} \left(\Lambda \Gamma_{34}^\dagger \Lambda - \hat\Gamma_{12} \Xi_{(1)}^{-1} \Lambda \Gamma_{13} \Lambda\right)\right) \Xi_{(3)}^{-1} \nonumber \\
    &\left(-\Lambda \Gamma_{13}^\dagger \Lambda \Xi_{(1)}^{-1} \hat\Gamma_{23} + \hat\Gamma_{34} - \left(-\Lambda \Gamma_{13}^\dagger \Lambda \Xi_{(1)}^{-1} \hat\Gamma_{12}^\dagger + \Lambda \Gamma_{34} \Lambda\right) \Xi_{(2)}^{-1} \left(- \hat\Gamma_{12} \Xi_{(1)}^{-1} \hat\Gamma_{23} + \hat\Gamma_{13}\right)\right) \\
    \Gamma_{\mathcal C}^{(12)} &= -\Gamma_{12}^\dagger\Lambda \Xi_{(1)}^{-1} \hat\Gamma_{12}^\dagger \Xi_{(2)}^{-1} \left(- \hat\Gamma_{12} \Xi_{(1)}^{-1} \hat\Gamma_{23} + \hat\Gamma_{13}\right) + \Gamma_{12}^\dagger\Lambda \Xi_{(1)}^{-1} \hat\Gamma_{23} \nonumber \\ 
    &- \left(\Gamma_{12}^\dagger\Lambda \Xi_{(1)}^{-1} \Lambda \Gamma_{13} \Lambda -\Gamma_{12}^\dagger\Lambda \Xi_{(1)}^{-1} \hat\Gamma_{12}^\dagger \Xi_{(2)}^{-1} \left(\Lambda \Gamma_{34}^\dagger \Lambda - \hat\Gamma_{12} \Xi_{(1)}^{-1} \Lambda \Gamma_{13} \Lambda\right) - \Gamma_{23} \Lambda\right) \Xi_{(3)}^{-1} \nonumber \\
    &\left(-\Lambda \Gamma_{13}^\dagger \Lambda \Xi_{(1)}^{-1} \hat\Gamma_{23} + \hat\Gamma_{34} - \left(-\Lambda \Gamma_{13}^\dagger \Lambda \Xi_{(1)}^{-1} \hat\Gamma_{12}^\dagger + \Lambda \Gamma_{34} \Lambda\right) \Xi_{(2)}^{-1} \left(- \hat\Gamma_{12} \Xi_{(1)}^{-1} \hat\Gamma_{23} + \hat\Gamma_{13}\right)\right),
\end{align}

with the inverse terms of

\begin{align}
\Xi_{(1)}^{-1} &= \left(\Lambda\Gamma_{11}\Lambda + \hat\Gamma_{22}\right)^{-1}  \\
&\approx \frac{1}{\cosh s} \left[I_2 - \frac{\cosh r \cos^2\theta + \sin^2 \theta}{\cosh s} I_2\right] \\
\Xi_{(2)}^{-1} &= \left(\Lambda\Gamma_{44}\Lambda + \hat\Gamma_{11} - \hat\Gamma_{12} \Xi_{(1)}^{-1} \hat\Gamma_{12}^\dagger\right)^{-1} \\
& \approx \frac{1}{\cosh r'} \left[I_2 - \frac{\sin^2\alpha + \cos^2\alpha (\cosh r \cos^2\theta + \sin^2\theta)}{\cosh r'} I_2 \right] \\
\Xi_{(3)}^{-1} &= (- \Lambda \Gamma_{13}^\dagger \Lambda \Xi_{(1)}^{-1} \Lambda \Gamma_{13} \Lambda + \Lambda \Gamma_{13}^\dagger \Lambda \Xi_{(1)}^{-1} \hat\Gamma_{12}^\dagger \Xi_{(2)}^{-1} \Lambda \Gamma_{34}^\dagger \Lambda - \Lambda \Gamma_{13}^\dagger \Lambda \Xi_{(1)}^{-1} \hat\Gamma_{12}^\dagger \Xi_{(2)}^{-1} \hat\Gamma_{12} \Xi_{(1)}^{-1} \Lambda \Gamma_{13} \Lambda + \Lambda \Gamma_{33} \Lambda \nonumber \\ 
&- \Lambda \Gamma_{34} \Lambda \Xi_{(2)}^{-1} \Lambda \Gamma_{34}^\dagger \Lambda + \Lambda \Gamma_{34} \Lambda \Xi_{(2)}^{-1} \hat\Gamma_{12} \Xi_{(1)}^{-1} \Lambda \Gamma_{13} \Lambda + \hat\Gamma_{44})^{-1} \\
&= \left(\Lambda\Gamma_{33}\Lambda + \hat\Gamma_{44} + 2\Lambda \Gamma_{13}^\dagger \Lambda \Xi_{(1)}^{-1} \hat\Gamma_{12}^\dagger \Xi_{(2)}^{-1} \Lambda \Gamma_{34}^\dagger \Lambda - \Lambda \Gamma_{34} \Lambda \Xi_{(2)}^{-1} \Lambda \Gamma_{34}^\dagger \Lambda \right)^{-1} \\
&\approx \frac{1}{\cosh s'} \left[I_2 - \frac{s^2_\theta(u_r s^2_\theta + c^2_\theta)+2c_\alpha s^2_\theta c^2_\theta (1- u_r)+ c^2_\theta(s^2_\alpha + c^2_\alpha(u_r c^2_\theta + s^2_\theta))}{\cosh s'} I_2\right],
\end{align}

where we have used the shorthand of $c_\theta = \cos\theta, s_\theta =\sin\theta, u_r = \cosh r, v_r=\sinh r$ in the expression for $\Xi_{(3)}^{-1}$.

While the block matrices look unwieldy, calculations can be greatly simplified by considering and keeping only the terms that will survive taking the limits of $s,r', s' \to \infty$. This can be done by considering the variable dependencies present in $\Gamma_{kl}$ and $\hat\Gamma_{kl}$ in relation to the variable(s) that is (are) taken $\to\infty$ in the corresponding $\Xi_{(j)}^{-1}$ they are multiplied with: if the variables {\it do not} overlap, then the term will not survive. For example, in $\Gamma_{12}^\dagger \Lambda \Xi_{(1)}^{-1} \Lambda \Gamma_{12}$, $\Xi_{(1)}^{-1}$ corresponds to the first join, $j_1$, and thus corresponds to $s\to\infty$. However, $\Gamma_{12}$ only contains an $r$ dependence, and will thus not survive after taking the limit of $s\to\infty$.

From this, we simplify all the above expressions to

\begin{align}
    \Gamma_{\mathcal C}^{(11)} &= \Gamma_{22} \\
    \Gamma_{\mathcal C}^{(22)} &= \hat\Gamma_{23}^\dagger \Xi_{(1)}^{-1} \hat\Gamma_{23} + \hat \Gamma_{33} +2 [\hat \Gamma_{23}^\dagger \Xi_{(1)}^{-1} \Lambda \Gamma_{13} \Lambda - \frac{\hat \Gamma_{34}}{2}+ (\hat \Gamma_{13}-\hat\Gamma_{23}^\dagger \Xi_{(1)}^{-1}\hat\Gamma_{12}^\dagger)\Xi_{(2)}^{-1}\Lambda \Gamma_{34}^\dagger \Lambda] \Xi_{(3)}^{-1} \hat \Gamma_{34} \\
    \Gamma_{\mathcal C}^{(12)} &= \Gamma_{12}^\dagger\Lambda \Xi_{(1)}^{-1} \hat\Gamma_{23} + \left( \Gamma_{12}^\dagger\Lambda \Xi_{(1)}^{-1} \hat\Gamma_{12}^\dagger \Xi_{(2)}^{-1} (-\Lambda \Gamma_{34}^\dagger \Lambda) + \Gamma_{23} \Lambda\right) \Xi_{(3)}^{-1} \hat\Gamma_{34},
\end{align}

Substituting all the terms and taking the appropriate limits for $s,r', s' \to \infty$, we should obtain

\begin{align}
\Gamma_{\mathcal C}^{(11)} &= \cosh (r) I_2 \\
    \Gamma_{\mathcal C}^{(22)} &= [s_\alpha^4c_\theta^2+c_\alpha^2 s_\theta^4 + c_\alpha^4 c_\theta^4 + 2(c_\alpha s_\alpha^2 c_\theta s_\theta^2 - c_\alpha^3 c_\theta^2 s_\theta^2 - c_\alpha ^2 s_\alpha^2 c_\theta^3)] u_r I_2 \nonumber \\
    &+ c_\alpha^2 (c_\alpha^2 c_\theta^2 s_\theta^2 + s_\alpha^2 + c_\theta^2 s_\theta^2 + s_\alpha^2 c_\theta^2 +2 s_\alpha^2 c_\theta -2 s_\alpha^2 c_\theta s_\theta^2 + 2 c_\alpha s_\theta^2 c_\theta^2) I_2 \nonumber \\
    &+s_\alpha^2 (s_\alpha^2 s_\theta^2 - 2 c_\alpha c_\theta s_\theta^2) I_2 \\
    \Gamma_{\mathcal C}^{(12)} &= \left(-\sin^{2}{\left(\alpha \right)} \cos{\left(\theta \right)} + \cos^{2}{\left(\alpha \right)} \cos^2{\left(\theta \right)} - \sin^{2}{\left(\theta \right)} \cos{\left(\alpha \right)}\right) \sinh{\left(r \right)} \Lambda
\end{align}

where we have used the shorthand of $c_\theta = \cos\theta, s_\theta =\sin\theta, u_r = \cosh r, v_r=\sinh r$. For uniformity with how we have expressed all our other results in the main text, for its presentation in the main text we re-arranged the mode order to be $(o_2^{\mathcal L}, i_1^{\mathcal K})$ such that the covariance matrix has the form

\begin{align}
    \Gamma_{\mathcal C} = \lim_{s,r',s'\to\infty} \displaystyle \left[\begin{matrix} \Gamma_{\mathcal C}^{(22)} & \Gamma_{\mathcal C}^{(12)T} \\ \Gamma_{\mathcal C}^{(12)} & \Gamma_{\mathcal C}^{(11)} \end{matrix}\right].
\end{align}

From this general form, we can consider some simplified cases as a quick sanity check. 

For our first case, if we take $\alpha=0$, i.e., $\mathcal A'$ is the identity channel, this is equivalent to the case of $(\mathcal A\otimes I_2)\circ (\mathcal A\otimes I_2)[\rho_E \otimes \Phi(r)]$ where $\rho_E$ is the vacuum state:

\begin{align}
    \Gamma_{\mathcal C} = \begin{pmatrix}
        [Y^2\cosh r+4\sin^2\theta\cos^2\theta]I_2 & Y\sinh r\Lambda \\
        Y\sinh r\Lambda & \cosh rI_2
    \end{pmatrix}
\end{align}

where $Y=\cos^2 \theta- \sin^2\theta$. This will match precisely with the resulting covariance matrix from calculating  $(\mathcal A\otimes I_2)\circ (\mathcal A\otimes I_2)[\rho_E \otimes \Phi(r)]$ directly. This simplification is therefore what we would expect, and is the same simplified result as the previous example.

For our second case, we can let $\mathcal A' = \mathcal A$. The resulting covariance matrix is simplified significantly, and we obtain:

\begin{align}
    \Gamma_{\mathcal C}^{(11)} &= u_r I_2 \\
    \Gamma_{\mathcal C}^{(22)} &= [(4 c_\theta^2 s_\theta^4 + c_\theta^8 - 4c_\theta^5 s_\theta^2)u_r + c_\theta^2s_\theta^2(1+c_\theta^4 + 4 c_\theta^3 + 2c_\theta^2 ) + s_\theta^4(s_\theta^2-2c_\theta^2) ] I_2  \\
    \Gamma_{\mathcal C}^{(12)} &= (c_\theta^4-2s_\theta^2 c_\theta) v_r \Lambda
\end{align}

which is precisely the covariance matrix if we had simply made a direct calculation applying $\mathcal A$ four times to obtain the covariance matrix. In such a case, simplifying further as per the main text, we see that this corresponds to an effective attenuating channel on the system.

\end{appendix}
\end{document}